%% file: writeup.tex
\documentclass[11pt]{article}

\usepackage{sn-preamble}
\usepackage{enumitem}
\usepackage{url}
\usepackage{booktabs}
\usepackage{wasysym}
\usepackage{fancyhdr}
\usepackage{microtype}

\newcommand{\sparsity}{\mathrm{sparsity}}

\newcommand{\Bin}{\mathrm{Bin}}
\newcommand{\dtv}{d_\mathrm{TV}}
\newcommand{\cum}{\kappa}
\newcommand{\ion}{\boldsymbol{\eta}}

\newcommand{\coeff}{\mathrm{coeff}}
\newcommand{\Unif}{\mathrm{Unif}}
\newcommand{\MSG}{\mathrm{Max\text{-}Sparsity\text{-}Gap}}
\newcommand{\newLs}{T}
\newcommand{\SSparse}{$s\mathsf{\text{-}sparse}$}
\newcommand{\FFLS}{$\mathsf{far\text{-}from\text{-}}\newLs\text{-}\mathsf{sparse}$}
\newcommand{\CENP}{{\tt Construct-Net-of-Polynomials}}
\newcommand{\CEND}{{\tt Construct-Net-of-Random-Variables}}
\newcommand{\TS}{{\tt Test-Sparsity}}
\newcommand{\EW}{{\tt Estimate-Wasserstein}}
\newcommand{\RV}{\mathrm{RV}}
\newcommand{\Mom}{\mathrm{Mom}}

\newcommand{\zetacoeff}{\zeta_\coeff}
\newcommand{\zetaMom}{\zeta_\Mom}

\newcommand{\ignore}[1]{}

\newcommand{\EMWN}{$\tt{Estimate\text{-}Moments\text{-}With\text{-}Noise}$}
\newcommand{\upl}[1][]{%
  \if\relax\detokenize{#1}\relax
    [\ell,d\ell]%
  \else
    [#1\ell,(#1\ell)d)]%
  \fi
}
\newcommand{\maxl}[2]{[#1,#2]}

\title{Testing noisy low-degree polynomials for sparsity}

\author{
\begin{tabular}{ccc}
Yiqiao Bao & Anindya De & Shivam Nadimpalli \\
University of Pennsylvania & University of Pennsylvania & MIT \\
\url{yiqiaob@seas.upenn.edu} & \url{anindyad@cis.upenn.edu} & \url{shivamn@mit.edu} \\
\\
\multicolumn{3}{c}{
\begin{tabular}{cc}
Rocco A.~Servedio & Nathan White \\
Columbia University & University of Pennsylvania \\
\url{rocco@cs.columbia.edu} & \url{nathanlw@cis.upenn.edu}
\end{tabular}
}
\end{tabular}
\vspace{0.5em}
}

\begin{document}

\maketitle

\pagenumbering{gobble}

\begin{abstract}
We consider the problem of testing whether an unknown low-degree polynomial $p$ over $\R^n$ is sparse versus far from sparse,  given access to {noisy} evaluations of the polynomial $p$ at \emph{randomly chosen points}.  
This is a natural property-testing version of various well-studied problems about learning low-degree sparse polynomials in the presence of noise, and is a generalization of the work of Chen, De, and Servedio~\cite{CDS20stoc}, on testing noisy \emph{linear} functions for sparsity, to the more challenging setting of low-degree polynomials.

Our main result gives a \emph{precise characterization} of when sparsity testing for low-degree polynomials can be carried out with constant sample complexity independent of dimension, 
along with a constant-sample algorithm for this problem in the parameter regime where this is possible.
In more detail, for any mean-zero variance-one finitely supported distribution $\bX$ over the reals, any degree parameter $d$, and any sparsity parameters $s$ and $\newLs \geq s$,
we define a computable function $\MSG_{\bX,d}(\cdot)$, and:

\begin{itemize}
    \item For $\newLs \geq \MSG_{\bX,d}(s)$ we give an $O_{s,\bX,d}(1)$-sample algorithm for the problem of distinguishing whether a multilinear degree-$d$ polynomial over $\R^n$ is $s$-sparse versus $\eps$-far from $\newLs$-sparse, given independent labeled examples $(\bx,p(\bx)+\mathrm{noise})_{\bx \sim \bX^{\otimes n}}$.  (Crucially, this sample complexity is \emph{completely independent} of the ambient dimension $n$.) On the other hand,

    \item For $\newLs \leq \MSG_{\bX,d}(s) - 1$, we show that even in the absence of noise, any algorithm for distinguishing whether a multilinear degree-$d$ polynomial is $s$-sparse versus $\eps$-far from $\newLs$-sparse, given independent labeled examples $(\bx,p(\bx))_{\bx \sim \bX^{\otimes n}}$, must use $\Omega_{\bX,d,s}(\log n)$ examples.
\end{itemize}

Our techniques employ a generalization of the results of Dinur, Friedgut, Kindler, and O'Donnell~\cite{DFKO:journal} on the Fourier tails of bounded functions over $\bits^n$ to a broad range of finitely supported distributions, which may be of independent interest.

\end{abstract}

\newpage
\setcounter{tocdepth}{2}
 {\small{\tableofcontents}
 }

 \newpage

\pagenumbering{arabic}

\input{sections/intro}
\input{sections/prelims}

\input{sections/moment_cum_est}

\input{sections/DFKO-structural}

\input{sections/DFKO-algorithm}
\input{sections/MSG-well-defined}

\input{sections/alg-to-compute-MSG}

\input{sections/lower-bound-sharp}

\input{sections/upper-bound-sharp}

\section*{Acknowledgements}

Y.B.~is supported by NSF CCF 2045128. 
A.D.~is supported by NSF CCF 2045128. 
R.A.S.~is supported by NSF CCF-2211238 and CCF-2106429. 
N.W.~is supported by the Department of Defense (DoD) through the National Defense Science and Engineering Graduate (NDSEG) Fellowship Program. 

\bibliography{allrefs}
\bibliographystyle{alpha}

\appendix
\input{sections/appendix-boundsonMSG}

\input{sections/appendix-estimating-moments}

\end{document}

%% file: sections/intro.tex

\section{Introduction}
\label{sec:intro}

The field of property testing investigates algorithms that, given query access to an unknown function $f : \calX \to \calY$, determine whether (i) $f$ belongs to some known class of functions $\mathcal{C}$ (i.e.,~$f$ has the property $\mathcal{C}$), versus (ii) $f$ is ``far'' from any function $g \in \mathcal{C}$ (i.e.,~$f$ is ``far'' from having property $\mathcal{C}$).  
Starting with the seminal works of \cite{BLR93, blum1995designing, RS96}, several fundamental properties $\mathcal{C}$ have been intensively studied from this perspective, including
low-degree polynomials~\cite{BLR93,raz1997sub,AKKLR}, juntas~\cite{FKRSS03, Blais08}, monotonicity~\cite{GGLRS, FLNRRS, CWX17stoc}, and geometric classes such as halfspaces~\cite{MORS:10, Harms19,DMN21}, to name a few. 
From the point of view of complexity theory, this framework has been highly influential, most notably through its applications to probabilistically checkable proofs (PCPs) and hardness of approximation.

Despite its success, the classical framework rests on two assumptions that limit its relevance to contemporary data analysis and machine learning. 
The first is that most function classes $\mathcal{C}$ traditionally studied in theoretical computer science are Boolean-valued. 
This choice is natural from the viewpoint of complexity theory but is less suited to statistical and machine learning settings, where the outputs---or ``response variables''---are typically real-valued.
It is therefore natural to consider classes $\mathcal{C}$ of real-valued functions.  

An arguably bigger chasm is the assumption on how the testing algorithm accesses the function $f$. 
Motivated by applications to PCPs, hardness of approximation, and coding theory, the classical framework assumes that the tester has \emph{query access} to $f$. 
In contrast, settings arising in statistics and machine learning typically provide much weaker access---namely, \emph{random samples} labeled according to $f$. 
Property testing in this sample-based model is significantly more challenging from an algorithmic standpoint.  
We refer the reader to \cite{CFSS17, KR00, goldreich2016sample}  for some results on sample-based testing of Boolean functions. 
A recurring theme in this line of work is that the guarantees achievable in the sample-based model are often much weaker than those in the query model.
As an example, for the celebrated problem of linearity testing over $\mathsf{GF}[2]$, there are query algorithms that use only $\Theta(1)$ queries~\cite{BLR93}, but the same problem requires $\Theta(n)$ samples in the sample-based testing model~\cite{goldreich2016sample}. 

\paragraph{Background: Testing sparse linear functions from samples.}
The most relevant prior work for the present paper is the work of Chen, De and Servedio~\cite{CDS20stoc}. 
In their model, the algorithm gets independent 
labeled examples of the form 
\[
    (\bx,\by = w \cdot \bx + \boldeta)\,,
\]
where each $\bx$ is drawn from some distribution ${\cal D}$ over $\R^n$, $w$ is an unknown vector, and $\boldeta$ is a  ``noise distribution" supported over $\R$. 
 The goal in this model is to distinguish between the following two cases: (i) $w$ is $s$-sparse, versus (ii) $w$ is at least $\epsilon$-far from every $s$-sparse vector (in the sense that at most $1-\eps$ of the Euclidean norm of $w$ comes from its $s$ largest coefficients).  In other words, \cite{CDS20stoc} studies the problem of ``testing noisy linear functions for sparsity."
 
 Roughly speaking, \cite{CDS20stoc} show as their main result that if (i) ${\cal D}$ is an i.i.d.~distribution ${\cal D}=\bX^{\otimes n}$ (where $\bX$ is known to the testing algorithm), (ii) the noise distribution $\boldeta$ is known to the testing algorithm (and independent of $\bx$), and (iii) ${\cal D}$ is not a Gaussian distribution, then there is an algorithm for this problem whose sample complexity depends only on the sparsity parameter $s$, the error parameter $\epsilon$, and the real-valued distribution $\bX$. Crucially, the sample complexity is \emph{dimension independent}, i.e., does not depend on the ambient dimension $n$. 
 Furthermore, \cite{CDS20stoc} established that if any of the three conditions above are relaxed, the sample complexity of the problem becomes $\widetilde{\Omega}(\log n)$.  
 
\paragraph{This work: Testing sparse low-degree polynomials from samples.}

We study a natural and challenging extension of the problem considered in \cite{CDS20stoc}. 
In particular, we investigate the problem of ``testing \emph{low-degree polynomials} for sparsity"---thus, the results of \cite{CDS20stoc} correspond to the special case of degree-1 polynomials. 
In more detail, in our model the algorithm gets independent labeled examples of the form 
\[
    (\bx,\by = p(\bx) + \boldeta)\,,
\]
where $p$ is now promised to be a degree $d$ polynomial. 
The goal is to distinguish between the two cases that (i) $p$ is $s$-sparse, versus (ii) $p$ is ``far from $s$-sparse'' (we make this notion precise shortly).  
As in \cite{CDS20stoc}, and in keeping with the usual ``gold standard'' for efficient property testing algorithms, our main goal is to determine conditions under which this can be achieved with a sample complexity that is \emph{completely independent} of $n$. 

Since this problem is more challenging than the one in \cite{CDS20stoc}, we will (at least) require the same assumptions as \cite{CDS20stoc}. 
In particular, we will assume that (i)  ${\cal D}$ is an i.i.d. product distribution $\calD=\bX^{\otimes n}$; (ii)  the distribution of $\boldeta$ is known (and independent of $\bx$); and (iii) ${\cal D}$ is not Gaussian. 
However, these conditions alone turn out to be insufficient to obtain {\em dimension-independent} sample complexity (i.e., independent of $n$) for the problem of testing sparsity of low-degree polynomials. 
To see this, and to motivate the additional assumptions that we will introduce, we consider two illustrative examples: 

\medskip

\noindent {\bf Example \#1:} Consider two polynomials $p_1(x) := 1$ and $p_2(x) := (x_1^2 + \ldots + x_n^2)/n$. 
Although $p_1$ and $p_2$ are distinct as polynomials, we have $p_1(x) = p_2(x)$ for every $x \in \{-1,1\}^n$. 
Since $p_1$ and $p_2$ have very different sparsities, it follows that---even with infinitely many noiseless samples of the form $(\bx, p(\bx))$ where $\bx$ is supported on $\bn$---one cannot distinguish whether $p$ is $1$-sparse versus $n$-sparse.
To rule out such degeneracies, we will restrict our attention to {\em multilinear} polynomials---i.e., polynomials whose degree in each variable at most one. 
The multilinearity restriction ensures that any polynomial $p$ (and consequently, its sparsity) is completely determined by its values on $\{-1,1\}^n$.

\medskip

\noindent {\bf Example \#2:} 
Even with the additional restriction of multilinearity, achieving the kind of sharp testing guarantees ($s$-sparse versus $\eps$-far from $s$-sparse) obtained in \cite{CDS20stoc} for linear functions is no longer possible for degree-$d$ polynomials when $d>1$. 
More precisely, our results imply that for the uniform distribution over $\{-1,1\}^n$, the task of distinguishing
\begin{itemize}
    \item degree-$2$ multilinear polynomials which are $1$-sparse, versus
    \item degree-$2$ multilinear polynomials which are $c$-far from $3$-sparse (for some absolute constant $c>0$),
\end{itemize}
requires $\Omega(\log n)$ samples. 
(Looking ahead, this is an immediate consequence of \Cref{thm:sharp-lower} and the lower bound on ``$\MSG_{\Unif(\bits),2}(1)$'' given in \Cref{obs:upper-lower-uniform}.)
Thus, even for a simple background distribution ${\cal D}$ such as the uniform distribution over $\bits^n$, the problem of distinguishing $s$-sparse degree-2 multilinear polynomials from $\epsilon$-far-from-$s$-sparse degree-2 multilinear polynomials cannot be solved with  dimension-independent sample complexity. 

\medskip

Thus, we turn to a natural ``gapped" version of the sparsity testing problem, which is as follows: 
Given a background distribution ${\cal D} = \bX^{\otimes n}$ and the noise distribution $\boldeta$, the algorithm receives noisy labeled samples of the form $(\bx, p(\bx) + \boldeta)$ and must distinguish between (i) $p$ being $s$-sparse, versus (ii) $p$ being $\epsilon$-far from $\newLs$ sparse, for some value $\newLs>s$. 
Our focus will be on the setting where $\bX$ has finite support, and our goal is to obtain sharp bounds on the ``gap value" $\newLs$ such that a dimension-independent sample complexity \emph{is} achievable for this testing problem. 

Before proceeding, we give a brief heuristic reason for why the degree-$d$ case is substantially more challenging than the linear (i.e., degree-$1$) case considered by~\cite{CDS20stoc}.
Any degree-$d$ multilinear polynomial over $\R^n$ can be viewed as a linear form over an expanded feature space of ${n \choose \leq d}$ many ``meta-features'' corresponding to all degree-at-most-$d$ multilinear monomials.  However, when $\R^n$ is endowed with the i.i.d.~product measure $\bX^{\otimes n}$, the induced measure on this expanded feature space fails to satisfy the key conditions shown to be necessary in \cite{CDS20stoc}: not only is it not i.i.d., it is not even 3-wise independent.  
Consequently, there is no simple reduction to the linear case handled in previous work.

\paragraph{Other related work.} Several prior works in the area of property testing have studied the testability of  function classes admitting sparse representations. This includes a long line of work on junta testing~\cite{FKR+:04, Blais08, chen2018settling, DMN19},  testing sparse polynomials over finite fields~\cite{DLM+:08, grigorescu2010local} and Gaussian space~\cite{ABJS2025}, as well as the testability of function classes admitting various kinds of (exact or approximate) concise representations~\cite{DLM+:07, GOS+:09}. 
However, as is standard in property testing, the access model considered in all these works is via queries, whereas in the current paper the form of access we consider is significantly weaker, namely, random samples. 

Sparse low-degree polynomials have also been studied from the point of view of machine learning. In particular, in the same model as the current paper, \cite{andoni2014learning} designed an efficient algorithm to recover an underlying degree-$d$ $s$-sparse polynomial $p$ for two different choices of background distributions ($\bX = N(0,1)$ and the uniform distribution on $[-1,1]$). 
The key contribution of \cite{andoni2014learning} was an algorithm whose running time scales as a fixed polynomial in $n$ (as opposed to a polynomial whose exponent on $n$ depends on $d$ or the sparsity parameter $s$). 
In contrast with that work, our focus is on testing rather than learning, and our goal is to achieve sample complexity \emph{completely independent} of $n$.

A bit farther afield, sparse low-degree polynomials are also connected to other tasks in  machine learning; for e.g., \cite{bertsimas2020sparse} studied  the problem of learning such polynomials in the context of sparse hierarchical regression. 
Connections have also been made between the problem of learning sparse polynomials and the problems of graph and hypergraph sketching~\cite{kocaoglu2014sparse, li2015active}.









\subsection{Our results and techniques}
Let us define the notion of \emph{distance to $s$-sparsity} that we will use:

\begin{definition} \label{def:coeff}
Given a degree-$d$ multilinear polynomial
\[
p(x_1,\dots,x_n) = \sum_{S \in {[n] \choose \leq d}} \widehat{p}(S) \prod_{i \in S} x_i,
\]
we write $\|p\|_\coeff$ for
\[
\pbra{\sum_{S} \widehat{p}(S)^2}^{1/2},
\]
the square root of the sum of squares of $p$'s coefficients.
Given a nonzero degree-$d$ multilinear polynomial $p$, we define its \emph{distance to $s$-sparsity} to be
\begin{equation}~\label{def:distance}
\dist(p,s\text{-sparse}):= \min_{q \text{~is an $s$-sparse degree-$d$ multilinear polynomial}} {\frac {\|p - q\|_\coeff}{\|p\|_\coeff}};
\end{equation}
this can be viewed as the fraction of the coefficient-norm of $p$ that comes from the coordinates that are not among the $s$ largest-magnitude ones (see \Cref{lem:first-s-sum-of-squares} for a precise statement).
\end{definition}

We remark that for the case of $d=1$ (linear functions), \Cref{def:coeff} coincides with the notion of ``distance to $s$-sparsity'' used in \cite{CDS20stoc}.

\subsubsection{First result:  A coarse upper bound} 

As described earlier in Example \#2, unlike the case for linear polynomials, once the degree $d$ of the unknown polynomial is greater than $1$ we can no longer test whether the polynomial $p$ is $s$-sparse versus $\epsilon$-far from $s$-sparse with a number of samples that is independent of $n$. Instead, we seek to design algorithms with constant sample complexity to distinguish between the cases that (i) $p$ is $s$-sparse, versus (ii) $p$ is $\epsilon$-far from $T$-sparse for some $T>s$. As alluded to earlier, we seek to understand the exact
threshold $T$ for which such a constant sample complexity tester is possible. But en route to this sharp result, both as a conceptual first step (but also a key technical ingredient in our sharp result), we will first get a {\em coarse upper bound}. 

Before going further, let us define the precise algorithmic problem we are interested in: 

\begin{definition} [Sparsity testing problem with gap parameter $\newLs$]
\label{def:problem}
Fix a finitely supported random variable $\bX$ over $\R$, a degree bound $d$, a noise distribution $\boldeta$ over $\R$, a sparsity parameter $s \in \N$, a gap value $\newLs \in \N$, and a real error parameter $\eps > 0$.

The \emph{$(\bX,d,\boldeta,s,\newLs,\eps)$ polynomial sparsity testing problem} is the following:
An algorithm is given as input access to i.i.d.~samples $(\bx \sim \bX^{\otimes n}, \by=p(\bx) + \boldeta)$, where $p$ is promised to be a multilinear degree-$d$ polynomial with $1/K \leq \|p\|_{\coeff} \leq K$. The task for the algorithm is to satisfy the following criteria:

\begin{itemize}

\item If $\sparsity(p) \leq s$, then with probability at least $9/10$ the algorithm must output ``$s$-sparse''; 

\item If $p$'s distance to $\newLs$-sparsity is at least $\eps$, then with probability at least $9/10$ the algorithm must output ``$\eps$-far from $\newLs$-sparse.''

\end{itemize}
Note that the algorithm is allowed to depend on $\bX,d,\boldeta,s,\newLs,K$, and $\eps$; in particular, we may assume that the algorithm ``knows'' the distributions of $\bX$ and $\boldeta$ and ``knows'' the values of $d,s,\newLs,K,$ and $\eps$.
\end{definition}

Our first result is an algorithm for the above problem, with constant sample complexity, for $T$ chosen as a certain function of $s$.  While this gap between $s$ and $T$ may not be the best possible, an advantage of this result is that the value of $T$ (as a function of $s$ and $d$) is quite explicitly given.
Before stating this algorithmic result, we remark that throughout the paper we  assume that the background distribution $\bX$ satisfies $\E[\bX]=0$ and $\Var[\bX]=1$.  We impose these normalization conditions on $\bX$ to avoid a proliferation of parameters, though we expect that our arguments can be extended to the general case without significant difficulty.
\begin{restatable} 
[The first  algorithm:  A coarse but explicit upper bound]
{theorem}
{DFKOinformal}
\label{thm:DFKO-informal}
Fix any finitely supported  random variable $\bX$ over $\R$ 
which has $\E[\bX]=0$ and $\Var[\bX]=1$. 
Let $\calD$ be the i.i.d.~product distribution $\calD \equiv \bX^{\otimes n}$.
Let $\boldeta$ be a real random variable corresponding to a noise distribution over $\R$ which is such that all its moments are finite.
Let $\newLs := \Upsilon_{K,\bX}(s,d,\eps),$ where
\begin{equation} \label{eq:upsilon}
\Upsilon_{K,\bX}(s,d,\eps) = {\frac {(O_{\bX,K}(1))^{d^2} \cdot s^d}{\eps^{2d}}}.
\end{equation}
Then there is an algorithm for the $(\bX,d,\boldeta,s,\newLs,\eps)$ polynomial sparsity testing problem that uses a number $m$ of samples that is at most 
\begin{equation} \label{eq:bookshelf}
\exp\pbra{\poly\pbra{\frac{1}{\eps},s,(O_{K,\bX}(1))^d}}\cdot
m_{[1,\ell+1]}(\boldeta)^\ell
\,,
\end{equation}
where $\ell = (O_{\bX,K}(1))^{d}\cdot s/\eps$ and $m_{[1,\ell+1]}(\boldeta)$ is defined in \Cref{sec:notation};
in particular, it is independent of $n$. 
We refer to such an algorithm as an \emph{$(\eps,\newLs)$-tester for $s$-sparsity under $\calD$ and $\boldeta$ with sample complexity $m$.}
\end{restatable}

Note that while ~\Cref{thm:DFKO-informal} is potentially non-optimal vis-a-vis the gap between $T$ and $s$, the sample complexity is completely explicit.
This feature will be useful in designing our \emph{optimal} testing algorithm, which is discussed in \Cref{sec:intro-sharp-result}. 

\paragraph{Techniques for the coarse upper bound (\Cref{thm:DFKO-informal}).} 
Before turning to \Cref{sec:intro-sharp-result}, though, let us discuss the techniques we use to prove \Cref{thm:DFKO-informal}. 
To explain the main idea, we assume for simplicity that $K=1$; in other words, we are promised that $\|p\|_\coeff=1$. 
It is an easy observation that any such $s$-sparse multilinear degree-$d$ polynomial $p$ must take values in the interval $[-\sqrt{s} \cdot M^d, \sqrt{s} \cdot M^d]$ with probability $1$, where $M$ is an upper bound on the values in $\supp(\bX)$, the support of $\bX$. 
In a central technical result of this paper, \Cref{thm:far-from-sparse-large-values}, we establish a \emph{robust converse} to this observation. 
Specifically, we show that if $p$ is $\epsilon$-far from $T$ sparse where $T = \Upsilon_{K, \bX}(s,d,\eps)$, then $|p(\bx)|$ takes values ``noticeably larger'' than  $\sqrt{s} \cdot M^d$ with ``noticeable probability'' when $\bx \sim \bX^{\otimes n}$. 

This result alone would suffice for an $(\epsilon, T)$-tester if the algorithm had access to noise-free labeled samples $(\bx, p(\bx))$. 
However, in our setting, the algorithm gets noisy samples of the form $(\bx,\by=p(\bx)+\boldeta)$. 
Because of the added noise $\boldeta$,  even in the case when $p$ is $s$-sparse, the labels $\by$ may
take values larger than $\sqrt{s} \cdot M^d$ or any other pre-specified threshold. 
Thus, it is not enough to simply ``look for large values.'' 
Instead, we show that when $p$ is far from $s$-sparse, the \emph{higher-order moments} of $p(\bx)$ differ noticeably from those of any $s$-sparse polynomial. 
In more detail, if $p$ is $s$-sparse, then the $\ell^\text{th}$ moment of $p(\bX^{\otimes n})$ is bounded in magnitude by $(\sqrt{s} \cdot M^d)^{\ell}$ (as $p(\bX^{\otimes n})$ is always bounded in $[-\sqrt{s} \cdot M^d, \sqrt{s} \cdot M^d]$). 
In contrast, we  use \Cref{thm:far-from-sparse-large-values} to show that if $p$ is $\epsilon$-far from $T$ sparse where $T = \Upsilon_{K, \bX}(s,d,\eps)$, then the $\ell^\text{th}$ moment of $p(\bX^{\otimes n})$ is noticeably larger than $(\sqrt{s} \cdot M^d)^{\ell}$ provided $\ell$ is sufficiently large (\Cref{cor:far-from-sparse-large-moments}). 
The remaining challenge then is whether these moment can be reliably estimated from noisy  samples of the form $(\bX^{\otimes n},\by=p(\bX^{\otimes n})+\boldeta)$. 

To handle this, we work with the \emph{cumulants} of the random variable $\by$ (see \Cref{sec:prelims} for the definition of cumulants). 
In particular, for any $\ell \ge 1$, since the random variables $p(\bX^{\otimes n})$ and $\boldeta$ are independent, the $\ell^\text{th}$ cumulant (denoted by $\kappa_\ell(\cdot)$) satisfies
\[
    \kappa_\ell(p(\bX^{\otimes n})+\boldeta) = \kappa_\ell(p(\bX^{\otimes n})) +\kappa_\ell(\boldeta). 
\]
Since the distribution $\boldeta$ is known to the algorithm, this means that given the {\em exact} cumulants of the label $\by$, we can also  compute the {\em exact} cumulants of  $p(\bX^{\otimes n})$. 
Finally, since the first $\ell$ cumulants of a random variable uniquely determine the first $\ell$ moments and vice versa, this means that given the exact  moments of the label $\by$, we can also compute the  exact moments of $p(\bX^{\otimes n})$. 

While exact computation of moments of $\by$ is not feasible in our model, 
since the algorithm gets access to samples of $\by$ it can compute the moments of $\by$ to high accuracy, and these can be used (together with knowledge of the noise distribution $\boldeta$) to get high accuracy estimates of the moments of $p(\bX^{\otimes n})$.   
The details of this estimation procedure are somewhat technical and deferred to \Cref{sec:est_cum_moment}.
Summarizing this discussion, our coarse algorithm can distinguish between the cases when the polynomial $p$ is $s$-sparse versus $\epsilon$-far from $\newLs$ sparse, where $\newLs$ and the number of samples used are as given in \Cref{thm:DFKO-informal}. 

As noted above, the main technical workhorse underlying this algorithm is \Cref{thm:far-from-sparse-large-values}, which establishes a lower bound on the tail of $p(\bX^{\otimes n})$ if $p$ is far from $T$-sparse. 
In the special case where $\bX$ is uniformly distributed over ${\pm1}$, this theorem can be derived from Theorem 7 of \cite{DFKO:journal}. 
We present a generalization of their result beyond the Boolean setting to arbitrary finitely supported mean-$0$ distributions, stated as \Cref{thm:DFKO7}. 
The proof of \Cref{thm:DFKO7} is technically involved: although it follows the broad outline of \cite{DFKO:journal}, extending the argument to arbitrary finitely supported distributions requires us to go beyond the arguments in \cite{DFKO:journal} at a number of points. 
In particular, our use of a noise operator with a \emph{negative} ``noise rate'' $\rho$ (\Cref{def:nrho}) does not have a counterpart in \cite{DFKO:journal}. 

\subsubsection{Main result: A sharp characterization of constant-sample testability}\label{sec:intro-sharp-result}

The main drawback of \Cref{thm:DFKO-informal} is that it only distinguishes between the cases where the unknown polynomial $p$ is $s$-sparse versus $\epsilon$-far from $T$-sparse where $T = \Upsilon_{K, \bX}(s,d,\eps)$. 
As discussed in Example $\#2$ previously (see also \Cref{obs:upper-lower-uniform} and \Cref{thm:Murakami}), achieving constant sample complexity may require $T$ to exceed $s$; however, the optimal ``gap'' between $s$ and $T$ that permits constant-sample distinguishability has remained unclear. 

Our main result gives an exact characterization of this gap.
For any finitely supported distribution $\bX$ with $\mathbf{E}[\bX]=0$ and $\Var[\bX]=1$, we give a \emph{precise characterization} of the ``gap parameter'' $T$ for which constant sample complexity is achievable. 
Formally, we define a parameter $\MSG_{\bX,d}(s)$  (and give an algorithm to compute its value) with the following guarantee: 
\begin{itemize}
    \item If $T \ge \MSG_{\bX,d}(s)$, there exists an algorithm whose sample complexity is independent of the ambient dimension $n$ that distinguishes between the case (i) $p$ is $s$-sparse, versus (ii) $p$ is $\epsilon$-far from $T$-sparse, even in the presence of noise. 

    \item On the other hand, if $T < \MSG_{\bX,d}(s)$, then even in the absence of noise, distinguishing between cases (i) and (ii) requires $\Omega(\log n)$ samples.  
\end{itemize}

Equivalently, in the formulation of \Cref{def:problem}, we show that the $(\bX,d,\boldeta,s,\newLs,\eps)$ polynomial sparsity testing problem can be solved with $O(1)$ samples (i.e.,~a number of samples with no dependence on $n$, rather depending only on $\bX,d,\boldeta,s,\newLs$ and $\eps$) if 
$T \ge \MSG_{\bX,d}(s)$. 
On the other hand, if $T < \MSG_{\bX,d}(s)$, then the $(\bX,d,\boldeta,s,\newLs,\eps)$ polynomial sparsity testing problem necessarily requires $\Omega(\log n)$ samples. 

Furthermore, we do not just show the \emph{existence} of this ``sharp threshold'' for the gap parameter, but, as mentioned earlier, we also show that the threshold value $\MSG_{\bX,d}(s)$ is \emph{computable}. 
We do this by providing an algorithm which, given $\bX$, $d$ and $s$, exactly computes the value of $\MSG_{\bX,d}(s)$.
From an algorithmic standpoint, this yields a complete characterization of the $(\bX, d, \boldeta, s, \newLs, \epsilon)$ polynomial sparsity testing problem for finitely supported $\bX$. 
Specifically, we first compute $\MSG_{\bX,d}(s)$: if $\newLs \ge \MSG_{\bX,d}(s)$, then we give a constant sample complexity algorithm for the $(\bX,d,\boldeta,s,\newLs,\eps)$  polynomial sparsity testing problem. 
On the other hand, if $\newLs < \MSG_{\bX,d}(s)$, we can \emph{certifiably} conclude that the $(\bX,d,\boldeta,s,\newLs,\eps)$ polynomial sparsity testing problem requires $\Omega(\log n)$ samples. 

\paragraph{Defining and computing the $\MSG_{\bX,d}(s)$ function.} 

It turns out to be the case that the $\MSG_{\bX,d}(s)$ function is tightly connected to the following question: when can there exist two polynomials $p$ and $q$ with \emph{different sparsities} but \emph{identical output distributions} when each input coordinate is drawn i.i.d.~from $\bX$? 
This motivates the following definition:

\begin{definition} \label{def:max-sparsity}
Fix a distribution $\bX$ over $\R$ that has finite support. For any degree $d \geq 1$, define the function $\MSG_{\bX,d}(s)$ to be the largest natural number $t$ such that:
\begin{enumerate}

\item There exist two degree-at-most-$d$ multilinear polynomials $p(x_1,\dots,x_n)$, $q(x_1,\dots,x_n)$, where $\sparsity(p)=s$, $\sparsity(q)=t$, such that the distributions $p(\bX^{\otimes n})$ and $q(\bX^{\otimes n})$ are exactly identical; but

\item There are no two degree-at-most-$d$ polynomials $p(x_1,\dots,x_n)$, $q(x_1,\dots,x_n)$ with $\sparsity(p)=s,$ $\sparsity(q)>t$ such that the distributions $p(\bX^{\otimes n})$ and $q(\bX^{\otimes n})$ are exactly identical.
\end{enumerate}
\end{definition}

Since we can take $q = p$ in item (1) above, we clearly have that $\MSG_{\bX,d}(s) \geq s$ for any $\bX$ and $d$.  
Simple examples also show that we may have $\MSG_{\bX,d}(s) > s$; see \Cref{ap:examples} for a simple argument establishing that taking $\bX = \Unif(\{-1,1\})$, for any $s \geq 1$ we have $\MSG_{\bX,d}(s) \geq 4^{d-1}s.$

The following result, which is proved in \Cref{sec:MSG-well-defined}, 
gives an upper bound on $\MSG_{\bX,d}(s)$ and hence shows that $\MSG_{\bX,d}(s)$ is well-defined:

\begin{restatable} 
[$\MSG_{\bX,d}(s)$ is well-defined]
{theorem} 
{fiswelldefined}
\label{thm:f-is-well-defined}
    Let $\bX$ be a finitely supported distribution that takes at most $\ell$ distinct output values,
    and let $d \geq 1$ be any natural number.  Then we have that
    \[
        \MSG_{\bX,d}(s) \leq \Phi(d,s,\ell) \qquad\text{where}~\Phi(d,s,\ell) := 2^{2d^2 \cdot (\ell^{ds} + 3)}\,.
    \]
\end{restatable}

Roughly speaking, the proof proceeds by giving both (i) an \emph{upper} bound  on the number of distinct output value that an $s$-sparse degree-$d$ multilinear polynomial $p(x_1,\dots,x_n)$ can assume when its input is drawn from $\bX^{\otimes n}$, as well as (ii) a \emph{lower} bound on the number of distinct output values that a degree-$d$ multilinear polynomial $q(x_1,\dots,x_n)$ with $\sparsity(q)=T$ \emph{must} take under the same input distribution. 
If $T$ is large enough so that the lower bound in (ii) exceeds the upper bound in (i), then the distributions of $p(\bX^{\otimes n})$ and $q(\bX^{\otimes n})$ cannot be the same, and hence such a value of $T$ is an upper bound on $\MSG_{\bX,d}(s)$. 
The proof of (ii) builds on an upper bound, originally due to Nisan and Szegedy \cite{NisanSzegedy:94}, on the number of variables that any degree-$d$ \emph{Boolean-valued} function $f: \zo^n \to \zo$ can depend on.

The next theorem, which we prove in \Cref{sec:alg-to-compute-MSG}, builds on \Cref{thm:f-is-well-defined} and shows that in fact there is an algorithm to compute $\MSG_{\bX,d}(s)$:

\begin{restatable} 
[Algorithm for computing $\MSG_{\bX,d}(s)$] 
{theorem}
{algforf}
\label{thm:alg-for-f}
There is an algorithm, which we call {\tt Compute-Max-Sparsity-Gap}, with the following performance guarantee:
Let $\bX$ be a mean-zero variance-one  finitely supported distribution that puts probability $\alpha_i>0$ on $v_i \in \Q$ for $i = 1,\dots,\ell$, where each $\alpha_i$ and $v_i$ is a rational number whose numerator and denominator each have bit length at most $L$. 
The {\tt Compute-Max-Sparsity-Gap} algorithm is given as input a description of $\bX$, a positive integer $d$, and a positive integer $s$.  
It runs in time $d^{O(dL\Phi(s,d,\ell)^2)}$ and outputs the value of $\MSG_{\bX,d}(s)$.
\end{restatable}

We remark that by \Cref{thm:alg-for-f}, the value of $\MSG_{\bX,d}(s)$ is computable in time $O(1)$ for any fixed base distribution $\bX$, any constant degree bound $d$, and any sparsity bound $s$.  
(Note that none of these parameters depend on $n$---indeed, $n$ does not appear in the definition of $\MSG_{\bX,d}(s)$.)

Before turning to our characterization of constant-sample testability, we give a proof overview of \Cref{thm:alg-for-f}. 
To compute $\MSG_{\bX,d}(s)$, it suffices to decide, given values $s$ and $t$,  whether there exist two polynomials $p$ and $q$ with $\sparsity(p) =s$, $\sparsity(q)=t$ and identical output distributions when the input coordinates are drawn i.i.d.~from $\bX$. 
Note that the total number of variables in either $p$ or $q$ is at most $k:= \max\{s,t\} \cdot d$. 
As the distribution $\bX$ is finitely supported, so are the distributions of $p(\bX^{\otimes k})$ and $q(\bX^{\otimes k})$.  
Hence the distributions of $p(\bX^{\otimes k})$ and $q(\bX^{\otimes k})$ are identical if and only if their sufficiently high moments match (see~\Cref{claim:moments}). 
Furthermore, any moment of 
$p(\bX^{\otimes k})$ (respectively, of $q(\bX^{\otimes k})$) can be expressed as a polynomial in the coefficients of 
$p$ (respectively, of $q$). 
From these observations, with a little work we can express the problem of finding two polynomials $p$ and $q$ such that $\sparsity(p) =s$, $\sparsity(q)=t$ with identical distributions for $p(\bX^{\otimes k})$ and $q(\bX^{\otimes k})$  as a system of constraints (equalities and inequalities) defined by polynomials. 
Consequently, we can employ known decision procedures for the existential theory of the reals~\cite{Tarski1948,Ren:88} to algorithmically find such polynomials $p$ and $q$ if they exist, or certify that no such pair exists. 

\paragraph{The main result:  Sharp upper and lower bounds for constant-sample testability at  $\MSG_{\bX,d}(s)$.}

With the definition of $\MSG_{\bX,d}(s)$ in place, and the assurance that it is both well-defined and computable, we can now precisely characterize when the polynomial sparsity testing problem admits constant sample complexity.
This characterization is provided by the following pair of theorems:

\begin{restatable} 
[Lower bound on the polynomial sparsity testing problem] 
{theorem}
{sharplower}
\label{thm:sharp-lower}
Fix any finitely supported mean-zero variance-one distribution $\bX$ over $\R$, any degree bound $d$, any sparsity parameter $s \in \N$, and a gap value $\newLs \in \N$.
Let $\boldeta=\boldeta_0$ be the noise distribution which always outputs 0 (i.e.~there is no noise).

If $\newLs \leq  \MSG_{\bX,d}(s)-1$, then 
there exists a real error parameter $\eps>0$, independent of $n$, such that
any algorithm for the $(\bX,d,\boldeta_0,s,\newLs,\eps)$ polynomial sparsity testing problem must use $\Omega_{\bX,d,s}(\log n)$ samples.
\end{restatable}

\begin{restatable}[Upper bound on the polynomial sparsity testing problem] 
{theorem}
{sharpupper}\label{thm:sharp-upper}
Fix any finitely supported mean-zero variance-one distribution $\bX$ over $\R$, any degree bound $d$, any sparsity parameter $s \in \N$, any noise distribution $\boldeta$ over $\R$ all of whose moments are finite, a gap value $\newLs \in \N$, and a real error parameter $\eps > 0$.

If $\newLs \geq  \MSG_{\bX,d}(s)$, then there is an algorithm \TS~for the $(\bX,d,\boldeta,s,\newLs,\eps)$ polynomial sparsity testing problem that uses $O(1)$ samples (where the $O$ notation hides ``constant'' factors that depend only on $\bX,d,\boldeta,s,\newLs,K,$ and $\eps$, but with no dependence on $n$).
\end{restatable}

%
%
%
%

We stress that \Cref{thm:sharp-lower,thm:sharp-upper} together establish a ``sharp cutoff'' in the sample complexity of sparsity testing.
At the critical value $\newLs = \MSG_{\bX,d}(s)$, the required number of samples transitions abruptly: if $\newLs$ decreases by even one, the complexity jumps from $O(1)$ to $\Omega(\log n)$ for distinguishing whether $p$ is $s$-sparse or $\epsilon$-far from $\newLs$-sparse. 
We now give a brief overview of the proofs of \Cref{thm:sharp-lower,thm:sharp-upper}. 

\paragraph{Proof overview of \Cref{thm:sharp-lower}.} 

Let $p$ and $q$ be two degree-$d$ polynomials over $\mathbb{R}^k $ where $\sparsity(p)= s$ and $\sparsity(q) =  \MSG_{\bX,d}(s)$. Note that 
for any $T$ which is strictly less than  $ \MSG_{\bX,d}(s)$, the polynomial $q$ is $\epsilon$-far from any $T$-sparse polynomial for some $\epsilon>0$. Now, suppose the algorithm get samples from $(\bX^{\otimes k}, r(\bX^{\otimes k}))$ where 
\begin{enumerate}
    \item In the ``YES" case, $r=p$; 
    \item In the ``NO" case, $r=q$, i.e., the polynomial $r$ is $\epsilon$-far from any $T$-sparse polynomial. 
\end{enumerate}
The goal of the algorithm is to distinguish which of the two cases it is in. An easy observation is that if the algorithm is only allowed to look at the label, i.e. $r(\bX^{\otimes k})$, then no finite number of samples is sufficient, as the distribution of $r(\bX^{\otimes k})$ is identical in the ``YES" and the ``NO" cases.

To transform this observation into a real lower bound for our problem (since in the problem we consider, the algorithm receives as input both the example points drawn from $\bX^{\otimes k}$ as well as their labels $r(\bX^{\otimes k})$), we consider the following construction. First, split $n$ into $n/k$ blocks of size $k$ each. Let $\bX_{[i]} := (\bX_{(i-1)k +1}, \ldots, \bX_{ik})$ denote the $i^\text{th}$ block of variables (note that the distribution of each $\bX_{[i]}$ is the same as $\bX^{\otimes k}$). 
We construct ``YES'' and ``NO'' instances as follows:
\begin{enumerate}
    \item In the ``YES" case, we sample $1 \le \bi \le n/k$ uniformly at random. We then set $r= p(\mathbf{X}_{[\bi]})$, so in this case
    the algorithm gets samples from $(\bX^{\otimes n}, p(\bX_{[\bi]}))$. 
\item In the ``NO" CASE, we again sample $1 \le \bi \le n/k$ uniformly at random. We then set $r= q(\mathbf{X}_{[\bi]})$, so in this case the algorithm gets samples from $(\bX^{\otimes n}, q(\bX_{[\bi]}))$. 
\end{enumerate}
In \Cref{thm:sharp-lower}, we show that to distinguish between the ``YES" and the ``NO" cases, any algorithm must use $\Omega(\log n)$ samples. The rough intuition is that there are $n/k=\Theta(n)$ possible choices of $i$, and until the algorithm gets $\Theta(\log n)$ samples, it does not have enough information to infer the value of $i$; without knowledge of $i$, the algorithm can essentially just look at the label $r(\bX_{[i]}).$ As we have described earlier, merely looking at the label is insufficient to distinguish between the ``YES" and the ``NO" cases. 
While this is the intuition behind the lower bound, the actual proof proceeds by calculating the total variation distance between a sequence of samples drawn from the ``YES" case versus a sequence of samples drawn from the ``NO'' case, and showing that as long as we have fewer than $\Theta(\log n)$ samples, the total variation distance is negligible. 

\paragraph{Proof overview of \Cref{thm:sharp-upper}.} 

The proof of \Cref{thm:sharp-upper} proceeds in multiple phases, which we describe below. 
We first recall our setting: we are given samples of $(\bX^{\otimes n}, p(\bX^{\otimes n}))$, and the goal is to distinguish whether $p$ is $s$-sparse versus $\epsilon$-far from $T$-sparse, where $T \geq \MSG_{\bX,d}(s).$  
\begin{itemize}
    \item {\bf Phase Zero:} If $T \ge  \Upsilon_{K, \bX}(s,d,\eps)$ (where $\Upsilon_{K, \bX}(s,d,\eps)$ is as in \Cref{eq:upsilon}), we can simply run our algorithm for the ``coarse upper bound," i.e.,~the algorithm from \Cref{thm:DFKO-informal}, and if $p$ is $\epsilon$-far from $T$-sparse, then the algorithm from \Cref{thm:DFKO-informal} rejects $p$ with high probability. 
    
    \item  {\bf Phase One:} For the ensuing discussion, it is useful to have some notation. Let $\calP$ denote the set of polynomials which are $s$-sparse and $\RV(\calP)$ denote the set of random variables realizable as $p(\bX^{\otimes n})$ where $p \in \calP$. Let $\calP({\epsilon})$ be the set of polynomials which are both $\Upsilon_{K, \bX}(s,d,\eps)$-sparse and $\epsilon$ far from $T$-sparse, and let  $\RV(\calP(\epsilon))$ denote the set of random variables realizable as $p(\bX^{\otimes n})$ for some $p \in \calP({\epsilon})$. We use a compactness argument to show that there is some $c_\eps>0$ such that any element of $\RV(\calP(\epsilon))$ is at least $c_{\epsilon}$-far in Wasserstein distance from $\RV(\calP)$ (\Cref{lem:wasserstein-compact}).
    Algorithmically, in Phase~1 we estimate the minimum Wasserstein distance between any element of $\RV(\calP(\epsilon))$ and $\RV(\calP)$ and this will be used to set a parameter $\epsilon' \ll \epsilon$ whose purpose we explain next. 
  
    \item {\bf Phase Two:} We  now run the ``coarse upper bound'' algorithm from  \Cref{thm:DFKO-informal} again, but this time with its closeness parameter set to the value $\epsilon'$ mentioned above. If the algorithm rejects then we reject the polynomial $p,$ and if the algorithm  does not reject then we know that $p$ must be $\epsilon'$-close to $\Upsilon_{K, \bX}(s,d,\eps)$-sparse.
    Let $\calP(\eps,\eps')$ denote the set of polynomials which are both $\epsilon$-far from $T$-sparse and $\eps'$-close to $\Upsilon_{K, \bX}(s,d,\eps)$-sparse, and let  $\RV(\calP(\epsilon))$ denote the corresponding set of random variables.  At this point, the remaining task is to distinguish whether $p \in \calP$ versus $p \in \calP(\eps,\eps').$
    
    \item {\bf Phase Three:} Now, note that both $\RV(\calP)$ and $\RV(\calP(\epsilon))$ are compact sets which are $c_{\epsilon}$-far from each other in Wasserstein distance. Using a connection between moments and Wasserstein distance (\Cref{lem:original-kong-valiant}), this implies that the distributions in these sets can be robustly specified using their low-degree (degree independent of $n$) moments. We now explain the role of the $\eps'$ parameter: this parameter is chosen so that every element of $\RV(\calP(\epsilon,\epsilon'))$ is close (in terms of its low-degree moments) to some element of $\RV(\calP(\epsilon))$. Algorithmically, in Phase~3 we construct a net for $\RV(\calP)$ (here we are again using the fact that this is a compact set), and for each element (distribution) in the net, we compute its low-degree moments (we call this its ``moment profile'').
    
    \item {\bf Phase Four:} We estimate the low-degree moments of the given random variable 
    $p(\bX^{\otimes n})$ and check if this moment profile is sufficiently close to that of any element in the net for $\RV(\calP)$.  If this is the case then we accept, and otherwise we reject.



\end{itemize}

Note that if $p$ is $s$-sparse, then by construction, the distribution of $p(\bX^{\otimes n})$ is close in Wasserstein distance to an element in the net for $\RV(\mathcal{P})$. 
This closeness also implies that the moment profile of $p(\bX^{\otimes n})$ is close to that of an element in the same net. 
Hence, with appropriately chosen parameters, any $s$-sparse polynomial $p$ will pass the above test. 

On the other hand, if the polynomial $p$ is $\epsilon$-far from $T$-sparse and additionally survives Phase Two, then the distribution of $p(\bX^{\otimes n})$ must be $c_\epsilon$-far in Wasserstein distance from every element of $\RV(\mathcal{P})$.
For sufficiently small $\epsilon'$, it follows that $p(\bX^{\otimes n})$ also has a moment profile that is ``far'' from the moment profile of every element in the net for $\RV(\mathcal{P})$ (see \Cref{lem:moment-difference}). 
Such a polynomial will therefore be rejected, completing the proof overview of \Cref{thm:sharp-upper}. 

%% file: sections/prelims.tex

\section{Preliminaries}
\label{sec:preliminaries}

\subsection{The distributions we consider}
\label{sec:setup}

All of our results deal with finitely supported real random variables $\bX$ that have $\E[\bX]=0,\Var[\bX]=1.$  For concreteness, throughout the paper we suppose that $\bX$ puts probability $\alpha_i>0$ on $v_i \in \R$ for $i = 1,\dots,\ell$. Throughout the paper we let $M \geq 1$ denote $\max_{i \in [\ell]} |v_i|$, so the support of $\bX$ is contained in $[-M,M]$, and we write $\lambda$ to denote $\min_{i \in [\ell]} \alpha_i$, the minimum nonzero probability that $\bX$ puts on any outcome.

\subsection{Notation} \label{sec:notation}

Given a real random variable $\bY$ and $q \geq 1$, we write $\|\bY\|_q$ for the $q$-norm 
of $\bY$, i.e.
\[
\|\bY\|_q := \Ex[ |\bY|^q]^{1/q}.
\]
We write $m_\ell(\bY)$ to denote the $\ell$-th raw moment of $\bY$, i.e.
\[m_\ell(\bY) := \E[\bY^\ell].
\]
For $a \leq b$ we write 
$m_{\maxl{a}{b}}(\bY)$ 
to denote the largest raw moment in the range $a,a+1,\dots,b$, i.e.
\begin{equation} \label{eq:moment-notation}
m_{\maxl{a}{b}}(\bY):=\max\{m_k(\bY):a\leq k\leq b\}.
\end{equation}

\subsection{Multilinear polynomials over product distributions}
\label{subsec:fourier}

Our notation and terminology follow Chapter~8 of \cite{odonnell-book}. 
We write a real multilinear polynomial $p(x_1,\dots,x_n)$ of degree at most $d$ as
\begin{equation} \label{eq:polynomial}
p(x) = \sum_{S \in {[n] \choose \leq d}} \widehat{p}(S) x_S,
\quad \quad \text{where~}x_S := \prod_{i \in S} x_i.
\end{equation}
For $i\in[n]$, we will write $\wh{f}(i)$ instead of $\wh{f}(\{i\})$ for notational simplicity.

As the above notation suggests, we refer to the coefficient $\widehat{p}(S)$ as the \emph{$S$-th Fourier coefficient of $p$.}  
To justify this choice of terminology, we observe that if $\bX$ is a mean-0, variance-1 real random variable (as it is throughout this paper), then for $\bX^{\otimes n}=(\bX_1,\dots,\bX_n)$ an i.i.d.~tuple of draws from $\bX$ we have that for any two distinct multilinear monomials $x_S \neq x_T$, it is the case that $\E[\bX_S \bX_T]=0$ and $\E[(\bX_S)^2]=1$. So the ${n \choose \leq d}$ many multilinear monomials $\pbra{x_S}_{S \in {[n] \choose \leq d}}$ form an orthonormal basis for the vector space of all multilinear degree-at-most-$d$ polynomials under $\bX^{\otimes n}.$

For a multilinear polynomial $p(x)$ as in \Cref{eq:polynomial}, the \emph{influence of coordinate $i$ on $p$} is defined as
\begin{equation}
    \label{eq:influence}
\Inf_i[p] 
= \Ex\sbra{\vabs{\frac{\partial p(\bX)}{\partial x_i}}^2} 
= \sum_{S \ni i} \widehat{p}(S)^2\,.
\end{equation}

\subsection{Coefficient distance and distance to sparsity} \label{sec:coeff-distance}

As described earlier in \Cref{def:coeff}, for a degree-$d$ multilinear polynomial 
\[
p(x_1,\dots,x_n) = \sum_{S \in {[n] \choose \leq d}} \widehat{p}(S) x_S,
\]
we write $\|p\|_\coeff$ for
\[
\|p\|_\coeff := \pbra{\sum_{S \in {[n] \choose \leq d}} \widehat{p}(S)^2}^{1/2},
\]
and we define $p$'s \emph{distance to $s$-sparsity} to be
\[
\dist(p,s\text{-sparse}):= \min_{q \text{~is an $s$-sparse degree-$d$ multilinear polynomial}} {\frac {\|p - q\|_\coeff}{\|p\|_\coeff}}.
\]
The following simple lemma, whose proof is an easy exercise, will be useful:

\begin{lemma}\label{lem:first-s-sum-of-squares}
    For any $\delta > 0$ and  degree-$d$ multilinear polynomial $p: \mathbb{R}^n \to \mathbb{R}$ which satisfies $\|p\|_\coeff=1$,  $p$ is $\delta$-far from $s$-sparse if and only if the sum of squares of its $s$ coefficients with largest magnitude is at most $1-\delta^2$. 
    %
%
\end{lemma}

\subsection{Wasserstein distance and moment distance} \label{sec:wass-moment}

Some of our arguments will use the notion of Wasserstein distance.  Given two real random variables $\bY$ and $\bZ$, we recall that the \emph{Wasserstein distance between $\bY$ and $\bZ$} is
\begin{equation}
\label{eq:wasserstein}
W_1(\bY,\bZ) := \inf_{(\bY',\bZ')}\Ex \sbra{|\bY' - \bZ'|},
\end{equation}
where the infimum is taken over all couplings $(\bY',\bZ')$ of $\bY$ and $\bZ$ (so the marginal distribution $\bY'$ of $(\bY',\bZ')$ is distributed according to $\bY$, and likewise for $\bZ'$ and $\bZ$).

Another notion of the distance between two real random variables $\bY$ and $\bZ$ (with finite moments) that we will need is the ``$k$-th order moment distance,'' which we denote $\Mom_k(\bY,\bZ).$
This is defined as
\begin{equation} \label{eq:moment-distance}
\Mom_k(\bY,\bZ) = \|m(\bY) - m(\bZ)\|_2,
\end{equation}
where for a real random variable $\bZ$, we write $m(\bZ)$ to denote the $k$-dimensional vector whose $i$-th coordinate is $m_i(\bZ)$.

\subsection{Tools from hypercontractivity}
\label{subsec:prelims-hypercontractivity}

Throughout the paper we will be concerned with product distributions of the form $\bX^{\otimes n}$ where $\bX$ is as specified in \Cref{sec:setup}.
A fundamental fact about product probability spaces $\bX^{\otimes n}$ of the above form is that they are \emph{hypercontractive}.  See Chapter~10 of \cite{odonnell-book} for a detailed treatment of hypercontractivity over such product spaces; for our purposes, it will suffice to use the following result: 

\begin{theorem} [Theorem~10.21 of \cite{odonnell-book}] \label{thm:OD10.21}
Let $\bY$ be a finitely supported real random variable, and let $\lambda > 0$ be the minimum nonzero probability of any outcome of $\bY$.
For any degree-$d$ multilinear polynomial $r(x_1,\dots,x_n)$, for any $q > 2$, for $\by \sim \bY^{\otimes n}$ we have that 
\begin{equation} \label{eq:10.21}
\E[|r(\by)|^q] \leq \pbra{\sqrt{q - 1} \cdot \lambda^{1/q - 1/2}}^{dq} \E[r(\by)^2]^{q/2}.
\end{equation}


\end{theorem}

A few consequences of \Cref{thm:OD10.21} that will be useful for us are detailed in \Cref{subsubsec:bonami}.

\subsection{Cumulants} \label{sec:prelims}

In this section we recall some necessary background on cumulants and how they are related to moments.

\begin{definition} \label{def:cumulants}
The \emph{cumulants} of a real random variable $\bY$ are defined by the cumulant generating function $K(t)$, which is the natural logarithm of the moment generating function $M(t) = \E[e^{t \bY}]$:
\[
K(t)=\ln \E[e^{t\bY}],
\quad \quad \text{or equivalently,}
\quad \quad
e^{K(t)}= \E[e^{t\bY}].
\]
For $\ell > 0$ the cumulants of $\bY$, which are denoted $\cum_\ell(\bY)$, are the coefficients in the Taylor expansion of the cumulant generating function about the origin:
\[
K(t) = \sum_{\ell=1}^\infty \cum_\ell(\bY) {\frac {t^\ell}{\ell!}},
\quad \quad \text{or equivalently,}
\quad \quad
\cum_\ell(\bY)=K^{(\ell)}(0).
\]
\end{definition}

An easy consequence of the definition of cumulants is that they are additive for independent random variables:

\begin{fact} \label{fact:additive}
If $\bY$ and $\bZ$ are independent random variables then $\cum_\ell(\bY+\bZ)=\cum_\ell(\bY)+\cum_\ell(\bZ)$.
\end{fact}

Recall that in our framework, the input data consists of labeled examples $(\bx,\by=p(\bx) + \ion)$, where $\bx \sim \bX^{\otimes n}$ and $\ion$ is independently drawn from a noise distribution. 
Our algorithms work by estimating cumulants of the real random variable $\by$; by \Cref{fact:additive}, we can subtract out the effect of the noise random variable $\ion$ on $\kappa_\ell(\by)$ and we are left with (an estimate of) $\kappa_\ell(p(\bx)).$


It is well known that for any $n \geq 1$, there is a one-to-one mapping between the first $n$ moments and the first $n$ cumulants, which can be derived by relating coefficients in the Taylor series expansions of the cumulant and moment generating functions (see, e.g., \cite{Barndorff}):

\begin{fact}\label{fact:transfer_moments_cumulants}
Let $\bY$ be a random variable with mean zero. 
Then
\begin{equation}\label{eq:degree_l_cum}
\cum_\ell(\bY)=m_\ell(\bY)-\sum_{j=1}^{\ell-1} {\ell-1 \choose j-1} \cum_j(\bY) \cdot m_{\ell-j}(\bY).
\end{equation}
Cumulants can be expressed in terms of moments and vice-versa:
\begin{equation}\label{eq:moments_cumulants}
m_\ell(\bY)=\sum_{k=1}^\ell B_{\ell,k}\bigg( \cum_1(\bY),\ldots,\cum_{\ell-k+1}(\bY) \bigg)
\end{equation}
and
\begin{equation}\label{eq:cumulants_moments}
\cum_\ell(\bY)=\sum_{k=1}^\ell (-1)^{k-1} (k-1)! B_{\ell,k}\bigg( m_1(\bY),\ldots,m_{\ell-k+1}(\bY) \bigg	),
\end{equation}
where $B_{\ell,k}$ are incomplete \emph{Bell polynomials},
$$
B_{\ell,k}(x_1,\ldots,x_{\ell-k+1})=\sum \frac{\ell!}{j_1!\cdots j_{\ell-k+1}!} \left(\frac{x_1}{1} \right)^{j_1} \cdots \left(\frac{x_{\ell-k+1}}{(\ell-k+1)!}\right)^{j_{\ell-k+1}},
$$
where the summation is over all non-negative sequences $(j_1,\ldots,j_{\ell-k+1})$ that satisfy
$$
j_1+\cdots+j_{\ell-k+1}=k \text{ and } j_1+2j_2 + \cdots + (\ell-k+1)j_{\ell-k+1}=\ell.
$$
\end{fact}

\Cref{eq:degree_l_cum} can be used to give a upper bound on $\cum_\ell(\bY)$ in terms of the absolute moments of $\bY$:

\begin{claim} 
\label{clm:up_bound_cumulant}
For any random variable $\bY$ with mean zero and any positive integer $\ell$, we have $|\cum_\ell(\bY)| \le \Ex[|\bY|^\ell] \cdot e^{\ell} \cdot \ell!$.
\end{claim}
\begin{proof}
    The proof is given in Appendix~A.1 in \cite{CDS20stoc}.
\end{proof}



%% file: sections/moment_cum_est.tex


\section{Estimating moments of $p(\bX^{\otimes n})$ given noisy samples}\label{sec:est_cum_moment}

The take-away result from this section is \Cref{lem:estimating-clean-moments}. 
Roughly speaking, \Cref{lem:estimating-clean-moments} provides an algorithm
\EMWN~which, given i.i.d.~noisy data points $(\bx,p(\bx)+\ion)$ where each $\bx \sim \bX^{\otimes n}$ and each $\ion$ is independent of $\bx$, outputs an additively accurate estimate of the ``noiseless'' $\ell$-th raw moment of $p(\bX^{\otimes n})$:

\begin{lemma}
[Estimating moments of $p(\bX^{\otimes n})$ in the presence of noise] \label{lem:estimating-clean-moments}
Let $\bX$ be a real-valued random variable as in \Cref{sec:setup} (with $\E[\bX]=0,\Var[\bX]=1$,  $\supp(\bX) \subseteq [-M,M]$, and minimum nonzero probability $\lambda$ on any outcome), and let $\ion$ be a real-valued random variable with finite moments of all orders. 

There is an algorithm \EMWN$_{\bX,\ion}$ (depending on $\bX$ and $\ion$) with the following property:  Let $p$ be any (unknown) degree-at-most-$d$ multilinear polynomial  that is promised to have $\|p\|_\coeff \in [1/K,K]$. Given any desired additive accuracy parameter $\tau$, any desired confidence parameter $\delta$ and any positive integer $\ell$, the algorithm draws\footnote{Recall that the $m_{\maxl{1}{\ell+1}}(\ion)$ notation was defined in \Cref{eq:moment-notation}.}
\[m=\poly(\ell^{d\ell^2},K^{\ell^2},1/(\delta\tau),1/\lambda^{d\ell^2},m_{\maxl{1}{\ell+1}}(\ion)^\ell)
\]
many independent random samples $(\bx^{(1)},\by^{(1)}),\dots, (\bx^{(m)},\by^{(m)})$, where each $\bx^{(i)}\sim \bX^{\otimes n}$ and each $\by^{(i)}=p(\bx^{(i)}) + \ion$, where each draw of $\ion$ is independent of everything  else.
It outputs an estimate $\wt{m}_{\ell}(p)$ which with probability at least $1-\delta$ satisfies
\[
\left| \wt{m}_{\ell}(p) - m_{\ell}(p(\bX^{\otimes n})) \right| \le \tau.
\]
\end{lemma}

Looking ahead, we will use \Cref{lem:estimating-clean-moments} and the \EMWN$_{\bX,\ion}$ algorithm in the proofs of \Cref{thm:DFKO-informal} and \Cref{thm:sharp-upper}.

\subsection{Overview of the argument}

At a high level, a top-down description of our approach is that we  estimate the ``noiseless'' moment $m_\ell(p(\bX^{\otimes n}))$ using highly accurate estimates of the noiseless cumulants $\kappa_1(p(\bX^{\otimes n})),\dots,\kappa_{\ell}(p(\bX^{\otimes n}))$ and \Cref{eq:moments_cumulants}, which expresses the $\ell$-th moment of a random variable in terms of its cumulants of orders $1,\dots,\ell$. 
To estimate each noiseless cumulant $\kappa_i(p(\bX^{\otimes n}))$, we obtain an estimate of the ``noisy'' cumulant $\kappa_i(p(\bX^{\otimes n})+\ion)$ and subtract $\kappa_i(\ion)$ (note that the quantity $\kappa_i(\ion)$ is ``known'' to the algorithm  since the \EMWN$_{\bX,\ion}$ algorithm depends on the random variable $\ion$); since $\ion$ is independent of $\bX^{\otimes n}$ and hence of $p(\bX^{\otimes n})$, \Cref{fact:additive} ensures correctness of this approach.
To estimate the noisy cumulant $\kappa_i(p(\bX^{\otimes n})+\ion)$, we obtain an estimate of the noisy moments $m_1(p(\bX^{\otimes n}) + \ion),\dots,m_i(p(\bX^{\otimes n})+\ion)$ and use \Cref{eq:cumulants_moments}, which expresses the $i$-th cumulant of a random variable in terms of its moments of order $1,\dots,i$. 
To estimate the noisy moments $m_j(p(\bX^{\otimes n})+\ion)$, we  draw many independent pairs $(\bx^{(1)},\by^{(1)}),(\bx^{(2)},\by^{(2)})\dots$ and take an empirical average of the $(\by^{(1)})^j, (\by^{(2)})^j, \dots$ values, applying Chebychev's inequality and suitable moment/variance bounds to establish correctness.

A bottom-up restatement of the above plan, which corresponds to the sequence of intermediate results we establish below, is as follows:

\begin{itemize}

\item \Cref{claim:bounds-noiseless-moments} establishes upper and lower bounds on the noiseless moments of $p(\bX^{\otimes n})$;

\item \Cref{lem:estimating-noiseless-moments} gives an algorithm to estimate the noiseless moments of $p(\bX^{\otimes n})$ given noiseless draws from $p(\bX^{\otimes n})$;

\item Building on \Cref{claim:bounds-noiseless-moments}, \Cref{claim:bounds-noisy-moments} establishes upper and lower bounds on the noisy moments of $p(\bX^{\otimes n}) + \ion$;

\item Building on \Cref{lem:estimating-noiseless-moments}, \Cref{lem:estimating-noisy-moments-noisy-data} gives an algorithm to estimate the noisy moments of $p(\bX^{\otimes n}) + \ion$ given noisy draws from $p(\bX^{\otimes n}) + \ion$;

\item \Cref{lem:estimating-noisy-cumulants} gives an algorithm to estimate the noisy cumulants of $p(\bX^{\otimes n}) + \ion$ given noisy draws from $p(\bX^{\otimes n}) + \ion$;

\item Finally, we use the algorithm of \Cref{lem:estimating-noisy-cumulants} to prove \Cref{lem:estimating-clean-moments}.

\end{itemize}
The detailed proof is given in \Cref{ap:estimating-moments}.

%% file: sections/DFKO-structural.tex

\section{Structural results for finitely supported product distributions}
\label{sec:DFKO-structural}

\newcommand{\qmax}{q_{\max}}

The main result of this section is a generalization of an anti-concentration bound established by Dinur, Friedgut, Kindler, and O'Donnell~\cite{DFKO:journal} for bounded low-degree polynomials over the Boolean hypercube $\bn$. 
We give a broad generalization of this result which extends it to low-degree polynomials over $\bX^{\otimes n}$ for any mean-0 finitely supported random variable $\bX$. 

Our notation throughout follows~\Cref{sec:setup}. 
In particular, $\bX$ will be a finitely supported real-valued random variable with $\Ex[\bX] = 0$ and $\Varx[\bX] = 1$. 
As before, we will assume that $\bX$ puts probability $\alpha_i > 0$ on $v_i \in \R$ for $i = 1, \dots, \ell$, and will write $\lambda := \min_{i \in [\ell]} \alpha_i$. 
The main result we will prove, \Cref{thm:DFKO7}, is a generalization of Theorem~7 of~\cite{DFKO:journal} which forms the basis of our main structural result in \Cref{sec:DFKO-algorithm}, \Cref{thm:far-from-sparse-large-values}.

\begin{restatable}{theorem}{DFKOseven}
\label{thm:DFKO7}
	There exists a constant $C = C_{\bX}$ such that the following holds: Suppose $f: \R^n \to \R$ is a multilinear polynomial with degree at most $d$.
	Let $J \subseteq [n]$ and suppose that
	\begin{equation} \label{eq:DFKO7-cond1}
		\sum_{S: S \setminus J \neq \emptyset} \widehat{f}(S)^2 \geq \delta\,.
	\end{equation}
	Suppose, furthermore, that 
	\begin{equation}
		\label{eq:DFKO7-cond2}
	t \geq \sqrt{\delta} \quad \quad \text{and} \quad \quad \Inf_i[f] \leq \delta^2 t^{-2} C^{-d}
		\text{~for all~}i \notin J\,.
	\end{equation}
	Then we have 
	\begin{equation}
	\label{eq:DFKO7-cond3}
            \Prx_{\bx\sim\bX^{\otimes n}}\sbra{|f(\bx)| \geq t} \geq \exp(-Ct^2 d^2 \log (d) / \delta)\,.
	\end{equation}
\end{restatable}

The high-level structure of our argument is similar to that of \cite{DFKO:journal}'s proof of their Theorem~7, but our more general setting requires us to go beyond their approach at a number of points in the argument. 
In particular, our use of a noise operator with a \emph{negative} noise rate (\Cref{def:nrho}) and the subsequent arguments employing it do not have a counterpart in~\cite{DFKO:journal}.

\subsection{Useful preliminaries}
\label{subsec:DFKO-prelims}

We first record some useful preliminaries before proving \Cref{thm:DFKO7}. 

\subsubsection{The noise operator}
\label{subsec:noise}

Throughout this section, we will write 
\begin{equation} \label{eq:def-rho-min}
	\rho_{\min}(\bX) := \frac{1}{4}\cdot\min_{i \in [\ell]} \frac{\alpha_i}{1-\alpha_i} 
    = \frac{\lambda}{4(1-\lambda)}\,.
\end{equation} 

\begin{definition} \label{def:nrho}
	For $\rho \in [-\rho_{\min}(\bX), 1]$ and $v_i \in \supp(\bX)$, we define $N_{\rho}(v_i)$ as the distribution on $\supp(\bX)$ where 
	\begin{equation} \label{eq:nrho}
		\Prx_{\by\sim N_\rho(v_i)}[\by = v_j] 
		= 
		\begin{cases}
			\rho  + (1-\rho)\alpha_i & i = j\,,\\
			(1-\rho)\alpha_j & i \neq j\,.
		\end{cases}
	\end{equation}
	More generally, for $x \in \supp(\bX^{\otimes n})$ we will abuse notation and define the product distribution $N_\rho(x) := \otimes_{i=1}^n N_{\rho}(x_i)$. 
\end{definition}

We note that the above definition is usually made for $\rho \in [0,1]$ (cf.~Definition~8.26 of~\cite{odonnell-book}), but we will require the extension to mildly negative $\rho$. 
It is readily verified from~\Cref{eq:def-rho-min} that the probability densities from \Cref{eq:nrho} do indeed form a valid probability distribution, and furthermore $\Ex_{\by\sim N_\rho(x)}[\by] = \rho\cdot x$ since $\bX$ has mean $0$. 

\begin{definition}[Noise operator] \label{def:noise}
	For $\rho \in [-\rho_{\min}(\bX), 1]$, we define the \emph{noise operator with parameter $\rho$} as the linear operator $\T_\rho$ on multilinear polynomials $f: \R^n\to\R$ defined by 
	\[
		\T_\rho f(x) := \Ex_{\by\sim N_\rho(x)}\sbra{f(\by)}\,.
	\]
\end{definition}

It is readily verified that 
\begin{equation} \label{eq:noise-action}
    \T_\rho f (x) = \sum_{S} \rho^{|S|}\wh{f}(S)\chi_S(x)\,.
\end{equation} 
 
\subsubsection{Hypercontractivity and its applications}
\label{subsubsec:bonami}

Our proof of~\Cref{thm:DFKO7} makes crucial use of some consequences of hypercontractivity (\Cref{thm:OD10.21} in \Cref{subsec:prelims-hypercontractivity}). 
We refer the reader to Chapters~9 and 10 of~\cite{odonnell-book} for further background on hypercontractivity and its applications. 

\begin{proposition}[Theorem~10.23 of~\cite{odonnell-book}]
\label{prop:DFKO-lemma-2.5}
    Let $\bY := \bX_1 \otimes \dots \otimes \bX_n$ be a product distribution where each marginal has minimal outcome probability at least $\lambda$. 
	Suppose $f:\R^n\to\R$ is a multilinear polynomial of degree at most $d$. 
	Then 
	\[
		\Prx_{\bx\sim\bY}\sbra{f(\bx) \geq \E[f]} \geq \frac{1}{4}\pbra{\frac{2\lambda}{e^2}}^d \geq \pbra{\frac{\lambda}{15}}^d\,.
 	\]
\end{proposition}

\begin{proposition}[Theorem~10.24 of~\cite{odonnell-book}]
\label{prop:DFKO-lemma-2.2}
	Suppose $f:\R^n\to\R$ is a multilinear polynomial of degree at most $d$. 
	Then for $t \geq (\sqrt{2e/\lambda})^d$, 
	\[
		\Prx_{\bx\sim\bX^{\otimes n}}\sbra{|f(\bx)| \geq t\|f\|_2} \leq \lambda^d \exp\pbra{-\frac{d}{2e}\lambda t^{2/d}}\,.
	\]
	
\end{proposition}

We record two consequences of~\Cref{prop:DFKO-lemma-2.2} for future use: 

\begin{proposition}
\label{prop:DFKO-lemma-2.3}
Let $\bX$ be as in~\Cref{sec:setup}. 
Suppose that $f:\R^n\to\R$ is a multilinear polynomial of degree at most $d$ and that $\|f\|_2=1$. 
Then for all $t \geq (\sqrt{2e/\lambda})^d$, we have
\[
	\Ex_{\bx\sim\bX^{\otimes n}}\!\Big[\,f(\bx)^2\cdot\mathbf{1}_{\{|f(\bx)|>t\}}\,\Big]
	\leq 
	t^2\exp\!\pbra{-\frac{\lambda d}{4e}\,t^{2/d}}.
\]
\end{proposition} 

\begin{proof}
Fix $t \geq (\sqrt{2e/\lambda})^d$ and let $r = t^{2/d}/C$, where $C:=e/2\lambda$. 
Note that $r \geq 4$.
By H\"older's inequality, 
\[
	\Ex\sbra{f^2\cdot\mathbf{1}_{\{|f|>t\}}}
	\leq \|f\|_r^2\cdot\Pr\sbra{|f|>t}^{1-2/r}.
\]
Since $\|f\|_2=1$, \Cref{thm:OD10.21} and our choice of $C$ gives $\|f\|_r \leq (C r)^{d/2}$.  
Applying~\Cref{prop:DFKO-lemma-2.2} to the tail,
\[
	\Pr\sbra{|f|>t} \le \lambda^d\exp\!\pbra{-\frac{d}{2e}\lambda t^{2/d}}\,,
\]
and substituting both bounds yields
\[
	\Ex\sbra{f^2\cdot\mathbf{1}_{\{|f|>t\}}}
	\leq (C r)^d\exp\!\Big(-\frac{d}{2e}\lambda t^{2/d}\,(1-2/r)\Big)\,.
\]
By the choice $r = t^{2/d}/C$, we have $(C r)^d=t^2$, and since $r\ge4$, it follows that $1-2/r\ge\tfrac12$. 
Hence
\[
	\Ex\sbra{f^2\cdot\mathbf{1}_{\{|f|>t\}}}
	\leq t^2\exp\!\pbra{-\frac{\lambda d}{4e}\,t^{2/d}},
\]
as claimed.
\end{proof}

\begin{corollary}
\label{cor:DFKO-2.4}
Under the hypotheses of~\Cref{prop:DFKO-lemma-2.3}, let $t_0 := (2e/\lambda)^d$. 
Then
\[
	\Ex_{\bx\sim\bX^{\otimes n}}\!\Big[\,f(\bx)^2\cdot\mathbf{1}_{\{|f(\bx)|>t_0\}}\,\Big]
	\leq 0.52\,.
\]
\end{corollary}

\begin{proof}
Applying~\Cref{prop:DFKO-lemma-2.3} with $t=t_0=(2e/\lambda)^d$ gives
\[
	\Ex\sbra{f^2\cdot\mathbf{1}_{\{|f|>t_0\}}}
	\leq t_0^2\,\exp\!\pbra{-\frac{\lambda d}{4e}\,t_0^{2/d}}
	\leq \pbra{\pbra{\frac{2e}{\lambda}}^2 e^{-e/\lambda}}^d 
	\leq \pbra{\pbra{\frac{2e}{\lambda}}^2 e^{-e/\lambda}}\,,
\]
since $d \geq 1$. 
Since $\lambda \leq 0.5$, it is readily checked that the final quantity above is at most $0.52$, which completes the proof.
\end{proof}

\subsubsection{Chebyshev polynomials}
\label{subsubsec:chebyshev-polys}

We write $T_d$ to denote the Chebyshev polynomial of the first kind; we refer the reader to Rivlin's monograph~\cite{Rivlin:90} for further background and information on Chebyshev polynomials.
(Note the distinction between $T_d$ and $\T_\rho$, where the latter is the noise operator from \Cref{def:noise}.) 
Following~\cite{DFKO:journal}, we will write \smash{$\eta^{(d)}_0, \dots, \eta^{(d)}_d$} for the $(d+1)$ extrema of $T_d$ in the segment $[-1,1]$ at which $|T_d| = 1$. 

\begin{lemma}[Corollary~2.7 of~\cite{DFKO:journal}] 
\label{cor:DFKO-2.7}
	Let $d$ be odd and let $p(x) = a_0 + a_1 x + \dots + a_d x^d$. 
	Then there exists $j \in \{0,\dots, d\}$ such that
	\[
		|p(\eta^{(d)}_j)| \geq \frac{|a_1|}{d}\,.
	\]
\end{lemma}

\subsubsection{Kolmogorov minoration}
\label{subsubsec:kolmogorov-minoration}

We will require the following supergaussian estimate on sums of independent bounded mean-0 random variables: 

\begin{lemma}[Kolmogorov minoration, Lemma~8.1 of~\cite{LedouxTalagrand}]
\label{lemma:kolmogorov}
	Let $\by_1, \dots, \by_n$ be independent real-valued random variables such that $\E[\by_i] = 0$ and furthermore each $\by_i$ is supported in $[-M, M]$. 
	Let $V := \sum_i \Var[\by_i]$. 
	Finally, let $\gamma > 0$. 
	There exist positive constants $c_\gamma$ and $C_\gamma$ (depending only on $\gamma$) such that for every $t$ satisfying 
	\[
		c_\gamma \sqrt{V} \leq t \leq C_\gamma\frac{V}{M}\,,
	\]
	we have 
	\[
		\Prx\sbra{\sumi \by_i > t} \geq \exp\pbra{\frac{-(1+\gamma)t^2}{2V}}\,.
	\]
\end{lemma}

\subsection{The principal lemma}
\label{subsec:DFKO-lemma-1.3}

The proof of \Cref{thm:DFKO7} will rely on the following lemma, which can be viewed as a non-linear generalization of the Kolmogorov minoration. 

\begin{lemma}[Generalization of Lemma~1.3 of~\cite{DFKO:journal}]
\label{lemma:our-1.3}
	Suppose $\bX$ is as in~\Cref{sec:setup}. 
	There exists a constant $K := K_{\bX}$ such that the following holds:
	let $f:\R^n\to\R$ be a multilinear polynomial with degree at most $d$ and $\sum_{i=1}^n \wh{f}(i)^2 \geq 1$. 
	Let $t \geq 1$ and suppose that $|\wh{f}(i)| \leq 1/(Ktd)$ for all $i\in[n]$. Then 
	\[
		\Prx_{\bx\sim\bX^{\otimes n}}\sbra{|f(\bx)| \geq t} \geq \exp\pbra{-O_{\bX}(t^2d^2)}\,.
	\]
\end{lemma}

\begin{proof}
	By scaling $f$, we may assume that $\sum_{i=1}^n \wh{f}(i)^2 = 1$. 
	Let $\ell$ denote the linear part of $f$, i.e., 
	\[
		\ell(x) := \sumi \wh{f}(i)x_i\,. 
	\]
	Let $\gamma$ be an arbitrary positive constant, and let $c_\gamma, C_\gamma$ be as in~\Cref{lemma:kolmogorov}. 
	We define
	\begin{equation} \label{eq:1.3-t-prime-def}
			t' := c_\gamma\frac{(d+1)}{\rho_{\min}(\bX)}\cdot t\,.
	\end{equation}
	
	\paragraph{Step 1: Applying Kolmogorov minoration.} We will first apply~\Cref{lemma:kolmogorov} to obtain $x_0 \in \supp(\bX^{\otimes n})$ such that $\ell(x_0)$ is large.  
	Let $\by_i := \wh{f}(i)\cdot\bx_i$ for $\bx\sim\bX^{\otimes n}$. 
	Note that the $\by_i$'s are independent and satisfy 
	\[
		\E[\by_i] = 0\,, 
		\quad 
		\Var[\by_i] = \wh{f}(i)^2\,, 
		\quad\text{and}\quad 
		|\by_i| \leq \max_{i\in[n]}|\wh{f}(i)|\cdot M \leq \frac{M}{K td}\,,
	\]
	where the final inequality above uses the promised upper bound on $|\wh{f}(i)|$ from the statement of~\Cref{lemma:our-1.3}. 
	Next, note that $V := \sum_i \Var[\by_i] = 1$ since we rescaled $f$ to ensure $\sumi \wh{f}(i)^2 = 1$. 
	By definition (\Cref{eq:1.3-t-prime-def}), $t' \geq c_\gamma$, and note that the following inequality
	\[
		t' = c_\gamma\frac{(d+1)}{\rho_{\min}(\bX)}\cdot t \stackrel{?}\leq C_\gamma \cdot \frac{Ktd}{M} 
	\]
	holds if we take $K \geq \frac{c_\gamma 2M}{C_\gamma\rho_{\min}(\bX)}$. (This is the source of the dependence of the constant $K = K_{\bX}$ on the random variables $\bX$; recall that $\bX$ is supported in the interval $[-M, M]$.) 
	We can thus apply the Kolmogorov minoration to obtain
	\begin{equation} \label{eq:final-KM-app}
		\Pr\sbra{\sumi \by_i >t'} = \Prx_{\bx\sim\bX^{\otimes n}}\sbra{\ell(\bx) > t'} \geq \exp\pbra{\frac{-(1+\gamma)t'^2}{2}}\,.
	\end{equation} 
	We will say that any $x_0$ in the support of $\bX^{\otimes n}$ for which $|\ell(x_0)| \geq t'$ is \emph{good}. 
	
\paragraph{Step~2: A great noise rate for every good outcome.} 
	
	Let $x_0 \in \supp(\bX^{\otimes n})$ be good. 
	For $\rho \in [-\rho_{\min}(\bX), \rho_{\min}(\bX)]$ define 
	\[
		p_{x_0}(\rho) := \T_\rho f(x_0)\,,
	\]
	and note that since $f$ has degree at most $d$, the polynomial $p_{x_0}$ has degree at most $d$ (in the variable $\rho$).
	Note also that by~\Cref{eq:noise-action}, the linear coefficient of $p_{x_0}$ is $\ell(x_0)$. 
	
	Define the polynomial $q : [-1,1] \to \R$ as
	\[
		q(\rho) := p_{x_0}\pbra{\rho\cdot\rho_{\min}(\bX)}\,,
	\]
	and note that it is a polynomial with degree at most $d$ and linear coefficient $\ell(x_0)\cdot\rho_{\min}(\bX)$. 
    We assume that the degree bound on the function $f$ is odd, incrementing $d$ by $1$ if necessary. 
    In particular, we will use $(d+1)$ as the degree bound on $f$ in what follows to ensure this.
	
	For notational convenience, we define
	\[
		\rho_j := \eta_j^{(d+1)}\cdot\rho_{\min}(\bX)~\text{for}~j\in\{0, \dots, d+1\}\,,
	\]
	where $\eta_0^{(d+1)}, \dots \eta_{d+1}^{(d+1)} \in [-1,1]$ are the $(d+2)$ extrema of the degree-$(d+1)$ Chebyshev polynomial $T_{d+1}$, cf.~\Cref{subsubsec:chebyshev-polys}. 
	It follows from~\Cref{cor:DFKO-2.7} that there exists some $j(x_0) \in \{0,\dots, d+1\}$ such that 
	\begin{equation} \label{eq:motobecane}
		\abs{p_{x_0}(\rho_{j(x_0)})} = \abs{q(\eta_{j(x_0)}^{(d+1)})} \geq \frac{|\ell(x_0)|\cdot\rho_{\min}(\bX)}{d+1} \geq \frac{t'\cdot\rho_{\min}(\bX)}{d+1} = 
		c_\gamma t\,, 
	\end{equation}
	where the second inequality relies on the fact that $x_0$ is good. 
	
	Now, note that 
	\begin{equation} \label{eq:SEC}
		\T_{\rho_{j(x_0)}} f(x_0) = \Ex_{\bz\sim N_{\rho_{j(x_0)}}(x_0)}\sbra{f(\bz)}\,,
	\end{equation} 
	where $\bz$ is drawn from the product distribution $N_{\rho_{j(x_0)}}(x_0)$ as defined in~\Cref{def:nrho}. 
    Since $|\rho| \leq \rho_{\min}(\bX)$, it is readily checked that the minimum probability of any outcome in each marginal of $N_{\rho_j(x_0)}(x_0)$ is at least 
    \[  
        \lambda'_{\bX} := \min\cbra{-\rho_{\min}(\bX) + (1-\rho_{\min}(\bX))\lambda, (1-\rho_{\min}(\bX))\lambda}\,,
    \]
    and recalling \Cref{eq:def-rho-min} it is readily checked that $\lambda'_{\bX}$ is a constant strictly greater than $0$ for any positive $\lambda$, i.e.~for any finitely supported $\bX$.
    
	Returning to~\Cref{eq:SEC}, since $f(\bz)$ is a polynomial in $\bz$ of degree at most $d$, we can apply~\Cref{prop:DFKO-lemma-2.5} to it (and by replacing $f$ by $-f$, we get the same bound on the probability that $-f$ goes below $-\E[f]$). 
	This gives
	\begin{equation} \label{eq:sofra}
		\Prx_{\bz\sim N_{\rho^*(x_0)}(x_0)}\sbra{\abs{f(\bz)} \geq c_\gamma t} 
		\geq 
		\Theta\pbra{\lambda'_{\bX}}^d\,.
	\end{equation} 

	\paragraph{Step~3: Putting everything together.} 
	Consider the following two-step process of drawing a random point $\by \sim \bX^{\otimes n}$. 
	First, draw $\boldsymbol{\rho}$ uniformly from the set $\{\rho_0, \dots, \rho_{d+1}\}$. Then draw $\bx\sim\bX^{\otimes n}$, and output $\by \sim N_{\boldsymbol{\rho}}(\bx)$. 
	Note that $\by$ is distributed as $\bX^{\otimes n}$. 
	From~\Cref{eq:final-KM-app}, we know that with probability at least $\exp(-O(t'^2))$, the outcome of $\bx$ is good and furthermore $\boldsymbol{\rho} = \rho_{j(\bx)}$ with probability $(d+2)^{-1}$; in this case \Cref{eq:sofra} implies that $\Pr[|f(\by)| \geq c_\gamma t] \geq \Theta_{\bX}(1)^d$. 
	Thus, 
	\[
		\Prx_{\by\sim\bX^{\otimes n}}\sbra{|f(\by)| \geq c_\gamma t} \geq \exp\pbra{-O(t'^2)}\cdot \Theta_{\bX}(1)^d = \exp\pbra{-O_{\bX}(t^2d^2)}\,,
	\]
	as desired, recalling once again our choice of $t'$ from \Cref{eq:1.3-t-prime-def}.

    To conclude the proof of \Cref{lemma:our-1.3}, note that the lemma's statement omits $c_\gamma$, but this is harmless. First, recall that $c_\gamma$ came from \Cref{lemma:kolmogorov}, and we took $\gamma$ to be an arbitrary positive constant (see the line before~\Cref{eq:1.3-t-prime-def}). When $c_\gamma \ge 1$, the bound remains valid as written; and when $c_\gamma \le 1$, we can instead apply the result with $t / c_\gamma \ge 1$ and absorb the constant $c_\gamma$ into the $O(\cdot)$ term.
\end{proof}

Towards proving~\Cref{thm:DFKO7}, which is the main goal of this section, it will be helpful to have the following slightly modified version of~\Cref{lemma:our-1.3}:

\begin{lemma}[Generalization of Lemma~4.1 of~\cite{DFKO:journal}] \label{lemma:our-4.1}
	There exists a constant $K' := K'_{\bX}$ (which is different from the constant $K_{\bX}$ from~\Cref{lemma:our-1.3}) such that the following holds: Suppose $f: \R^n\to\R$ is a multilinear polynomial of degree at most $d$. 
	Let $T \sse [n]$ and $t \geq 1$. 
	Suppose 
	\[
		\sum_{i \in T} \wh{f}(i)^2 \geq 1
		\qquad \text{and} \qquad 
		|\wh{f}(i)| \leq \frac{1}{K'td}~\text{for all}~i\in T\,.
	\]
	Then 
	\[
		\Prx_{\bx\sim\bX^{\otimes n}}\sbra{|f(\bx)| \geq t} \geq \exp(- K' t^2 d^2)\,.
	\]
\end{lemma}

Note that in contrast to~\Cref{lemma:our-1.3}, we now only have a hypothesis on the linear part of $f$ restricted to the set $T$.  

\begin{proof}
	By scaling $f$, we may assume that $\sum_{i\in T} \wh{f}(i)^2 = 1$. 
	We continue with almost the same proof as of~\Cref{lemma:our-1.3}, except for the following difference: instead of simply using \Cref{eq:final-KM-app}, which stated that $\Pr[|\ell(\bx)| \geq t'] \geq \exp(-Kt'^2)$ thanks to the Kolmogorov minoration (\Cref{lemma:kolmogorov}), we instead break $\ell(x)$ into two pieces by defining the random variables $\bU := \sum_{i \in T} \wh{f}(i)\bx_i$ and $\bV := \sum_{i \notin T} \wh{f}(i)\bx_i$ where $\bx\sim\bX^{\otimes n}$. 
	
	Note that $\bU$ and $\bV$ are mean-$0$ random variables, and furthermore $\Var[\bU] =1$. 
	As in~\Cref{lemma:our-1.3}, we apply Kolmogorov minoration (\Cref{lemma:kolmogorov}) to $\bU$ to obtain
	\[
		\Prx\sbra{|\bU| \geq t'} \geq \exp\pbra{-O_{\bX}(t'^2)}\,.
	\]
	It follows that either 
	\[
		\Prx\sbra{\bU \geq s} \geq \frac{1}{2}\exp\pbra{-O_{\bX}(t'^2)}
		\qquad\text{or}\qquad
		\Prx\sbra{\bU \leq -s} \geq \exp\pbra{-O_{\bX}(t'^2)}\,. 
	\]
	Assume without loss of generality that it is the former; the argument in the other case is symmetric (by considering $-f$). 
	Now, it follows from \Cref{prop:DFKO-lemma-2.5} that 
	\[
		\Pr[\bV \geq 0] \geq \frac{\lambda}{15}\,,
	\]
	since $\bV$ is a degree-$1$ polynomial. 
	As $\bV$ is independent of $\bU$, we have 
	\[
		\Prx_{\bx\sim\bX^{\otimes n}}\sbra{|\ell(\bx)| \geq t'} 
		\geq 
		\Prx\sbra{\bU \geq t'~\text{and}~\bV \geq 0} 
		\geq \frac{\lambda}{30}\exp\pbra{-O_{\bX}(t'^2)}\,,
	\]
    and the proof follows by absorbing the $\Theta(\lambda)$ term into the $O_{\bX}(\cdot)$.
\end{proof}

\subsection{Proof of~\Cref{thm:DFKO7}}

In the spirit of delayed gratification and to keep the reader appropriately on edge, we first prove the following result which is very close to (but not quite) \Cref{thm:DFKO7}. 

\begin{restatable}[Generalization of Theorem~3 of \cite{DFKO:journal}]{theorem}{DFKOthree}
\label{thm:DFKO3}
	Let $\bX$ be as in \Cref{sec:setup}. 
	There exists a constant $C = C_{\bX}$ such that the following holds: 
    Suppose that $f: \R^n \to \R$ is a multilinear polynomial with degree at most $d$ with $\Var[f(\bX)] = 1$. 
    Suppose that $t \geq 1$ and 
    \begin{equation}
    \label{eq:DFKO3-cond1}
        \Inf_i[f] \leq t^{-2} C^{-d}~\text{for all}~i \in [n]\,.
    \end{equation}
    Then
    \[
        \Prx_{\bx\sim\bX^{\otimes n}}\sbra{|f(\bx)| \geq t} \geq \exp(-Ct^{2} d^2 \log d)\,. 
    \]
\end{restatable}

\begin{proof}

    Given $s \geq 1$ we write $|S| \sim 2^s$ if $S\sse[n]$ satisfies $2^{s-1}\leq |S| \leq 2^s$. 
    Since $\Var[f] = \sum_{S\neq \emptyset} \wh{f}(S)^2 = 1$, there exists some $s \in [\ceil{\log_2 d} + 1]$ such that 
    \begin{equation} \label{eq:special-s-thm3}
        \sum_{S : |S|\sim 2^s} \wh{f}(S)^2 \geq \frac{1}{2\log d}\,.
    \end{equation}
	For the remainder of the argument, we will fix an $s$ for which~\Cref{eq:special-s-thm3} holds. 

    Given $U \sse [n]$ and $i\in [n]$, define 
    \[
        \gamma_i(U) := \sum_{S: S\cap U = \{i\}} \wh{f}(S)^2\,.
    \]
    Note that $\gamma_i(U) = 0$ if $i \notin U$, and furthermore $\gamma_i(U) \leq \sum_{S \ni i} \wh{f}(S)^2 = \Inf_i[f]$, hence by \Cref{eq:DFKO3-cond1} we have
    \begin{equation} \label{eq:DFKO-5}
        \gamma_i(U) \leq t^{-2} C^{-d} \qquad\text{for all}~i\in[n]\,. 
    \end{equation}

    Let $\calS$ be the distribution on subsets of $[n]$ where $\bU\sim\calS$ is drawn by including each coordinate in $\bU$ independently with probability $2^{-s}$.  
    We have 
    \begin{align*}
        \Ex_{\bU\sim\calS}\sbra{\gamma_i(\bU)} 
        &= \sum_{S \sse [n]} \Prx_{\bU\sim\calS}[S \cap \bU = \{i\}]\cdot\wh{f}(S)^2 \\
        &\geq \sum_{|S|\sim2^s,\, S \ni i} 2^{-s}\cdot(1-2^{-s})^{2^s}\cdot\wh{f}(S)^2 \\
        &\geq \sum_{|S|\sim2^s,\, S \ni i} 2^{-s}\cdot \frac{1}{4}\cdot\wh{f}(S)^2\,.
    \end{align*}
    Summing over all $i \in [n]$, each $S$ with $|S| \sim 2^s$ is counted at least $2^{s-1}$ times and so thanks to \Cref{eq:special-s-thm3} we have that 
    \[
        \Ex_{\bU\sim\calS}\sbra{\sumi \gamma_i(\bU)} \geq \frac{1}{8} \sum_{|S| \sim 2^s} \wh{f}(S)^2 \geq \frac{1}{16\log d}\,.
    \]
    Since $\sum_i \gamma_i(U) \in [0,1]$ for every $U \sse [n]$, it follows by a ``reverse Markov'' inequality\footnote{Formally, if $\bY$ is bounded in $[0,1]$, then $\Pr\sbra{\bY \geq 0.5\E[\bY]} \geq 0.5\E[\bY]$.} that 
    \begin{equation} \label{eq:dfko-eq-6}
        \Prx_{\bU\sim\calS}\sbra{\sum_{i \in \bU}\gamma_i(\bU) \geq \frac{1}{32\log d}} \geq \frac{1}{32\log d}\,,
    \end{equation}
    where we relied on the fact that $\gamma_i(U) = 0$ if $i \notin U$. 
    
    For now, fix $U$ to be a set 
    for which 
    \begin{equation} \label{eq:jungle-volcano}
        \sum_{i \in U}\gamma_i(U) \geq \frac{1}{32\log d}\,.
    \end{equation}
    Let $\by$ be a random assignment to the coordinates in $[n]\setminus U$, where each $\by_i \sim \bX$. 
    Let $f_{\by} : \R^U \to \R$ denote the restriction of $f$ obtained by fixing the coordinates in $[n]\setminus U$ to $\by$. 
    Considering \smash{$\wh{f_{\by}}(i)$} as a function of $\by$, it is readily seen that this is a function of degree at most $d$, that it has \smash{$\big\{\hat{f}(S) : S \cap U = \{i\}\big\}$} as Fourier coefficients, and that it has no other nonzero Fourier coefficients. 
    Therefore by definition of $\gamma_i$, for all $i \in U$ we have that 
    \[
    	\Ex_{\by}\sbra{\wh{f_{\by}}(i)^2} = \gamma_i(U)\,.
    \]
	
	We will view $\widehat{f_{\by}}(i)$ as a random variable with variance at most~$\gamma_i(U)$. 
	Applying~\Cref{cor:DFKO-2.4} to 
	the function (in $y$) $\hat{f_{y}}(i)$ and rescaling it appropriately, we get 
	\[
		\Ex_{\by}\sbra{\widehat{f_{\by}}(i)^2\cdot\mathbf{1}_{\{\widehat{f_{\by}}(i)^2>(2e/\lambda)^{2d}\gamma_i(U)\}}} 
		\leq 0.52\,\gamma_i(U),
	\]
	and hence
	\[
		\Ex_{\by}\sbra{\widehat{f_{\by}}(i)^2\cdot\mathbf{1}_{\{\widehat{f_{\by}}(i)^2\le(2e/\lambda)^{2d}\gamma_i(U)\}}}
		\geq 0.48\,\gamma_i(U).
	\]
	Define the indicator $\mathbf{1}_i(y)$ for the event $\{\widehat{f_y}(i)^2\le(2e/\lambda)^{2d}\gamma_i(U)\}$.
	Summing over $i\in U$ and taking expectations, we obtain
	\begin{equation}\label{eq:fy-linear-mass}
		\Ex_{\by}\sbra{\sum_{i\in U}\widehat{f_{\by}}(i)^2\cdot\mathbf{1}_i(\by)}
		\geq 0.48\,\sum_{i\in U}\gamma_i(U)
		\geq \frac{0.48}{32\log d}\,,
	\end{equation}
	where the final inequality relies on \Cref{eq:jungle-volcano}. 
	By definition of of $\mathbf{1}_i(\cdot)$ it holds for every $y$ that 
	\[
		\sum_{i \in U} \wh{f_y}(i)^2\cdot\mathbf{1}_i(\by) \leq \pbra{\frac{2e}{\lambda}}^{2d}\sum_{i \in U} \gamma_i(U) \leq \pbra{\frac{2e}{\lambda}}^{2d}\,,
	\]
	where the final inequality relies on the fact that $\gamma_i(U) = \sum_{S\cap U = \{ i\}} \wh{f}(S)^2$ and the assumption that $\|f\|_2 = 1$. 
	We thus conclude from~\Cref{eq:fy-linear-mass} and yet another application of the ``reverse Markov'' inequality that 
	\begin{equation} \label{eq:our-DFKO-eq-7}
		\Prx_{\by}\sbra{\sum_{i\in U}\widehat{f_{\by}}(i)^2\cdot\mathbf{1}_i(\by)
		\geq \frac{0.48}{64\log d}}
		\geq \frac{0.48}{64\log d}\pbra{\frac{\lambda}{2e}}^{2d}\,.
	\end{equation}
	Now, drawing $\bU$ as before and combining \Cref{eq:dfko-eq-6,eq:our-DFKO-eq-7} gives 
	\begin{equation} \label{eq:our-dfko-8}
		\Prx_{\bU,\by}\sbra{\sum_{i\in \bU}\wh{f_{\by}}(i)^2\cdot\mathbf{1}_i(\by) \geq \frac{0.0001}{\log d}} \geq \exp\pbra{-O(d)}\,.
	\end{equation} 
	
	Let us now condition on the event that $\bU$ and $\by$ satisfy the event in~\Cref{eq:our-dfko-8}, i.e., 
	\begin{equation} \label{eq:DFKO-9}
		\sum_{i\in \bU}\wh{f_{\by}}(i)^2\cdot\mathbf{1}_i(\by) \geq \frac{0.0001}{\log d}\,.
	\end{equation}
	Define the function $g : \R^U \to \R$ as $g(z) = f_y(z)$, and let $T := \{i \in U : \mathbf{1}_i(y) = 1\}$ and $\sigma^2 := \sum_{i\in T}\wh{g}(i)^2$. 
	In particular, it follows from \Cref{eq:DFKO-9} that $\sigma^2 \geq 0.0001(\log d)^{-1}$. 
	Note that $g$ has degree at most $d$, and by definition of $T$ and of the indicator function $\mathbf{1}_i(\cdot)$, 
	\begin{equation} \label{eq:DFKO-10}
		\max_{i \in T} |\wh{g}(i)| 
		\leq \pbra{\frac{2e}{\lambda}}^{d}\sqrt{\gamma_i(U)} 
		\leq t^{-1}\pbra{\frac{\sqrt{C}\lambda}{2e}}^{-d}\,,
	\end{equation}
	where the second inequality follows from~\Cref{eq:DFKO-5}. 
	
	We will now apply \Cref{lemma:our-4.1} to the function $g/\sigma$, the set $T$, and the parameter $t'$ where $t' := \max\{1,t/\sigma\}$. 
	To see that the conditions of the lemma hold, note that:
	\begin{itemize}
		\item Our choice of $\sigma$ ensures that $\sum_{i \in T} \wh{g/\sigma}(i)^2 \geq 1$. 
		\item It follows from~\Cref{eq:DFKO-10} that 
			\[
				\frac{1}{\sigma}\max_{i\in T}|\wh{g}(i)| \leq \frac{1}{\sigma t}	\pbra{\frac{\sqrt{C}\lambda}{2e}}^{-d} \ll \frac{1}{td}\,,		
			\]
			with much room to spare (assuming $C$ is sufficiently large). 
	\end{itemize} 
	It then follows that 
	\begin{align}
		\Prx_{\bz}\sbra{|f_y(\bz)| \geq t} = \Prx_{\bz}\sbra{|g(\bz)| \geq t}
		&\geq \Prx_{\bz}\sbra{|g(\bz)/\sigma| \geq t'}	\nonumber \tag{Using $t' \geq t/\sigma$} \\
		&\geq \exp\pbra{-O_{\bX}(t'^2 d^2)} \tag{\Cref{lemma:our-4.1}} \\
		&\geq \exp\pbra{-O_{\bX}(t^2 d^2 \log d)}\,, \label{eq:cow}
	\end{align}
	where~\Cref{eq:cow} relies on our choice of $t',\sigma$ and on~\Cref{eq:DFKO-9}. 
	
	We now combine \Cref{eq:our-dfko-8,eq:cow} to get 
	\[
		\Prx_{\bU,\by,\bz}\sbra{|f_{\by}(\bz)| \geq t} \geq \exp\pbra{-O_{\bX}(t^2 d^2 \log d)}\,,
	\]
	but note that the L.H.S.~above is nothing but $\Prx\sbra{|f(\bx)| \geq t}$ for $\bx\sim\bX^{\otimes n}$. 
	This completes the proof of \Cref{thm:DFKO3}. 
\end{proof}

We finally turn to the proof of our main structural result, \Cref{thm:DFKO7}, which essentially follows from the proof of~\Cref{thm:DFKO3}. 

\DFKOseven* 

\begin{proof}
	Rescale by a factor of $1/\sqrt{\delta}$, letting $f' = (f/\sqrt{\delta})$ and $t' = (t/\sqrt{\delta})$. 
	Note we may repeat the proof of~\Cref{thm:DFKO3} for $f'$ and $t'$, with the following alterations: $s$ is chosen so that $\sum_{|S\setminus J| \sim 2^s} \wh{f'}(S)^2 \geq 1/(2\log d)$, and $\bU$ is chosen at random from $[n]\setminus J$ rather than from $[n]$. 
\end{proof}

\begin{remark} \label{remark:invariance}
O'Donnell and Zhao~\cite{odonnell2016polynomial} give a different and ``more conceptual'' proof of the main result of~\cite{DFKO:journal} using the invariance principle (see Chapter~11 of~\cite{odonnell-book}).
Even in the Boolean setting, however, their approach yields a quantitatively weaker bound than the one that \cite{DFKO:journal} obtain. 
While the conclusion of their Corollary~3.5 matches the strength of \cite{DFKO:journal}'s Theorem~3, the required upper bound on $\max_i |\widehat{f}(i)|$ is exponentially stronger, $2^{-O(t^2 d^2)}$ rather than $2^{-O(d)}/t^2$. 
In our setting, this would translate to a tester distinguishing $s$-sparse functions from those that are constant-far from $2^{O(s)}$-sparse, instead of the sharper $s$-sparse versus constant-far-from-$s^d$-sparse guarantee that we obtain (in the next section) via the results of this section.
\end{remark}

%% file: sections/DFKO-algorithm.tex


\section{Proof of \Cref{thm:DFKO-informal}: A coarse upper bound}
\label{sec:DFKO-algorithm}

The main result of this section is a proof of \Cref{thm:DFKO-informal}, which we recall below:

\DFKOinformal*

\subsection{Analyzing far-from-sparse polynomials}
\label{sec:DFKO-alg-far-from-sparse}

We start with the following claim, which establishes that far-from-sparse polynomials satisfy the conditions of \Cref{thm:DFKO7}

\begin{claim}
\label{claim:gumbo} 
    Let $\bX$ be a finitely supported real random variable as in \Cref{sec:setup}, so in particular  $\E[\bX]=0$ and $\Var[\bX]=1$, and let $C = C_{\bX}$ be the constant from \Cref{thm:DFKO7}. 
    Suppose $t \geq \sqrt{\eps/K}$, and let $p(x_1,\dots,x_n)$ be a multilinear polynomial of degree at most $d$ satisfying  $1/K \leq \|p\|_{\coeff} \leq K$ which is $\eps$-far from being $\newLs$-sparse, where
    \begin{equation} \label{eq:L1}
        \newLs := \pbra{e K^{4} t^{2} C^{d}/\eps^{2}}^d.
    \end{equation}
    Then taking $J \subseteq [n]$ to be the set of all coordinates $i$ such that $\Inf_i[p] >  \kappa := eK^2/\newLs^{1/d}$, we have that 
    \begin{equation} \label{eq:DFKO7-cond1new}
    \sum_{S: S \setminus J \neq \emptyset} \widehat{f}(S)^2 \geq \eps/K
    \end{equation}
    and
    \begin{equation} \label{eq:DFKO7-cond2new}
    \Inf_i(p) \leq (\eps/K)^2 t^{-2} C^{-d}
    \text{~for all~}i \notin J\,.
    \end{equation}
\end{claim}

Note that \Cref{eq:DFKO7-cond1new} and \Cref{eq:DFKO7-cond2new} correspond to \Cref{eq:DFKO7-cond1} and \Cref{eq:DFKO7-cond2} respectively, but with the $\delta$ parameter of \Cref{eq:DFKO7-cond1,eq:DFKO7-cond2} now being $\eps/K$.

\begin{proof}[Proof of \Cref{claim:gumbo}]
We have that
\begin{align*}
\sum_{i \in J} \Inf_i(p)
\leq 
\sum_{i \in [n]} \Inf_i(p) 
& =
\sum_{i \in [n]} \pbra{
 \sum_{S \ni i} \widehat{p}(S)^2} \\ 
&=
\sum_{S \in {[n] \choose \leq d}} |S| \cdot \widehat{p}(S)^2 \\ 
& \leq
d \sum_{S \in {[n] \choose \leq d}} \widehat{p}(S)^2 \\
& \leq d \|p\|_{\coeff}^2 \leq dK^2\,,
\end{align*}
and hence since each $i \in J$ has $\Inf_i(p) \geq \kappa$, it follows that $|J| \leq dK^2/\kappa.$
Hence the number of sets $S \in {[n] \choose {\leq d}}$ such that $S \setminus J = \emptyset$ is at most 
\begin{equation} \label{eq:limabean}
{|J| \choose \leq d} \leq \pbra{{\frac{e |J|} d}}^d \leq \pbra{{\frac {eK^2}{\kappa}}}^d
= \newLs.
\end{equation}
Since by assumption $p$ is $\eps$-far from being $\newLs$-sparse, \Cref{eq:limabean} implies (recalling \Cref{def:distance}) that 
\[
\sum_{S: S \setminus J \neq \emptyset} \widehat{f}(S)^2 
\geq
 \dist(p,\newLs\text{-sparse}) \cdot \|p\|_{\coeff} 
 \geq \eps \cdot \|p\|_\coeff
 \geq \eps/K,
 \]
 so indeed \Cref{eq:DFKO7-cond1new} holds.

Turning to \Cref{eq:DFKO7-cond2new}, we see that by definition of $J$, for all $i \notin J$ we have that
\[
\Inf_i(p) \leq \kappa = {\frac {eK^2}{\newLs^{1/d}}} = (\eps/K)^2 t^{-2} C^{-d}
\]
as claimed in \Cref{eq:DFKO7-cond2new}, and the proof is complete.
\end{proof}

We are ready to establish the following central result, which says that far-from-sparse polynomials will take ``large'' values with not-too-small probability:
\begin{theorem}
\label{thm:far-from-sparse-large-values}
Let $\bX$ be a finitely supported real random variable as in \Cref{sec:setup}, so $\E[\bX]=0,$ $\Var[\bX]=1$, and the support of $\bX$ is contained in $[-M,M]$, where $M \geq 1.$
Let $p(x_1,\dots,x_n)$ be a multilinear polynomial with $1/K \leq \|p\|_{\coeff} \leq K$ which is $\eps$-far from being $\newLs$-sparse, where 
\begin{equation} \label{eq:L2}
    \newLs := 
    \pbra{e K^{4} \pbra{2K M^d \sqrt{s}}^2 C^{d}/\eps^{2}}^d =
    \pbra{
    {\frac {4eK^6 M^{2d} C^d s}{\eps^2}}
    }^d =: \Upsilon_{K,\bX}(s,d,\eps)
\end{equation}
(cf.~\Cref{eq:upsilon}) and $C = C_{\bX}$ is the constant from \Cref{thm:DFKO7}. 
Then we have that 
\begin{equation} \label{eq:chickenpickle}
\Pr[|p(\bX^{\otimes n})| \geq 2K M^d \sqrt{s}] \geq q, 
\quad \text{where~}q := \exp\pbra{-4C K^3 M^{2d} s d^2 \log (d) / \eps}.
\end{equation}
\end{theorem}

\begin{proof}
We choose 
\[
t = 2K M^d \sqrt{s}
\]
in \Cref{claim:gumbo}.
(Note that with this choice, since $K,M,s \geq 1$ we have $t \geq 2$, so the $t \geq \sqrt{\eps/K}$ requirement of \Cref{claim:gumbo} is satisfied with room to spare.)
With this choice of $t$ we have that \Cref{eq:L2} is consistent with \Cref{eq:L1}.
The conclusions~(\ref{eq:DFKO7-cond1new}) and (\ref{eq:DFKO7-cond2new}) of \Cref{claim:gumbo} align with conditions~(\ref{eq:DFKO7-cond1}) and (\ref{eq:DFKO7-cond2}) of \Cref{thm:DFKO7} respectively (taking the $\delta$ of \Cref{thm:DFKO7} to now be $\eps/K$), so applying \Cref{thm:DFKO7}, we get \Cref{eq:chickenpickle} as claimed, and the theorem is proved.
\end{proof}

An immediate corollary of \Cref{thm:far-from-sparse-large-values} is the following moment lower bound:

\begin{corollary} [Moment lower bound for far-from-sparse polynomials] \label{cor:far-from-sparse-large-moments}
Let $\ell \geq 2$ be even.  
For $\bX$ and $p$ as in \Cref{thm:far-from-sparse-large-values}, we have that
\[
\Ex\sbra{p(\bX^{\otimes n})^\ell} \geq q \cdot \pbra{2KM^d \sqrt{s}}^\ell.
\]
\end{corollary}

\subsection{Analyzing sparse polynomials}
\label{sec:DFKO-alg-sparse}

In contrast with \Cref{cor:far-from-sparse-large-moments}, our goal in this subsection is to give an \emph{upper} bound on the moments of sparse polynomials. This is a much easier task:  a simple argument shows that any sparse polynomial can \emph{never} take large values:

\begin{lemma}
\label{lem:sparse-small-values}
Let $\bX$ be any random variable which is supported on $[-M,M]$, and let 
$p(x_1,\dots,x_n)$ be a multilinear polynomial with $1/K \leq \|p\|_{\coeff} \leq K$ which is $s$-sparse.  Then we have that 
\[
|p(\bX^{\otimes n})|\leq  K M^d \sqrt{s}
\]
with probability 1.
\end{lemma}
\begin{proof}
Write
\[
p(x_1,\dots,x_n) = \sum_{i=1}^s \widehat{p}(S_i) x_{S_i},
\]
where each $S_i \in {[n] \choose d}$ corresponds to a multilinear degree-$d$ monomial
and we have 
\[
    \pbra{\sum_{i=1}^s \widehat{p}(S_i)^2}^{1/2} \leq K
\]    
by assumption.
The lemma is an immediate consequence of the Cauchy-Schwarz inequality and the fact that each multilinear monomial $\prod_{i \in S_i} x_i$, for $i \in [s],$ will always have magnitude at most $M^d$ under $\bX^{\otimes n}$.
\end{proof}

An easy corollary of \Cref{lem:sparse-small-values} is the following moment upper bound:

\begin{corollary} [Moment upper bound for sparse polynomials] \label{cor:sparse-small-moments}
Let $\ell \geq 2$ be even.  
For $\bX$ and $p$ as in \Cref{lem:sparse-small-values}, we have that
\[
\Ex\sbra{p(\bX^{\otimes n})^\ell} \leq  \pbra{K M^d \sqrt{s}}^\ell.
\]
\end{corollary}

\subsection{The algorithm}

Comparing \Cref{cor:far-from-sparse-large-moments} and \Cref{cor:sparse-small-moments}, we see that taking $\ell > \log_2({\frac 1 {2 q}})$ where $q$ is as in \Cref{eq:chickenpickle}, we have that for any multilinear polynomial $p$ with $1/K \leq \|p\|_{\coeff} \leq K$,

\begin{enumerate}
    \item If $p$ is $\eps$-far from $\newLs$-sparse then \smash{$\E[p(\bX^{\otimes n})^\ell] \geq 2 \cdot \pbra{KM^d \sqrt{s}}^\ell $}, whereas
    
    \item If $p$ is $s$-sparse then \smash{$\E[p(\bX^{\otimes n})^\ell] \leq \pbra{KM^d \sqrt{s}}^\ell$}. 
\end{enumerate}

Intuitively, in words, the $\ell^\text{th}$ moment will be ``large'' when $p$ is far-from-$\newLs$-sparse and will be ``small'' when $p$ is $s$-sparse.  
So the task for our algorithm is clear: to estimate the $\ell^\text{th}$ order cumulant of $p(\bX^{\otimes n})$ to within an additive 
\[
    \pm 0.1 \cdot \pbra{CM^d \sqrt{s}}^\ell~\text{error}.
\]  
This is accomplished by taking \smash{$\tau = 0.1 \cdot (CM^d \sqrt{s})^\ell$} in the algorithm of \Cref{lem:estimating-clean-moments}, and outputting ``sparse'' if the estimate $\wt{m}_{\ell}(p)$ returned by that algorithm satisfies  \smash{$\wt{m}_{\ell}(p) \leq {\frac 3 2} (CM^d \sqrt{s})^\ell.$}  It is easy to check that the number of samples used is as claimed in \Cref{thm:DFKO-informal}, so this completes the proof of \Cref{thm:DFKO-informal}. \qed

%% file: sections/MSG-well-defined.tex


\section{Proof of \Cref{thm:f-is-well-defined}:  The $\MSG_{\bX,d}$ function is well-defined}
\label{sec:MSG-well-defined}

We recall the statement of \Cref{thm:f-is-well-defined}:

\fiswelldefined*

\Cref{thm:f-is-well-defined} follows easily from the following two lemmas. The first lemma (which is extremely simple) upper bounds the number of distinct output values that can be achieved by any $s$-sparse degree-at-most-$d$ multilinear polynomial $p$, and the second lemma lower bounds the number of distinct output values that any degree-at-most-$d$ multilinear polynomial $q$ with $\sparsity(q)=T$ must achieve.

\begin{lemma} \label{lem:output-values-ub}
Let $p(x)$ be an $s$-sparse degree-at-most-$d$ multilinear polynomial.  For $\bx \sim \bX^{\otimes n}$, the random variable $p(\bx)$ takes at most $\ell^{ds}$ distinct output values.
\end{lemma}
\begin{proof}
Since $\bX$ has support size $\ell$, each the (at most) $s$ monomials comprising $p(x)$ can take at most $\ell^d$ distinct output values.
\end{proof}

\begin{lemma} \label{lem:output-values-lb}
Let $q(x)$ be a degree-at-most-$d$ multilinear polynomial with $\sparsity(q)=T.$  For $\bx \sim \bX^{\otimes n}$, the random variable $q(\bx)$ takes at least $\frac{1}{2d^2} \cdot \log(T) - 3$ distinct output values.
\end{lemma}

\begin{proofof}{\Cref{thm:f-is-well-defined} using \Cref{lem:output-values-ub} and \Cref{lem:output-values-lb}}
It is readily checked that if $T>\Phi(d,s,\ell)$ then the lower bound $\frac{1}{2d^2} \cdot  \log (T) - 3$ of \Cref{lem:output-values-lb} is greater than the upper bound $\ell^{ds}$ of \Cref{lem:output-values-ub}, and hence no $T$-sparse polynomial can have a distribution identical to that of any $s$-sparse polynomial.
\end{proofof}
\subsection{Proof of \Cref{lem:output-values-lb}}
First, we prove this lemma when $\bX$ is the uniform distribution on the set $\{0,1\}$. We then transfer this bound to an arbitrary $\bX$.
\begin{lemma} \label{lem:output-values-ub-Boolean}
Let $q(x)$ be a degree-at-most-$d$ multilinear polynomial with $\sparsity(q)=T.$  For $\bx \sim \{0,1\}^n$, the random variable $q(\bx)$ takes at least $\frac{1}{2d^2} \cdot \log (T) - 2$ distinct output values.
\end{lemma}
\begin{proof}
The lemma clearly holds if $T < 2^{2d^2}$, so we suppose that $T \geq 2^{2d^2}.$
Suppose that $q(y)$ takes on $L$ distinct outputs values $z_1,\dots,z_L$ for $y \in \zo^n$.  For $i \in [L]$ we define the polynomial
\[
q_i(y) := {\frac {\prod_{j \in [L], j \neq i}(q(y)-z_j)}{\prod_{j \in [L], j \neq i} (z_i - z_j)}},
\quad \text{and we observe that} \quad
q_i(y) = \begin{cases}
1 & \text{~if~}q(y)=z_i\\
0 & \text{~if~}q(y) \in \{z_1,\dots,z_L\} \setminus \{z_i\}.
\end{cases}
\]
With this definition of the polynomials $q_1(y),\dots,q_L(y),$ we have 
\begin{equation} \label{eq:pie}
\sum_{i=1}^L z_i q_i(x_1,\dots,x_n) = q(x_1,\dots,x_n) \quad \text{for all~}x \in \{0,1\}.
\end{equation}
Let $N$ be the number of variables on which $q$ depends, i.e., that appear in at least one monomial that occurs in $q$. Since $\sparsity(q)=T\geq 2^{2d^2}$ and $q$ is multilinear and of degree at most $d$, it is easy to see that $N \geq d$. Since the total number of distinct degree-at-most-$d$ monomials over $N \geq d$ variables is at most 
\[
1 + {N \choose 1} + \cdots + {N \choose d} \leq \bigg( \frac{eN}{d} \bigg)^d,\]
we get $T \le  (e \cdot N/d)^d$. Or equivalently, we get that $N \ge \frac{d}{e} \cdot T^{1/d}$. 

Using \Cref{eq:pie}, we have that  since $q$ depends on $N$ variables, there must be some $1 \le i \le L$ such that $q_i$ depends on $N_i$ variables where 
\begin{equation}~\label{eq:TNIBOUND}
    N_i \ge \frac{N}{L} \ge \frac{d}{Le}  \cdot T^{1/d}. 
\end{equation}

For this $i$, observe that the function $q_i$ is a $0/1$ valued function over $\{0,1\}^n$. On the other hand, as a polynomial, the degree of $q_i(y)$ is at most $d(L-1)$ (this follows from the definition of $q_i$ and the fact that $q$ has degree at most $d$). 

We now recall a result of Wellens~\cite{Wellens2022Relationships} (which itself builds on an improvement by Chiarelli, Hatami and Saks \cite{CHS20} on the classical result of Nisan and Szegedy \cite{NisanSzegedy:92}): 
\begin{fact}~\label{fact:Nisan}
Let $f:\{0,1\}^M \rightarrow \{0,1\}$ be a degree-$D$ polynomial. Then, $f$ can depend on at most $4.416 \cdot 2^D$ of the variables. 
\end{fact}

Using \Cref{fact:Nisan}, it follows that 
$$
N_i \le 4.416 \cdot 2^{d(L-1)}.
$$
Plugging this into  \eqref{eq:TNIBOUND}, we get that
\[
4.416 \cdot 2^{d(L-1)} \ge \frac{d}{Le} \cdot T^{1/d}. 
\]
Using the simple bound that $L \cdot 2^{d (L-1)} \le 2^{2dL}$ and $d \ge 1$, it easily follows that 
\[
L \ge \frac{1}{2d^2} \cdot \log T - 2\,. \qedhere  
\]
\end{proof}

We turn to the proof of \Cref{lem:output-values-lb}:

\begin{proofof}{\Cref{lem:output-values-lb}}
Let the support points of $\bX$ include $\alpha_1$ and $\alpha_2$ (where $\alpha_1 \not=\alpha_2$). Now, first note that to prove \Cref{lem:output-values-lb}, it suffices to prove this for the case when $\bX$ is just supported on $\{\alpha_1, \alpha_2\}$. Now, define the polynomial $q':\{0,1\}^n \rightarrow \mathbb{R}$ as
\[
q'(y_1, \ldots, y_n):= q(\alpha_1 + (\alpha_2-\alpha_1)y_1, \ldots, \alpha_1 + (\alpha_2-\alpha_1)y_n). 
\]
We now make a few easy but important observations about the polynomial  $q'$. 
\begin{enumerate}
    \item The evaluations of $q'$ on the Boolean hypercube $\zo^n$   are in one-to-one correspondence with the evaluation of $q$ on $\bX^{\otimes n}$. 
    \item The degree of $q'$ is at most $d$ (using the fact that the degree of $q$ is at most $d$). 
    \item We can also express the relation between $q$ and $q'$ as 
    \[
    q(x_1,\ldots, x_n) = q' \bigg( \frac{x_1-\alpha_1}{\alpha_2-\alpha_1}, \ldots, \frac{x_n-\alpha_n}{\alpha_2-\alpha_1}\bigg). 
    \]
    Since the degree of $q'$ is at most $d$, if $\sparsity(q') = T'$, then $\sparsity(q) \le 2^d \cdot \sparsity(q') = 2^d \cdot T'$. This implies $T' \ge 2^{-d} \cdot T$. 
\end{enumerate}
From items (2) and (3),  it follows by applying \Cref{lem:output-values-ub-Boolean} that the number of distinct output values that $q'$ attains over $\{0,1\}^n$ is at least
\[
    \frac{1}{2d^2} \cdot \log T'-2 \ge \frac{1}{2d^2} \cdot(\log (T) - d) - 2 \geq {\frac 1 {2d^2}} \cdot \log(T) - 3.
\]
By item (1), the proof of \Cref{lem:output-values-lb} is complete.
\end{proofof}

%% file: sections/alg-to-compute-MSG.tex
\section{Proof of \Cref{thm:alg-for-f}: An algorithm to compute $\MSG_{\bX,d}(s)$} \label{sec:alg-to-compute-MSG}

We begin by recalling \Cref{thm:alg-for-f}, which we prove in this section:

\algforf*

We remark that the theorem is easily seen to hold if $\ell=1$, so in the rest of this section we assume $\ell \geq 2$.

\subsection{Overview of the argument} \label{sec:overview-MSG-alg}
The high-level idea behind \Cref{thm:alg-for-f} is as follows:  \Cref{thm:f-is-well-defined} implies that $\MSG_{\bX,d}(s)$ is some integer in the interval $[s,\Phi(s,d,\ell)]$ (recall that $\Phi(s,d,\ell)$ is defined in \Cref{thm:f-is-well-defined}).  
Working backwards from $t=\Phi(s,d,\ell), t=\Phi(s,d,\ell)-1, \dots$ in order, we check, for each candidate value of $t$, whether there exists a pair of degree-at-most-$d$ multilinear polynomials $p,q$ with $\sparsity(p)=s,$ $\sparsity(q)=t$ such that the distributions of $p(\bX^{\otimes k})$ and $q(\bX^{\otimes k})$ are identical.  The first time we find such a value of $t$, that  is the correct value of $\MSG_{\bX,d}(s)$.

The idea behind how we perform the above check for the existence of $p,q$ with identical distributions under $\bX^{\otimes k}$ is as follows. Since for any fixed pair of candidate polynomials $p',q'$ the distributions $p'(\bX^{\otimes k})$ and $q'(\bX^{\otimes k})$ have finite support, a standard fact (see \Cref{claim:moments} below) implies that these distributions are identical if and only if they match finitely many moments.  So our goal is to determine whether there exists a degree-at-most-$d$ multilinear polynomial $p$ with $\sparsity(p)=s$ and a degree-at-most-$d$ multilinear polynomial $q$ with $\sparsity(q)=t$ with a certain number $J$ of matching moments.  To determine this, we try all possible ``sparsity patterns'' (see \Cref{def:sparsity-pattern}) for $s$-sparse degree-at-most-$d$ multilinear polynomials $p$ and for $t$-sparse degree-at-most-$d$ multilinear polynomials $q$, and for each candidate pair of sparsity patterns we check whether there exist suitable vectors of $s$ coefficients for $p$ and $t$ coefficients for $q$ (note that these coefficient vectors, together with the sparsity patterns for $p$ and $q$, completely specify $p$ and $q$) that make the first $J$ moments match as required.  For a given pair of candidate sparsity patterns $(M_1,\dots,M_s)$ for $p$ and $(M'_1,\dots,M'_t)$ for $q$, the existence of such coefficients for $p$ and for $q$ is equivalent to the solvability of a system of $M$ polynomial inequalities in $s+t$ unknowns corresponding to those coefficients.  Using known algorithmic results for deciding the existential theory of the reals \cite{Ren:88}, we can determine whether such a system has a solution.

\subsection{Setup and necessary ingredients for the argument} \label{ap:setup}

We will use the following simple result on equivalence of finitely many moments implying equivalence of finitely supported random variables. While we believe that this is a well-known folklore result (see e.g. \cite{stackexchange-moments}), for the sake of completeness we give the simple proof below.

\begin{claim} \label{claim:moments}
Let $\bX,\bY$ be two real-valued random variables each of which is supported on at most $k$ distinct values.  If $m_j(\bX)=m_j(\bY)$ for all $j=1,2,\dots,2k-1$ then $\bX$ and $\bY$ are identical (and if $m_j(\bX) \neq m_j(\bY)$ for some $j \in \{1,\dots,2k-1\}$ then $\bX$ and $\bY$ are not identical).
\end{claim}
\begin{proof}
The parenthesized second claim is obvious so it suffices to prove the first claim, that moment-matching implies identity of distributions.  

Let $Z \subset \R$ be the set containing all the support points of either $\bX$ or $\bY$, so $|Z| \leq 2k$.
The $j$-th raw moment of $\bX$ is $m_j(\bX) : = \sum_{z \in Z} p_{\bX}(z) z^j$ where $p_{\bX}(z)$ is the probability mass that $\bX$ puts on $a$, and likewise
the $j$-th raw moment of $\bY$ is $m_j(\bY) : = \sum_{z \in Z} p_{\bY}(z) z^j$.
By the matching-moments condition on $\bX$ and $\bY$, we have that 
\begin{equation} \label{eq:overdetermined}
\text{for~}j =0,\dots,2k-1,
\quad \quad
m_j(\bX)=m_j(\bY).
\end{equation}
But for any set $Z$ of $|Z| \leq 2k$ distinct support points, the linear system in indeterminates $(p_{\bX}(z) - p_{\bY}(z))_{z \in Z}$ given by \Cref{eq:overdetermined}  has as its coefficient matrix a $(2k) \times |Z|$ transposed Vandermonde matrix; since such a matrix has rank $|Z|$, the only solution to the linear system is that $p_{\bX}(z)-p_{\bY}(z) =0$ for all $z \in Z$. So we must have that $\bX$ and $\bY$ are identical.
\end{proof}

We will use the following definition of a ``sparsity pattern'':

\begin{definition} [Sparsity pattern] \label{def:sparsity-pattern}
For a given degree bound $d$ and sparsity parameter $r$, a \emph{$(d,r)$-sparsity pattern} is a list $\overline{M} = (M_1,\dots,M_r)$ of $r$ many distinct multilinear monomials, each of degree at most $d$, where each $M_i$ is a monomial over variables belonging to $x_1,\dots,x_{dr}.$
\end{definition}

The following facts are immediate:

\begin{fact} \label{fact:num-sparsity-patterns}
There are at most ${dr \choose \leq d}^r \leq (dr)^{O(dr)}$ many distinct $(d,r)$-sparsity patterns.
\end{fact}

\begin{fact} \label{fact:poly-pattern}
Given any degree-at-most-$d$ real multilinear polynomial $p$ with $\sparsity(p)=r$, there exists some $(d,r)$-sparsity pattern $\overline{M} = (M_1,\dots,M_r)$ and vector of nonzero coefficients $c_1,\dots,c_r > 0 $ such that up to a renaming of variables, $p$ is equal to $c_1 M_1 + \cdots + c_r M_r.$
\end{fact}

Finally, we recall the following fundamental result on the complexity of the existential theory of the reals:\begin{theorem}\label{thm:first-order}\cite{Ren:88}
 There is an algorithm $\mathcal{A}_{\mathrm{Ren}}$ which, given two lists of real polynomials $P_1, \ldots, P_m : \mathbb{R}^n \rightarrow \mathbb{R}$ and $Q_1, \ldots, Q_k : \mathbb{R}^n \rightarrow \mathbb{R}$ with
 rational coefficients,  decides whether there exists an $x \in \mathbb{R}^n $ such that
 \begin{itemize}
 \item $\forall i \in [m]$, $P_ i(x) \ge 0$, and
 \item $\forall i \in [k]$, $Q_i (x) >0$.
 \end{itemize}
 If the bit length of all coefficients in all polynomials is at most $L$ and the maximum degree of
 any polynomial is at most $d$, then the running time of $\mathcal{A}_{\mathrm{Ren}}$ is $L^{O(1)} \cdot ((m+k)\cdot d)^{O(n)}$.
 \end{theorem}

\subsection{The {\tt Compute-Max-Sparsity-Gap} algorithm and the proof of \Cref{thm:alg-for-f}}

The {\tt Compute-Max-Sparsity-Gap} algorithm is given in \Cref{alg:Compute-f}; we refer the reader to the second paragraph of \Cref{sec:overview-MSG-alg} for an intuitive overview of how it works.

\begin{algorithm}
\addtolength\linewidth{-2em}

\vspace{0.5em}

\textbf{Input:} Description of distribution $\bX$ putting weight $\alpha_i>0$ on $v_i \in \Q$ for $i \in [\ell]$;  degree bound $d \geq 1$; sparsity parameter $s \geq 1.$\\ [0.25em]
\textbf{Output:} The value of $\MSG_{\bX,d}(s)$.

\

{\tt Compute-Max-Sparsity-Gap}$(\bX,d,s)$:
	\begin{enumerate}
    \item For $t = \Phi(s,d,\ell),\Phi(s,d,\ell)-1,\dots,s$:
		\begin{enumerate}
			\item For each $(d,s)$-sparsity pattern $\overline{M}=(M_1,\dots,M_s)$ and each $(d,t)$-sparsity pattern $\overline{M}'=(M'_1,\dots,M'_t)$:
			\begin{enumerate}
			 	\item Let $P_1, P_2, \dots: \R^{s+t} \to \R$ and $Q_1,Q_2, \dots:\R^{s+t} \to \R$ be the two lists of real polynomials given by  \Cref{claim:polynomial-system}.
				\item Run the algorithm $\mathcal{A}_{\mathrm{Ren}}$ on $P_1,P_2,\dots$ and $Q_1,Q_2,\dots$.  If it returns ``there exists a solution in $\R^{s+t}$'' then halt and return the value $t$.
				Otherwise, proceed to the next pair of sparsity patterns in step~(a).
			\end{enumerate}
			\item If all pairs of  $(d,s)$-sparsity patterns and  $(d,t)$-sparsity patterns have been tried, proceed to the next value of $t$ in step~(1).
		\end{enumerate}

\end{enumerate}
\caption{The {\tt Compute-Max-Sparsity-Gap} algorithm.}
\label{alg:Compute-f}
\end{algorithm}

The following claim is at the crux of the correctness argument for \Cref{alg:Compute-f}:

\begin{claim} \label{claim:polynomial-system}
Fix $\overline{M}$ to be any $(d,s)$-sparsity pattern and
$\overline{M}'$ to be any $(d,t)$-sparsity pattern for some $t\geq s$.
The condition 
\begin{quote}
(*) ``There exists an assignment of nonzero coefficients $c_1,\dots,c_s > 0$ and an assignment of nonzero coefficients $c'_1,\dots,c'_t>0$ such that, taking $p = c_1 M_1 + \cdots + c_s M_s$ and $q = c'_1 M'_1 + \cdots + c'_t M'_t$, we have  $m_j(p(\bX^{\otimes ds})) = m_j(q(\bX^{\otimes dt}))$ for all $j \in [2\ell^{dt}-1]$''
\end{quote}
is true if and only if a certain system of $2 \cdot (2\ell^{dt}-1)$ polynomial inequalities of the form $P(c,c') \geq 0$ and $s+t$ polynomial inequalities of the form $Q(c,c')>0$ has a solution $(c_1,\dots,c_s,c'_1,\dots,c'_t) \in \R^{s+t}.$
Moreover, in each of these polynomials the bit length of each coefficient is at most $(dL\ell^{dt})^{O(1)}$ and the maximum degree of any of these polynomials is at most $2\ell^{dt}$.
\end{claim}
\begin{proof}
The $s+t$ polynomial inequalities of the form $Q(c,c')>0$ are simply $(c_1)^2>0,\dots,(c_s)^2>0,(c'_1)^2>0,\dots,(c'_t)^2>0.$ 
Each moment equality condition $m_j(p(\bX^{\otimes ds})) = m_j(q(\bX^{\otimes dt}))$, for $1 \leq j \leq \ell^{dt}-1$, is easily seen to correspond to a pair of polynomial inequalities 
\begin{equation} \label{eq:geqleq}
m_j(p(\bX^{\otimes ds})) - m_j(q(\bX^{\otimes dt})) \geq 0, \quad
m_j(p(\bX^{\otimes ds})) - m_j(q(\bX^{\otimes dt})) \leq 0
\end{equation}
in the indeterminates $c_1,\dots,c_s,c'_1,\dots,c'_t.$
Each of the inequalities in \Cref{eq:geqleq} is easily seen to have degree at most $2\ell^{dt}$.  Using the fact that each $\alpha_i,v_i$ is an $L$-bit rational number, it is straightforward to verify that each coefficient in each polynomial has bit length at most $(dL\ell^{dt})^{O(1)}.$
\end{proof}

We first argue correctness:

\begin{lemma} \label{lem:Compute-f-correctness}
The {\tt Compute-Max-Sparsity-Gap}$(\bX,d,s)$ algorithm correctly outputs the value of $\MSG_{\bX,d}(s)$.
\end{lemma}
\begin{proof}
Suppose first that $t > \MSG_{\bX,d}(s)$. 
By \Cref{def:max-sparsity}, for any $(d,s)$-sparsity pattern $\overline{M}=(M_1,\dots,M_s)$, any $(d,t)$-sparsity pattern 
$\overline{M}'=(M'_1,\dots,M'_t)$, any vector of nonzero coefficients $(c_1,\dots,c_s)$ and any vector of nonzero coefficients $(c'_1,\dots,c'_t)$, the polynomials $p(x) = c_1 M_1 + \cdots + c_s M_s$ and $q(x) = c'_1 M'_1 + \cdots + c'_t M'_t$ must not satisfy $p(\bX^{\otimes ds}) \equiv q(\bX^{\otimes dt}).$
Since $\bX$ is supported on $\ell$ values, the distribution of $p(\bX^{\otimes ds})$ is supported on at most $\ell^{ds}\leq \ell^{dt}$ real values, and likewise the distribution of $q(\bX^{\otimes dt})$ is supported on at most $\ell^{dt}$ real values.
Hence, by \Cref{claim:moments} applied to the two random variables $p(\bX^{\otimes dt})$ and $q(\bX^{\otimes dt})$, we have that $m_j(p(\bX^{\otimes ds})) \neq m_j(q(\bX^{\otimes dt}))$ for some $j \in [2\ell^{dt}-1].$
By \Cref{claim:polynomial-system}, the polynomial system given in \Cref{claim:polynomial-system} has no solution.
Hence in Step~1.(a).i the {\tt Compute-Max-Sparsity-Gap} algorithm will not halt and return.
It follows that the {\tt Compute-Max-Sparsity-Gap}$(\bX,d,s)$ algorithm does not halt at any value $t>\MSG_{\bX,d}(s)$.

Now consider what happens when $t$ reaches the correct value $\MSG_{\bX,d}(s)$.  By \Cref{def:max-sparsity}, there exist $p(x_1,\dots,x_{ds})$ and $q(x_1,\dots,x_{dt})$ such that $\sparsity(p)=s,$ $\sparsity(q)=t$ and $p(\bX^{\otimes ds}) \equiv q(\bX^{\otimes dt})$.
Let $\overline{M} = (M_1,\dots,M_s)$ be a $(d,s)$-sparsity pattern and let $c_1,\dots,c_s \neq 0$ be such that $p(x) = c_1 M_1 + \cdots + c_s  M_s$ (the existence of which is guaranteed by \Cref{fact:poly-pattern}), and similarly let $\overline{M}' = (M'_1,\dots,M'_t)$ be a $(d,t)$-sparsity pattern and let $c'_1,\dots,c'_t \neq 0$ be such that $q(x) = c'_1 M'_1 + \cdots + c'_t M_t$.
By the easy direction of \Cref{claim:polynomial-system} , the moments $m_{j}(p(\bX^{\otimes ds}))$ and $m_j(q(\bX^{\otimes dt}))$ are equal for all $j = 1,\dots,2\ell^{dt}-1$.
By \Cref{claim:polynomial-system}, the polynomial system given in \Cref{claim:polynomial-system} has a solution.
Hence when the {\tt Compute-Max-Sparsity-Gap} algorithm reaches the pair of sparsity patterns $\overline{M},\overline{M}'$, it will halt and return the value $t=\MSG_{\bX,d}(s)$ as desired.
\end{proof}

It remains to argue efficiency:

\begin{lemma} \label{lem:f-efficiency}
For $\bX$ as described in \Cref{thm:alg-for-f}, the {\tt Compute-Max-Sparsity-Gap}$(\bX,d,s)$ algorithm runs in time $d^{O(dL\Phi(s,d,\ell)^2)}$.
\end{lemma}
\begin{proof}
The outer loop of Step~1 goes for at most $\Phi(s,d,\ell)$ iterations.  In each iteration, by \Cref{fact:num-sparsity-patterns} Step 1.(a) iterates over at most $(dt)^{O(dt)}$ many pairs, where each candidate pair consists of a  $(d,s)$-sparsity pattern and a $(d,t)$-sparsity pattern.
Recalling the parameters of \Cref{claim:polynomial-system}, in each execution of Step 1.(a).(ii), by \Cref{thm:first-order} the algorithm ${\cal A}_{\mathrm{Ren}}$ runs in time 
\[
(dL\ell^{dt})^{O(1)} \cdot \pbra{2 \ell ^{dt} \pbra{2 \cdot (2\ell^{dt} - 1) + s+t}}^{O(s+t)} =
L^{O(1)} \cdot \ell^{O(dt^2)} = 2^{O(dLt^2)},\]
where the last equality uses the fact that $\ell$ must be at most $2^L$.
Since $t \leq \Phi(s,d,\ell)$, we get that the overall running time of {\tt Compute-Max-Sparsity-Gap}$(\bX,d,s)$ is at most 
\[
(d\Phi(s,d,\ell))^{O(d\Phi(s,d,\ell))} \cdot 2^{O(dL\Phi(s,d,\ell)^2)}
=d^{O(dL\Phi(s,d,\ell)^2)},
\]
and the lemma is proved.
\end{proof}

This concludes the proof of \Cref{thm:alg-for-f}. \qed
%
%
%
%
%
%
%
%
%
%
%
%

%% file: sections/lower-bound-sharp.tex


\section{Proof of \Cref{thm:sharp-lower}: Lower bound on the polynomial sparsity testing problem}
\label{sec:lower-bound-sharp}

We recall \Cref{thm:sharp-lower}:

\sharplower*

In words, this theorem says that distinguishing $s$-sparse degree-$d$ multilinear polynomials from $\eps$-far-from-$\newLs$-sparse degree-$d$ multilinear polynomials is hard   (requiring $\Omega(\log n)$ labeled examples from $\bX^{\otimes n}$) if $\newLs < f_{\bX,d}$, even in the absence of noise.

The key to our proof of \Cref{thm:sharp-lower} is \Cref{thm:Murakami}. This theorem can be interpreted as follows:  Suppose that $p$ and $q$ are two multilinear polynomials such that $p(\bX^{\otimes k})$ and $q(\bX^{\otimes k})$ have exactly identical distributions.  
It is clear that no algorithm, even given infinitely many independent samples of $r(\bX^{\otimes k})$, can determine whether $r = p$ or $r=q$ (note that here the algorithm is given only the output of the polynomial and not the input).
\Cref{thm:Murakami} says that even if an algorithm is given \emph{labeled} samples $(\bx,r(\bx))_{\bx \sim \bX^{\otimes n}}$ (i.e. it is given the input as well as the output of the polynomial for each example), then it is still impossible to determine whether $r=p$ or $r=q$ (over some unknown subset of the variables) with fewer than $\Omega(\log n)$ samples.  

At a rough intuitive level, \Cref{thm:Murakami} is true because there are (at least) $n/k = \Omega(n)$ many disjoint ``blocks of input variables'' that $r$ might depend on, and this number is too large for an algorithm to rule out $p$ on all blocks from fewer than $\Omega(\log n)$ samples (and likewise for $q$).  The formal argument works by establishing upper bounds on total variation distance for suitable random variables that are defined in the proof.

\subsection{Polynomials with identical output distributions are hard to distinguish}

\begin{theorem} [Polynomials with identical output distributions cannot be distinguished with $o(\log n)$ labeled examples]
 \label{thm:Murakami}
Let $\bX$ be a real random variable, with finite support $S=\{v_1,\dots,v_\ell\}$, which takes value $v_i$ with probability $\alpha_i>0$ for $i=1,\dots,\ell$.  Let $\alpha = \min_{i \in [\ell]} \alpha_i$.
Let $p(x_1,\dots,x_k)$ and $q(x_1,\dots,x_k)$ be two functions for which the distributions $p(\bX^{\otimes k})$ and $q(\bX^{\otimes k})$ are identical.  (Note that $k,\ell$, and the $\alpha_i$'s, and hence $\alpha$, are all independent of $n$.)

Then given access to independent labeled examples
\[
(\bx,r(\bx))_{\bx \sim \bX^{\otimes n}} \text{~(note  that there is no noise)},
\]
any algorithm that correctly (with high confidence) distinguishes between the two possibilities that
\begin{itemize}
\item $r(x)=p(x_{i_1},\dots,x_{i_k})$ for some $i_1,\dots,i_k \in [n]$, versus
\item $r(x)=q(x_{i_1},\dots,x_{i_k})$ for some $i_1,\dots,i_k \in [n]$,
\end{itemize}
must use  $m=\Omega_{k,\ell,\alpha}(\log n)$ samples.
\end{theorem}


\subsubsection{Proof of \Cref{thm:Murakami}}
 
 For simplicity we assume that $k$ divides $n$.
 We fix 
 \[
 m := c_{k,\ell,\alpha} \cdot \log n
 \]
 where we will state constraints on the small positive constant $c_{k,\ell,\alpha}$ later.
 
We consider two different distributions for the target function $r$.
The first distribution, denoted ${\cal D}_{p}$, is uniform over the $n/k$ polynomials
\[
r(x) = p(x_{(j-1)k+1},\dots,x_{jk})
\text{~as $j$ varies over $1,\dots,n/k$}.
\]
The second distribution, denoted ${\cal D}_{q}$, is uniform over the $n/k$ polynomials
\[
r(x) = q(x_{(j-1)k+1},\dots,x_{jk})
\text{~as $j$ varies over $1,\dots,n/k$}.
\]
Let ${\cal O}$ (for Output) denote the distribution $p(\bX^{\otimes k})$.  (Note that by our assumption on $p$ and $q$, we could equivalently have said that ${\cal O}$ is the distribution $q(\bX^{\otimes k}).$). If ${\cal O}$ puts all its mass on one value then \Cref{thm:Murakami} clearly holds, so we henceforth suppose that ${\cal O}$ puts nonzero mass on at least two distinct values.

It will be useful to have some notation for the distribution ${\cal O}$; so let us say that ${\cal O}$ puts probability weight $\tau_i$ on output value $w_i$, for $i=1,\dots,L$ (so $\tau_1 + \cdots + \tau_L = 1$). Note that we have that $2 \leq L \leq \ell^k$ (hence $L=O_{k,\ell}(1)$) and moreover for each $i \in [L]$ we have that $\alpha^k \leq \tau_i \leq 1-\alpha^k$ (hence each $\tau_i = \Omega_{k,\ell,\alpha}(1)$).

The analysis will also use the following distribution over $(\ba,\bb)$ pairs in $\R^n \times \R$, which we denote by ${\cal P}$: in a draw from ${\cal P}$, the vector $\ba=(\ba_1,\dots,\ba_n)$ is drawn from $\bX^{\otimes n}$, and the real value $\bb$ is independently drawn from ${\cal O}$.  
 Let $
(\overline{\ba},\overline{\bb}) := ((\ba^{(1)},\bb^{(1)}),\dots,(\ba^{(m)},\bb^{(m)}))$ be a sequence of $m$ pairs drawn independently from ${\cal P}.$

We will prove the following claim:

\begin{claim} \label{claim:mis2} Let $\star$ be a fixed element of $\{p,q\}.$ Let $\br$ be drawn from ${\cal P}_{\star}$ and let $(\overline{\bx},\overline{\by}) := ((\bx^{(1)},\by^{(1)}),$
$\dots,(\bx^{(m)},\by^{(m)}))$ be a sequence of $m$ examples independently generated as follows for each $t \in [m]$: 
$\bx^{(t)} \sim \bX^{\otimes n}$, and $\by^{(t)} = \br(\bx^{(t)}).$   Then the total variation distance between the distributions of $(\overline{\ba},\overline{\bb})$ and $(\overline{\bx},\overline{\by})$ is at most $O(1/n^{0.1})$.
\end{claim}

By the triangle inequality, an immediate consequence of \Cref{claim:mis2} is that the variation distance between $((\bx^{(1)},\bq(\bx^{(1)})),$
$\dots,(\bx^{(m)},\bq(\bx^{(m)})))$ and $((\bx^{(1)},\bp(\bx^{(1)})),$
$\dots,(\bx^{(m)},\bp(\bx^{(m)})))$ (where $\bq \sim {\cal D}_q$ and $\bp \sim {\cal D}_p$) is at most $O(1/n^{0.1})$, which suffices to give \Cref{thm:Murakami}.

It suffices to prove \Cref{claim:mis2} for $\star=p$ (the proof for $\star=q$ is identical).
Since $\overline{\ba}=(\ba^{(1)},\dots,\ba^{(m)})$ and $\overline{\bx}=(\bx^{(1)},\dots,\bx^{(m)})$ are both distributed according to $(\bX^{\otimes n})^m$, \Cref{claim:mis2} follows directly from the following:

\begin{claim} \label{claim:mis2-2}
With probability at least $1 - n^{-99.9}$ over a draw of $\overline{\bx}$ from $(\bX^{\otimes n})^m$, the distribution of $\overline{\by}=(\bp(\bx^{(1)}),\dots,\bp(\bx^{(m)}))$ conditioned on $\overline{\bx}$ (which we denote $\bp(\overline{\bx})$) has variation distance at most $O(1/n^{0.1})$ from the distribution of $\overline{\bb}$ (which is simply ${\cal O}^m$).
\end{claim}
\begin{proofof}{\Cref{claim:mis2-2}}
Fix any $\overline{i} := (i_1,\dots,i_m)\in [L]^m$, so 
\[
\Pr[\overline{\bb}=(w_{i_1},\dots,w_{i_m})]=\tau_{i_1} \cdots \tau_{i_m} =: \tau_{\overline{i}}.
\]
We observe that $n^{-0.5} \leq \alpha^{km} \leq \tau_{\overline{i}} \leq (1 - \alpha^k)^m \leq 1-\alpha^{km} \leq 1-n^{-0.5}$, where the bound $n^{-0.5} \leq \alpha^{km}$ is by a suitable choice of the constant $c_{k,\ell,\alpha}$ in the definition of $m=c_{k,\ell,\alpha} \cdot \log n$.

Given $\overline{x}=(x^{(1)},\dots,x^{(m)}) \in (S^{\otimes n})^m$, let us write $n_{\overline{i}}(\overline{x})$ for the number of indices $j \in [n/k]$ such that $p(x^{(t)}_{(j-1)k+1},\dots,x^{(t)}_{jk}) = w_{i_t}$ for all $t \in [m]$.  For a random draw of $\overline{\bx} \sim (\bX^{\otimes n})^m$, the distribution of $n_{\overline{i}}(\overline{\bx})$ is precisely $\Bin(n/k,\tau_{\overline{i}}),$ which has mean 
\[
\tau_{\overline{i}} \cdot {\frac n k} =: \mu_{\overline{i}}.
\]
Note that by our lower bound on $\tau_{\overline{i}}$ we have $\mu_{\overline{i}} \geq  n^{0.5}/k$; hence, by a standard multiplicative Chernoff bound,\footnote{Specifically, we are using the bounds $\Pr[\bZ \geq (1+\delta)\mu] \leq \exp(-\delta^2 \mu/3)$ and
$\Pr[\bZ \leq (1-\delta)\mu] \leq \exp(-\delta^2 \mu/2)$, taking $\delta = C\sqrt{\log n}/\sqrt{\mu}$, where $\bZ$ is the binomial random variable with mean $\mu$.} for a suitable constant $C=C_{k,\ell,\alpha}$, we have that
\[
\Prx_{\overline{\bx} \sim (\bX^{\otimes n})^m}\left[ 
\left(1-C \sqrt{{\frac {\log n} {\mu_{\overline{i}}}}}\right)\mu_{\overline{i}}
\leq n_{\overline{i}}(\overline{\bx}) \leq \left(1-C \sqrt{{\frac {\log n} {\mu_{\overline{i}}}}}\right) \mu_{\overline{i}}
\right]
\geq 1-n^{-100}, \quad \text{i.e.}
\]
\begin{equation} \label{eq:good1}
\Prx_{\overline{\bx} \sim (\bX^{\otimes n})^m}\left[ 
\mu_{\overline{i}}-C \sqrt{ \mu_{\overline{i}} \log n} 
\leq n_{\overline{i}}(\overline{\bx}) \leq 
\mu_{\overline{i}}-C \sqrt{\mu_{\overline{i}} \log n}
\right]
\geq 1-n^{-100}.
\end{equation}
Let us say that an $\overline{x}=(x^{(1)},\dots,x^{(m)}) \in (S^{\otimes n})^m$ such that $\mu_{\overline{i}}-C \sqrt{ \mu_{\overline{i}} \log n} 
\leq n_{\overline{i}}(x) \leq 
\mu_{\overline{i}}+C \sqrt{\mu_{\overline{i}} \log n}$ for all $\overline{i} \in [L]^m$ is a \emph{good} $\overline{x}$.
By \Cref{eq:good1} we have that 
\[
\Pr_{\overline{\bx} \sim (\bX^{\otimes n})^m}[\overline{\bx} \text{~is good}] \geq 1 - L^m n^{-100} > 1 - n^{-99.9},
\]
again by a suitable choice of the constant $c_{k,\ell,\alpha}$ which makes $L^m \leq n^{0.1}.$

It remains to argue that if $\overline{x} \in (S^{\otimes n})^m$ is good, then $\dtv(\bp(\overline{x}),{\cal O})  \leq O(1/n^{0.1})$.
This is straightforward:  both distributions $\bp(\overline{x})$ and ${\cal O}$ are supported on $\{w_1,\dots,w_L\}^m$.  For any $\overline{i}=(i_1,\dots,i_m) \in [L]^m$, the distribution ${\cal O}$ puts weight $\tau_{\overline{i}}$ on $(w_{i_1},\dots,w_{i_m})$, whereas the weight that the distribution $\bp(\overline{x})$ puts on $(w_{i_1},\dots,w_{i_m})$ is 
\[
{\frac {n_{\overline{i}}(x)}{n/k}} \in
\left[
{\frac {\mu_{\overline{i}} - C \sqrt{ \mu_{\overline{i}} \log n} }{n/k}},
{\frac {\mu_{\overline{i}} + C \sqrt{ \mu_{\overline{i}} \log n} }{n/k}}
\right]=
\left[
\tau_{\overline{i}} -  \sqrt{{\frac {\tau_{\overline{i}} \log n} {n/k}}},
\tau_{\overline{i}} +  \sqrt{{\frac {\tau_{\overline{i}} \log n} {n/k}}}
\right].
\]
Let $\overline{i}^* \in [L]$ be the choice which maximizes $\tau^{\overline{i}}.$ 
Summing over all $\overline{i} \in [L]^m$, the total variation distance $\dtv(\bp(x),{\cal O})$ is at most
\[
L^m \sqrt{{\frac {\tau_{\overline{i}^*} \log n} {n/k}}}
\leq
L^m \sqrt{{\frac {\log n} {n/k}}};
\]
recalling that $L^m \leq n^{0.1}$, this is at most $n^{-0.1}$,
and this concludes the proof of \Cref{claim:mis2-2}, \Cref{claim:mis2}, and \Cref{thm:Murakami}.
\end{proofof}
 
\subsection{Proof of \Cref{thm:sharp-lower}}

With \Cref{thm:Murakami} in hand it is not difficult to prove \Cref{thm:sharp-lower}.
Let $\bX,d,s,$ and $\newLs$ be as stated in \Cref{thm:sharp-lower}, so $\newLs = \MSG_{\bX,d}(s)-j$ for some $j \geq 1.$ 
Let $p(x_1,\dots,x_k)$ and $q(x_1,\dots,x_k)$ be two polynomials as in \Cref{def:max-sparsity}, so $\sparsity(p)=s$, $\sparsity(q)=\MSG_{\bX,d}(s)$, and the distributions $p(\bX^{\otimes k})$ and $q(\bX^{\otimes k})$ are exactly identical.
We assume without loss of generality that $\|q\|_\coeff=1$.
Let $\eps^2$ be the sum of squares of the $j$ smallest-magnitude coefficients of $q(x)$.
Then by \Cref{def:coeff} we have that 
\[
\dist(q,\newLs\text{-sparse}) = \eps,
\]
so $q$ is $\eps$-far from every $\newLs$-sparse multilinear polynomial.  
By \Cref{thm:Murakami}, any algorithm that distinguishes between $p$ and $q$ must use $\Omega_{k,\ell,\alpha}(\log n)$ samples (even if there is no noise), so any algorithm for the $(\bX,d,\boldeta_0,s,\newLs,\eps)$ polynomial sparsity testing problem problem must use this many samples.  Since $k,\ell$ and $\alpha$ are determined by $\bX,d$ and $s$ the sample complexity lower bound is $\Omega_{\bX,d,s}(\log n)$, and the proof is complete.
\qed
%
%

%% file: sections/upper-bound-sharp.tex

%

\section{Proof of \Cref{thm:sharp-upper}: Upper bound on the polynomial sparsity testing problem}
\label{sec:upper-bound-sharp}

In this section we prove \Cref{thm:sharp-upper}, which we recall below:

\sharpupper*

Recall that the input to the $(\bX,d,\boldeta,s,\newLs,\eps)$ polynomial sparsity testing problem is a collection of i.i.d.~samples $(\bx \sim \bX^{\otimes n}, \by=p(\bx) + \boldeta)$, where $p$ is promised to be a multilinear degree-$d$ polynomial with $1/K \leq \|p\|_{\coeff} \leq K$; throughout this section we will reserve ``$p$'' to denote this $n$-variable polynomial.

To simplify our presentation by eliminating one parameter, for most of this section we will normalize and assume that the unknown polynomial $p$ has $\|p\|_\coeff=1$. In \Cref{sec:reduction} we explain why this assumption is without loss of generality.


\subsection{Useful notation}


Throughout this section, we will write $\Upsilon := \Upsilon_{K, \bX}(s,d,\eps)$ (cf.~\Cref{eq:upsilon}). 
To explain the main ideas of our proof, it will be useful to have the following notation for various sets of polynomials and their associated random variables when each of their input variables is i.i.d.~according to $\bX$.  As we explain at the start of \Cref{sec:poly-sparsity-ub-overview}, $\newLs$ should be thought of as being strictly less than $\Upsilon$ in the following definitions.

\begin{definition}[Sets of polynomials and their associated random variables]
\label{def:poly-dist-space} 
~
    \begin{itemize}
    
        \item For $\ell \geq 1$, let $\calP_\ell$ be the set of all degree-at-most-$d$ multilinear real $\ell$-variable polynomials $q(x_1,\dots,x_\ell)$ which are $s$-sparse and have $\|q\|_{\coeff} = 1$.
        Let $\calP := \bigcup_{\ell \geq 1} \calP_\ell$.
        
        \item For $\ell \geq 1$, let $\calP_\ell(\eps)$ be the set of all degree-at-most-$d$ multilinear real $\ell$-variable polynomials $q(x_1,\dots,x_\ell)$ which are $\Upsilon$-sparse, are $\eps$-far-from-$\newLs$-sparse, and have $\|q\|_{\coeff} = 1$.
        Let $\calP(\eps) := \bigcup_{\ell \geq 1} \calP_\ell(\eps)$.

        \item For $\ell \geq 1$ and $\eps'\geq 0$, let $\calP_\ell(\eps,\eps')$ be the set of all degree-at-most-$d$ multilinear real $\ell$-variable polynomials $q(x_1,\dots,x_\ell)$ which are $\eps$-far-from-$\newLs$-sparse, are $\eps'$-close-to-$\Upsilon$-sparse, and have $\|q\|_{\coeff} = 1$.
                Let $\calP(\eps,\eps') := \bigcup_{\ell \geq 1} \calP_\ell(\eps,\eps')$.
                (Notice that $\calP(\eps)\subseteq \calP(\eps,\eps')$, with equality when $\eps'=0$.)
    \end{itemize}
    
    We also define sets of random variables induced by each of the previous sets of polynomials as follows:  given any set of polynomials ${\cal S}$, we write $\RV({\cal S})$ to denote the set of all real random variables induced by the polynomials in ${\cal S}$ when their inputs are i.i.d.~drawn from $\bX$, i.e.~
    \[
    \RV({\cal S})=\{q(\bX^{\otimes \ell}): q(x_1,\dots,x_\ell) \in {\cal S}\}.
    \]
    We will be concerned with the sets of random variables $\RV(\calP), \RV(\calP(\eps))$, and $\RV(\calP(\eps,\eps'))$ corresponding to the sets of polynomials $\calP, \calP(\eps)$ and $\calP(\eps,\eps')$ defined above.
\end{definition}

\begin{remark}
A natural question at this point is why the objects defined in \Cref{def:poly-dist-space} involve a union over all $\ell \geq 1$, given that the polynomial $p$ that we are testing is an $n$-variable polynomial.  The answer, roughly speaking, is that we are aiming for a result that has no dependence on $n$.
\end{remark}

\begin{remark} \label{rem:pierre}
Elaborating on the previous remark, we observe that every polynomial $q \in \calP$ can only depend on at most $ds$ variables, and that every polynomial $q \in \calP(\eps)$ can only depend on at most $d\Upsilon$ variables; hence, up to renaming of variables, $\calP$ is in fact equal to $\calP_{ds}$ and $\calP(\eps)$ is equal to $ \calP_{d\Upsilon}(\eps)$.  However, this is not true for $\calP(\eps,\eps')$; there is no upper bound on the number of variables that a polynomial in $\calP(\eps,\eps')$ may depend on.  At a high level, much of the work in what follows is devoted to showing that each random variable in $\RV(\calP(\eps,\eps'))$ (for a suitably small choice of $\eps'$) is ``close enough'' to some random variable in $\RV(\calP(\eps))$.
Since the polynomials in $\calP(\eps)$ depend on a number of variables that is independent of $n$, intuitively, by working with $\RV(\calP(\eps))$ we can get an ``$n$-independent bound.'' 
\end{remark}

Finally, let $M$ be the largest magnitude of any element in the support of $\bX$.  The following simple observation, which follows immediately from Cauchy-Schwarz and the definition of $\calP(\eps)$, will be useful:
\begin{observation} \label{obs:max-output} 
For any polynomial $q \in \calP(\eps)$, the value $L := M^d\sqrt{\Upsilon}$ is an upper bound on the magnitude of any value in the support of $q(\bX^{\otimes d\Upsilon})$.
\end{observation} 

\subsection{Proof overview} \label{sec:poly-sparsity-ub-overview}

{\bf Algorithm overview.}  The {\tt Test-Sparsity} algorithm is given in \Cref{alg:test-sparsity}.  The algorithm has four main conceptual phases, which we discuss below, but first we remark that our argument involves reasoning about three different kinds of distances:  (i) coefficient distance between polynomials; (ii) Wasserstein distance between random variables (which will be random variables induced by polynomials); and (iii) ``moment distance'' between random variables (again, random variables induced by polynomials).
(See \Cref{sec:coeff-distance,sec:wass-moment} for the definitions of these distances.)
Our discussion and analysis involves some delicate interplay between these various distance notions.

The initial ``zeroth'' phase handles the case that $\newLs \geq \Upsilon$; in this case the algorithm simply runs the {\tt CoarseTest}$(s,\eps)$ algorithm and halts. Note that if $\newLs\geq \Upsilon$, then by \Cref{thm:DFKO-informal},

\begin{itemize}

\item if $p$ is $s$-sparse then {\tt CoarseTest}$(s,\eps)$ accepts with high probability as desired, and

\item if $p$ is $\eps$-far from $\newLs$-sparse, then it is also $\eps$-far from $\Upsilon$-sparse and hence {\tt CoarseTest}$(s,\eps)$ rejects with high probability as desired.  So in the rest of our discussion we can suppose that $\newLs < \Upsilon.$

\end{itemize}

In the first phase, the algorithm calls a procedure named \EW\ (\Cref{alg:estimate-wasserstein}) that estimates the minimum Wasserstein distance between any two random variables $\bY \in \RV(\calP)$, $\bY' \in \RV(\calP(\eps))$, and uses this estimate to set various parameters.  We will return to this point later and say more about what is involved in this (see \Cref{rem:est-wass}), but for now we mention that it involves, among other things, choosing a new coefficient-distance parameter $\eps' \ll \eps$ that will be used in the second phase.

In the second phase, the algorithm performs a preliminary check by running the {\tt CoarseTest}$(s,\eps')$ algorithm and outputting ``\FFLS'' if that algorithm rejects.  The idea behind this phase is to handle the case that $p$ is $\eps$-far from $\newLs$-sparse and also $\eps'$-far from $\Upsilon$-sparse; observe that by \Cref{thm:DFKO-informal},

\begin{itemize}

\item if $p$ is $s$-sparse then {\tt CoarseTest}$(s,\eps')$ only rejects with small probability, and

\item if $p$ is $\eps$-far from $\newLs$-sparse and also $\eps'$-far from $\Upsilon$-sparse then {\tt CoarseTest}$(s,\eps')$ rejects with high probability as desired.

\end{itemize}

So for the rest of this overview, we assume that $\newLs < \Upsilon$ and that $p$ is either $s$-sparse, or else it is both $\eps$-far from $\newLs$-sparse and $\eps'$-close to $\Upsilon$-sparse.  Referring to Figure~1, this lets us assume that after Phase~2, the polynomial $p$ either belongs to the innermost region $\calP$, or to the two outermost rings (because if $p$ were outside the outermost ring, meaning that it is $\eps'$-far from $\Upsilon$-sparse, then {\tt CoarseTest}$(s,\eps')$ would have rejected with high probability in Phase~2).


\usetikzlibrary{decorations.text}
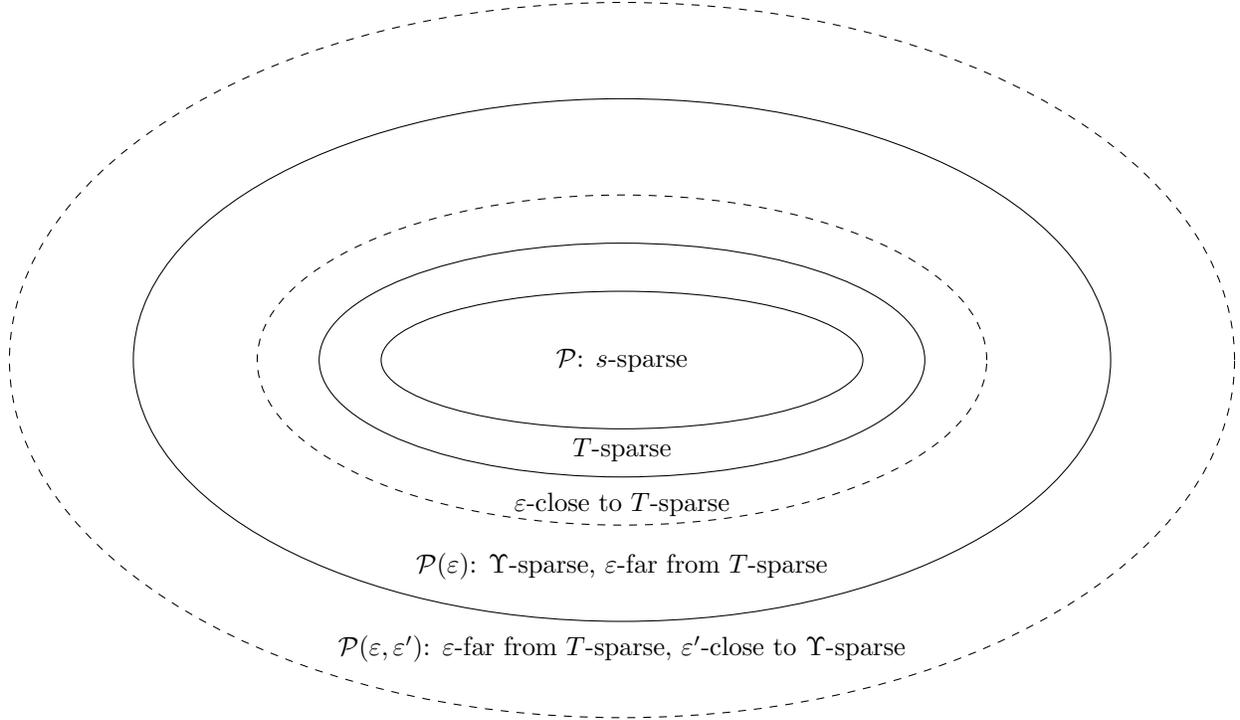
\begin{figure}
\centering
\begin{tikzpicture}[scale=0.915,every node/.style={font=\small}]

\def\n{5}
\def\xradius{3.5} 
\def\yradius{1.0} 
\def\spacing{0.9} 

\foreach \i in {0,1,4} {
    \draw
        (0,0) ellipse ({\xradius+\i*\spacing} and {\yradius+\i*0.7});
}

\foreach \i in {2,6} {
    \draw[dashed]
        (0,0) ellipse ({\xradius+\i*\spacing} and {\yradius+\i*0.7});
}

\node at (0,0) {$\calP$: $s$-sparse};
\node at (0,-1.3) {$T$-sparse};
\node at (0,-2.1) {$\eps$-close to $T$-sparse};
\node at (0,-3.0) {$\calP(\eps)$: $\Upsilon$-sparse, $\eps$-far from $T$-sparse};
\node at (0,-4.2) {$\calP(\eps,\eps')$: $\eps$-far from $T$-sparse, $\eps'$-close to $\Upsilon$-sparse};
\end{tikzpicture}	

\

\caption{Schematic diagram of the sets of polynomials we are concerned with.}
\label{fig:hiii}
\end{figure}

The third phase uses an algorithm called \CEND\  (\Cref{alg:construct-eps-net-dist}) to construct a finite net of real-valued random variables, $\RV^{(\zetaMom)}(\calP) \subset \RV(\calP)$ which has size independent of $n$. 
(As was suggested in \Cref{rem:pierre}, it is possible to do this because polynomials in $\calP$ depend only on constantly many variables independent of $n$.)
We will argue that if $q(x_1,\dots,x_{ds})$ is $s$-sparse, then there must be some random variable $\bY \in \RV^{(\zetaMom)}(\calP)$ 
which has moment-distance at most $\zetaMom$ to the random variable $q(\bX^{\otimes ds})$.

In the fourth and final phase, the algorithm estimates the first $k$ moments of $p(\bX^{\otimes n})$ and searches the net $\RV^{(\zetaMom)}(\calP)$ to try to find a random variable $\bY$ with small moment-distance from $p(\bX^{\otimes n}).$ If the net $\RV^{(\zetaMom)}(\calP)$ contains such a random variable $\bY$ then the algorithm accepts, outputting ``\SSparse,'' and otherwise it rejects, outputting ``\FFLS.'' 

\begin{remark} \label{rem:est-wass}
The reader who has looked ahead to the \CEND\ algorithm (\Cref{alg:construct-eps-net-dist}) will have observed that that procedure constructs \emph{two} finite nets of real-valued random variables, $\RV^{(\zetaMom)}(\calP) \subset \RV(\calP)$ and $\RV^{(\zetaMom)}(\calP(\eps)) \subset \RV(\calP(\eps))$; this may seem curious given that Phase~4 only involves the first net $\RV^{(\zetaMom)}(\calP)$.  The reason for this is that the \EW\ algorithm (\Cref{alg:estimate-wasserstein}) that is used in Phase~1 also calls the procedure \CEND, and \EW\ needs both a net $\RV^{(\zetaMom)}(\calP)$ for $ \RV(\calP)$ and a net $\RV^{(\zetaMom)}(\calP(\eps))$ for $\RV(\calP(\eps))$.
\end{remark}

\medskip
\noindent {\bf Sketch of correctness argument.}		
As indicated above, after the second phase of our algorithm it suffices to argue that for the polynomial $p(x_1,\dots,x_n)$ that is being tested, 

\begin{itemize}

\item (Completeness) If $p(x_1,\dots,x_n)$ is $s$-sparse, then w.h.p.~the {\tt TestSparsity} algorithm outputs ``\SSparse'', and

\item (Soundness) If $p(x_1,\dots,x_n)$ is $\eps$-far from $\newLs$-sparse (where $\newLs \leq \Upsilon$) and $\eps'$-close to $\Upsilon$-sparse, then w.h.p.~the {\tt TestSparsity} algorithm outputs ``\FFLS''.

\end{itemize}

Doing this requires a number of structural results, which we briefly survey in this paragraph and which are established in \Cref{sec:poly-dist-properties}. We first show,  using a compactness argument, that there is a nonzero lower bound on the Wasserstein distance between any two random variables $\bY \in \RV(\calP)$ and $\bZ \in \RV(\calP(\eps))$ (see \Cref{lem:wasserstein-compact}).
Another useful structural result established in this subsection shows that closeness between two polynomials in coefficient distance implies moment-distance closeness of the induced random variables 
(see \Cref{lem:close-coeff-moments}).
A result due to Valiant and Kong (restated as our \Cref{lem:original-kong-valiant}) gives us that moment-distance-closeness between random variables implies Wasserstein-distance closeness, or equivalently, if two random variables are Wasserstein-distance far from each other, then they must also be moment-distance far (see \Cref{cor:kong-valiant}).  Putting these results together, we can infer that any $q(\bX^{\otimes ds}) \in \RV(\calP)$ and $q'(\bX^{\otimes d\Upsilon}) \in \RV(\calP(\eps))$ must have moment distance lower bounded by a positive constant. \Cref{lem:moment-difference} lets us extend this to any $q(\bX^{\otimes ds}) \in \RV(\calP)$ and any $q'(\bX^{\otimes N}) \in \RV(\calP(\eps,\eps'))$.

We now discuss how these structural results relate to the various phases of our algorithm.  Phase~1 of our algorithm, which sets parameters, uses \Cref{lem:wasserstein-compact} to estimate the Wasserstein distance between $\RV(\calP)$ and $\RV(\calP(\eps))$.\footnote{We remark that like Phase~3, this estimation is carried out by building a suitable net of distributions; it does not require using any noisy samples $(\bx,p(\bx)+\ion)$ from the unknown polynomial $p$ that is being tested.}  
This estimate is used to define a suitable choice of the distance parameter $\eps'$ that is used for the {\tt CoarseTest}$(s,\eps')$ procedure in Phase~2. Intuitively, $\eps'$ is chosen to be extremely small relative to $\eps$ --- so small that a random variable $p(\bX^{\otimes \ell}) \in \RV(\calP(\eps,\eps'))$ will behave, in terms of moment-distance to random variables in 
$\RV(\calP)$, like a random variable from $\RV(\calP(\eps))$ (this is the upshot of \Cref{lem:moment-difference}).
Phase~3 of our algorithm, which constructs the net $\RV^{(\zetaMom)}(\calP)$ of random variables, uses \Cref{lem:close-coeff-moments} to establish that the constructed set of random variables is indeed a net, in terms of moment-distance of random variables.
Finally, to see that Phase~4 of our algorithm is indeed correct, we separately consider completeness and soundness:

\begin{itemize}

\item (Completeness) Suppose first that $p(x_1,\dots,x_n)$ is $s$-sparse. Then the net $\RV^{(\zeta)}(\calP)$ will indeed contain a random variable $\bY$ which is moment-distance close to $p(\bX^{\otimes n})$, and so we accept with high probability.

\item (Soundness) Next suppose that $p(x_1,\dots,x_n)$ is $\eps$-far from $\newLs$-sparse. As discussed earlier, we can assume in this case that $p$ is both $\eps$-far from $\newLs$-sparse and $\eps'$-close to $\Upsilon$-sparse, i.e.~$p \in \calP(\eps,\eps')$.
By 
\Cref{lem:moment-difference}, $p(\bX^{\otimes n})$ has ``large'' moment-distance from every random variable in $\RV(\calP)$, and hence from every random variable in the net $\RV^{(\zeta)}(\calP) \subset \RV(\calP)$, so the algorithm will output ``\FFLS'' as desired.

\end{itemize}

%
%
%
%

\subsection{Properties of $\calP,\calP(\eps),\calP(\eps,\eps')$ and the associated classes of induced random variables}
\label{sec:poly-dist-properties}

Recalling \Cref{def:max-sparsity}, by the assumption that $\newLs \geq \MSG_{\bX,d}(s)$ we have that for any $\eps > 0$,  the sets $\RV(\calP)$ and $\RV(\calP(\eps))$ are disjoint.
We begin by giving a universal lower bound (which depends on $\bX,d,s,T$, and $\eps,$ but not on $n$)
on the Wasserstein distance between any two random variables $\bY \in \RV(\calP)$ and $\bY' \in \RV(\calP(\eps))$:
\begin{lemma}\label{lem:wasserstein-compact}
    Given $\eps > 0$, there exists some value $c_\eps > 0$
    such that
    \begin{equation} \label{eq:infyinf}
        \inf_{\bY\in \RV(\calP), \bY'\in \RV(\calP(\eps))}W_1(\bY,\bY') = c_\eps.
    \end{equation}
\end{lemma}
\begin{proof}
    Recall, that by assumption, $\newLs \geq \MSG_{\bX,d}(s)$.
    Each polynomial $q'$ in $\calP(\eps)$ is $\eps$-far from $\newLs$-sparse, and thus is also $\eps$-far from $\MSG_{\bX,d}(s)$-sparse.
    Thus, by the definition of $\MSG_{\bX,d}(s)$, there is no $s$-sparse polynomial $q$ of degree at most $d$ such that the random variables $q(\bX^{\otimes n})$ and $q'(\bX^{\otimes n})$ are exactly identical.
    Therefore the sets $\RV(\calP)$ and $\RV(\calP(\eps))$ are disjoint.
    
    Given that $\RV(\calP) \cap \RV(\calP(\eps)) = \emptyset,$ to establish \Cref{eq:infyinf} it suffices to show that the two sets of random variables $\RV(\calP)$ and $\RV(\calP(\eps))$ are each compact (under the Wasserstein distance metric).
    To show that both $\RV(\calP)$ and $\RV(\calP(\eps))$ are compact, we (i) show that both $\calP = \calP_{ds}$ and $\calP(\eps)=\calP_{d\Upsilon}(\eps)$ are compact under the coefficient distance metric, and (ii) give continuous mappings $\psi_{ds}$ ($\psi_{d\Upsilon}$, resp.) on the space of all multilinear polynomials of degree at most $d$ over $ds$ many variables (over $d\Upsilon$ many variables, resp.) such that $\psi_{ds}(\calP)=\RV(\calP)$ ($\psi_{d\Upsilon}(\calP(\eps))=\RV(\calP(\eps))$, resp.). 
    
    For (i), we establish the compactness of $\calP(\eps)=\calP_{d\Upsilon}(\eps)$ (the compactness of $\calP=\calP_{ds}$ follows by an entirely similar and slightly simpler argument).  To do this, we construct a subset $V \subset \mathbb{R}^{d'}$ which is closed and bounded under the $\ell_2$-norm, and present a continuous mapping $\phi$ on the $d'$-dimensional vector space such that $\phi(V)=\calP(\eps)$, where $d'=\sum_{k=0}^d{d\Upsilon \choose k}$.

    Let $V\subset \mathbb{R}^{d'}$ be the set of $d'$-dimensional vectors with $\ell_2$ norm 1, $\ell_2$ distance to any $\newLs$-sparse vector at least $\eps$, and no more than $\Upsilon$ non-zero coordinates:
    \[
    V=\{v\in\mathbb{R}^{d'}: \|v\|_2=1,\|v\|_0\leq \Upsilon,\text{ and for all } \,u\in\mathbb{R}^{d'} \text{ with } \|u\|_0\leq \newLs,\|u-v\|_2\geq \eps \}.
    \]
    $V$ is closed and bounded under the $\ell_2$ norm, so it is also compact.
    Notice that there are exactly $d'$ many monomials on $d\Upsilon$ variables, including the degree-0 monomial, whose degree is at most $d$.
    Therefore, there exists a bijection $\phi_0$ from $[d']$ to the set of such monomials.
    We naturally define the mapping $\phi(v)=\sum_{i=1}^{d'}v_i\cdot\phi_0(i)$ for $v\in\mathbb{R}^{d'}$.
    It is then straightforward to verify that $\phi(V)=\calP(\eps)$.

    By construction, for any $u,v\in \mathbb{R}^{d'}$, $\|\phi(u)-\phi(v)\|_\coeff = \|u-v\|_2$.
    So, for any $u,v\in \mathbb{R}^{d'}$ and $\zeta > 0$ such that $\|u-v\|_2 < \zeta$, $\|\phi(u) - \phi(v)\| < \zeta$ as well, and thus $\phi$ is continuous.
    Therefore, $\calP(\eps)$ is compact under the coefficient distance metric, giving us (i).

    For (ii), as with (i), we give the argument for $\psi_{d\Upsilon}$ and the argument for $\psi_{ds}$ is similar.  
    Define the mapping $\psi_{d\Upsilon}$ from degree-$d$ polynomials to random variables in the following natural way:
    for any polynomial $q$, $\psi_{d\Upsilon}(q)$ is the random variable $q(\bX^{\otimes d\Upsilon})$.
    Taking expectations over $\bx \sim \bX^{\otimes n}$, for any two degree-$d$ multilinear polynomials $q_1,q_2$, we have
    \begin{align}
    W_1(\psi_{d\Upsilon}(q_1),\psi_{d\Upsilon}(q_2))&\leq \Ex\left[|q_1(\bx)-q_2(\bx)|\right] \tag{using the naive coupling}\\
        &=\Ex[|r(\bx)|] \tag{defining $r=q_1-q_2$}\\
        &\leq \Ex[r(\bx)^2]^{1/2} \tag{Cauchy-Schwarz}\\
        &= \left(\sum_{S\in{[n]\choose \leq d}}r_S^2\Ex\left[\prod_{i\in S}\bx_i^2\right]\right)^{1/2}
        \tag{using $\E[\bx_i]=0$ and multilinearity of $r$}\\
        &= \left(\sum_{S\in{[n]\choose \leq d}}r_S^2\right)^{1/2}
        \tag{using $\E[\bx_i^2]=1$}\\
        &= \|q_1-q_2\|_{\coeff}.
        \label{eq:wasserstein-coeff}
    \end{align}
    Thus, $\psi_{d\Upsilon}$ is Lipschitz continuous.
    So, $\RV(\calP(\eps))$, the image of $\calP(\eps)$ under $\psi_{d\Upsilon}$, is compact under the Wasserstein metric, since $\calP(\eps)$ is compact under coefficent distance.
    
    To conclude, given that  $\RV(\calP)$ and $\RV(\calP(\eps))$ are both compact and are  disjoint, it follows that $\inf_{\bY\in \RV(\calP), \bY'\in \RV(\calP(\eps))}[W_1(\bY,\bY')] =c_\eps$ for some positive value $c_\eps$, and the proof is complete.
\end{proof}

Next, we show, roughly speaking, that if $q_1(x_1,\dots,x_{ds})$ is $s$-sparse and $q_2(x_1,\dots,x_N)$ is such that\footnote{We suppose throughout that $ds \leq N$, which is clearly without loss of generality since we may view both $q_1$ and $q_2$ as defined over $\R^N$ even if either or both of them do not depend on some of those variables.} $\|q_1-q_2\|_\coeff$ is small, then the moment-distance between $q_1$ and $q_2$ is small. More precisely, we have the following:
\begin{lemma}\label{lem:close-coeff-moments}
Let $q_1(x_1,\dots,x_{ds})$ be a degree-$d$ multilinear polynomial with $\|q_1\|_\coeff=1$, and let $q_2(x_1,\dots,x_N)$ be a degree-$d$ multilinear polynomial  such that $\|q_1-q_2\|_\coeff \leq \eta \leq 1$.
    Then for any positive integer $k$, there exists a value 
    $\xi_{k,d,\bX}$ (depending only on $k$, $d$ and $\bX$) such that $\Mom_k(q_1(\bX^{\otimes ds}),q_2(\bX^{\otimes N})) \leq \xi_{k,d,\bX} \cdot \eta.$
\end{lemma}

\begin{proof}
We write $q_1(x) = \sum_{S \in {N \choose \leq d}} \widehat{q_1}(S) \prod_{i \in S} x_i$ and likewise for $q_2$, so 
\[
\|q_1-q_2\|_\coeff= \sqrt{\sum_{S \in {N \choose \leq d}} (\widehat{q_1}(S)-\widehat{q_2}(S))^2}
\leq \eta.
\]
Throughout the proof we write $\bx$ to denote a vector-valued random variable distributed according to $\bX^{\otimes N}$, and we sometimes simply write $r$ for the polynomial $r(\bx)$.

We handle the first and second moments separately from the others with the following simple arguments: for the first moment, since $\E[\prod_{i \in S} \bX_i]=0$ unless $S=\emptyset$, we have 
\[
\abs{m_1(q_1(\bx)) - m_1(q_2(\bx))} =
\abs{\widehat{q_1}(\emptyset) - \widehat{q_2}(\emptyset) } \leq \|q_1 - q_2\|_\coeff \leq \eta.
\]
For the second moment, as in the proof of \Cref{lem:wasserstein-compact} we have that for $i=1,2$, 
\[
m_2(q_i(\bx))=\sum_{S \in {N \choose \leq d}}\widehat{q_i}(S)^2 = \|q_i\|_\coeff^2.
\]
So $m_2(q_1)=1$ and by the triangle inequality $\|q_2\|_\coeff \leq 1 + \eta$, so $m_2(q_2(\bx)) \leq 1 + 2\eta + \eta^2$ and $|m_2(q_1(\bx))-m_2(q_2(\bx))| \leq 2\eta + \eta^2 \leq 3 \eta.$

For $3 \leq \ell \leq k$, we have
\begin{align}
\abs{m_\ell(q_1(\bx))-m_\ell(q_2(\bx))} &=
\abs{\Ex[q_1^\ell - q_2^\ell]} \label{eq:momdiff}\\
&=\abs{\Ex\sbra{(q_1 - q_2) \cdot \pbra{\sum_{i=0}^{\ell-1} q_1^{\ell-1-i}q_2^i}}} \nonumber\\
&\leq \overbrace{\pbra{\Ex\sbra{(q_1 - q_2)^2}}^{1/2}}^{=(A)} \cdot \overbrace{\pbra{\Ex\sbra{\pbra{\sum_{i=0}^{\ell-1} q_1^{\ell-1-i}q_2^i}^2}}^{1/2}}^{=(B)}. \tag{Cauchy-Schwarz}
\end{align}
We bound this as follows: for the first term, as in the proof of \Cref{lem:wasserstein-compact} we have
\[
(A)=\pbra{\Ex\sbra{(q_1 - q_2)^2}}^{1/2}=
\|q_1 - q_2\|_{\coeff} \leq \eta.
\]
For $(B)$, let us analyze a single monomial $\E[q_1^a q_2^b]$ among the $\ell^2$ monomials arising from the expansion of the square, and observe that each such monomial has $a+b=2(\ell-1)$. By Cauchy-Schwarz we have
$
\abs{\E[q_1^a q_2^b]} \leq \E[q_1^{2a}]^{1/2} \E[q_2^{2b}]^{1/2}$; we observe that for $i \in \{1,2\}$, if $c=0$ then $\E[q_i^{2c}]=1$, and if $c=1$ then as above, $\E[q_i^{2c}]= \|q_i\|_\coeff^2 \leq 1 + 3 \eta < 4$.  For the general case in which $c\geq 2$, we will use Theorem~10.21 of \cite{odonnell-book} (\Cref{thm:OD10.21}).
Applying this to the polynomial $r=q_i$ with $\ell=2c \geq 4$, since $\E[q_i^2]=\|q_i\|_\coeff^2 = m_2(q_i)\leq 1 + 2\eta + \eta^2  \leq 4,$ we have (recalling that $\ell \leq k$) that
\begin{align*}
\E[q_i^{2c}] &\leq \pbra{\sqrt{2c-1} \cdot \lambda^{1/(2c) - 1/2}}^{2dc}  \E[q_i^2]^{\ell/2}
\leq  C_{k,d,\bX}
\end{align*}
where $C_{k,d,\bX}$ is a constant depending only on $k,d$ and $\bX$ (whose value may change from line to line in what follows). Hence
$
\abs{\E[q_1^a q_2^b]}  \leq C_{k,d,\bX}$, and consequently we have
$
(B)
\leq C_{k,d,\bX}
$ as well. 
Putting the pieces together, we have that $(\ref{eq:momdiff})=C_{k,d,\bX} \cdot \eta$ for all $\ell=1,\dots,k$, and hence $\Mom_k(q_1(\bX^{\otimes ds}),q_2(\bX^{\otimes N})) \leq \zeta_{k,d,\bX} \cdot \eta$ as claimed.
\end{proof}


\subsection{Connecting Wasserstein distance and moment distance}

We will use a result of \cite{KongValiant17}, which shows that small ``moment-distance" between two random variables implies small Wasserstein distance:

\begin{lemma}[\cite{KongValiant17} Proposition 1]\label{lem:original-kong-valiant} 
    Let $\bY$ and $\bZ$ be two real random variables whose support is contained in $[-1,1]$.
    The Wasserstein distance $W_1(\bY, \bZ)$ is upper bounded in terms of the $k$-th order moment distance as follows:
    \[
    W_1(\bY, \bZ) \leq C/k+g(k) \cdot \Mom_k(\bY,\bZ),
    \]
    where $C$ is an absolute constant and $g(k) = C'3^k$ for an absolute constant $C'$.\qed
\end{lemma}

We mainly use an easy corollary of Lemma \ref{lem:original-kong-valiant}.
\begin{corollary}\label{cor:kong-valiant}
    Let $\bY,\bZ$ be two real random variables whose support is contained in $[-L,L]$ for some $L \geq 1$, and suppose the Wasserstein distance between them is
    $W_1(\bY,\bZ)=c.$ Taking the constants $C$ and $C'$ as in \Cref{lem:original-kong-valiant}, for $k\geq\lceil 2LC/c\rceil$, we have 
    \[
    \Mom_k(\bY,\bZ) \geq {\frac c {2C' \cdot 3^k}}.
    \]
\end{corollary}
\begin{proof}
    We rescale by defining $\bY'=\bY/L$, $\bZ'=\bZ/L$, so $\bY'$ and $\bZ'$ have support contained in $[-1,1]$ and have $W_1(\bY',\bZ') = c/L.$  
    Applying \Cref{lem:original-kong-valiant} to $\bY'$ and $\bZ',$ and observing that by our choice of $k$ we have $C/k < c/(2L),$ we see that
\begin{align*}
{\frac c {2L g(k)}} =
{\frac c {2LC' \cdot 3^k}}
&\leq \Mom_k(\bY',\bZ')\\
&= \sqrt{\sum_{i=1}^k (m_i(\bY')-m_i(\bZ'))^2}\\
&= \sqrt{\sum_{i=1}^k \pbra{{\frac {m_i(\bY)-m_i(\bZ)}{L^i}}}^2}\\
&\leq {\frac 1 L} \sqrt{\sum_{i=1}^k \pbra{m_i(\bY)-m_i(\bZ)}^2} \tag{since $L \geq 1$}\\
&\leq {\frac 1 L} \Mom_k(\bY,\bZ). \qedhere
\end{align*}
\end{proof}

\begin{lemma}
\label{lem:moment-difference}
Let $k \geq \lceil 2 L C/c_{\eps/2}\rceil$ where $L = M^d\sqrt{d\Upsilon}$ as in \Cref{obs:max-output}, $C$ is the constant from \Cref{lem:original-kong-valiant}, and $c_{\eps/2}$ is as defined in \Cref{lem:wasserstein-compact}. For $\eps' \leq {\frac {c_{\eps/2}}{4C' 3^k \xi_{k,d,\bX}}}$ (where $C'$ is the constant from \Cref{lem:original-kong-valiant} and $\xi_{k,d,\bX}$ is from \Cref{lem:close-coeff-moments}), for any random variables  $q(\bX^{\otimes ds}) \in \RV(\calP)$ and $q'(\bX^{\otimes N}) \in \RV(\calP(\eps,\eps'))$, we have that $\Mom_k(q(\bX^{\otimes ds}),q'(\bX^{\otimes N})) \geq {\frac {c_{\eps/2}}{4C' 3^k}}.$
\end{lemma}
\begin{proof}
    Fix any $q(\bX^{\otimes ds}) \in \RV(\calP)$ and any $q'(\bX^{\otimes N}) \in\RV(\calP(\eps,\eps'))$. By \Cref{lem:p-prime-perturbation} (stated and proved below), there exists a polynomial $q''\in\calP(\eps/2)$ such that $\|q-q''\|_\coeff\leq2\eps'$. By \Cref{lem:close-coeff-moments} applied to $q'$ and $q''$, we have that 
    \begin{equation} \label{eq:ham}
    \Mom_k(q'(\bX^{\otimes N}),q''(\bX^{\otimes d\Upsilon})) \leq 2\eps' \cdot \xi_{k,d,\bX}.
    \end{equation}
   By \Cref{lem:wasserstein-compact}, there exists $c_{\eps/2}>0$ such that $W_1(q(\bX^{\otimes ds}),q''(\bX^{\otimes d\Upsilon})) \geq c_{\eps/2}.$
   Recalling \Cref{obs:max-output}, we have that $L = M^d\sqrt{d\Upsilon}$ is an upper bound on the magnitude of any value in the support of either $q(\bX^{\otimes ds})$ or $q(\bX^{\otimes d\Upsilon})$.
Thus we may apply \Cref{cor:kong-valiant}, and we get that since $k \geq \lceil 2 L C/c_{\eps/2}\rceil,$ we have 
\begin{equation}
\label{eq:neggs}
\Mom_k(q(\bX^{ds}),q''(\bX^{d\Upsilon})) \geq {\frac {c_{\eps/2}}{2C' 3^k}}.
\end{equation}
By the triangle inequality applied to \Cref{eq:ham,eq:neggs}, since $\eps' \leq {\frac {c_{\eps/2}}{4C'3^k \xi_{k,d,\bX}}}$, we get the claimed lower bound on $\Mom_k(q(\bX^{\otimes ds}),q'(\bX^{\otimes N}))$.
\end{proof}

\begin{lemma}\label{lem:p-prime-perturbation}
    For any $0 \leq \eps' \leq \eps/2$ and any polynomial $q\in \calP(\eps,\eps')$, there exists a polynomial $q'\in\calP(\eps/2)$ such that $\|q-q'\|_\coeff \leq 2\eps'$.
\end{lemma}
\begin{proof}
    By definition of $\calP(\eps,\eps')$ and \Cref{lem:first-s-sum-of-squares}, the sum of squares of $q$'s $\newLs$ coefficients with the largest magnitude is at most $1-\eps^2$, and the sum of squares of $q$'s $\Upsilon$ coefficients with the largest magnitude is at least $1-\eps'^2$. We construct a polynomial $q'$ from $q$ by (i) keeping only the terms in $q$ that are among the $\Upsilon$ largest coefficients by magnitude, and (ii) rescaling so that the coefficient norm of $q'$ equals 1. Then the sum of squares of the $\newLs$ largest-magnitude coefficients of $q'$ is at most ${\frac {1-\eps^2}{1-\eps'^2}} \leq 1-\eps^2/4$, which implies that $q'\in\calP(\eps/2)$. Finally, we have
    \[
    \|q-q'\|_\coeff^2 \leq \eps'^2 + \eps'^2 \quad \text{and hence} \quad
    \|q-q'\|_\coeff \leq \sqrt{2} \cdot \eps' < 2 \eps',
    \]
    where the first $\eps'^2$ in the sum comes from the coefficients of $q$ whose magnitude is $\ell$-th largest as $\ell$ ranges over $\{\Upsilon+1,\dots\}$, and the second $\eps'^2$ comes from the rescaling of the coefficients of $q$ that are kept.
\end{proof}

\subsection{Constructing nets for sets of random variables}

The procedure \CEND\ that constructs suitable ``nets'' for the spaces $\RV(\calP)$ and $\RV(\calP(\eps))$ of random variables induced by polynomials in $\calP$ and $\calP(\eps)$ is given in \Cref{alg:construct-eps-net-dist}.
The main result we prove about this procedure is the following:
\begin{lemma}\label{lem:poly-space-eps-net}
There is an algorithm \CEND\ (\Cref{alg:construct-eps-net-dist}) which takes in as input $\bX$, degree bound $d,$ sparsity parameters $s,\newLs,$ tolerance $\eps>0$, a granularity parameter $0 < \zetaMom < 1$, and an integer $k$. The algorithm outputs two finite sets, $\RV^{(\zetaMom)}(\calP)\subset \RV(\calP)$ and $\RV^{(\zetaMom)}(\calP(\eps))\subset \RV(\calP(\eps))$, of $O_{\bX,d,s,\newLs,\eps,\zetaMom,k}(1)$ many random variables each, 
    such that for any polynomial $q \in \calP$ ($q' \in \calP(\eps)$, resp.)
    there exists a random variable  $\bY \in \RV^{(\zetaMom)}(\calP)$ 
    ($\bY' \in \RV^{(\zetaMom)}(\calP(\eps))$, resp.) such that
    $\Mom_k(q(\bX^{\otimes ds}),\bY) \leq \zetaMom$
    ($\Mom_k(q'(\bX^{\otimes d\Upsilon}),\bY') \leq \zetaMom$, resp.).
\end{lemma}

To prove \Cref{lem:poly-space-eps-net}, it will be useful to have the following preliminary result, which establishes that the sets $\calP^{(\zetacoeff)}$ and $\calP(\eps)^{(\zetacoeff)}$ that are constructed in step~2 of the algorithm \CEND\ (by the call to \CENP) are indeed nets (with respect to coefficient-distance) for the spaces of polynomials $\calP$ and $\calP(\eps)$:

\begin{claim} \label{claim:net-of-polys}
For every polynomial $q(x_1,\dots,x_{ds}) \in \calP$ ($q'(x_1,\dots,x_{d\Upsilon}) \in \calP(\eps)$, resp.) there is a polynomial $r \in \calP^{(\zetacoeff)} \subset \calP$ ($r' \in \calP(\eps)^{(\zetacoeff)} \subset \calP(\eps)$, resp.) such that $\|q-r\|_\coeff \leq \zetacoeff$ ($\|q' - r'\|_\coeff \leq \zetacoeff$, resp.).
\end{claim}

\begin{proof}
We establish the claim for $\calP(\eps)^{(\zetacoeff)}$; the claim for $\calP$ is similar but simpler.
  
We begin by arguing that the set of polynomials $\calP(\eps)^{(\zetacoeff)}$ is well-defined and can be constructed algorithmically.
First, it is clear that $\calP_1$ is well-defined and can be constructed in a straightforward way by enumerating over all polynomials $u'$ as described in Step~3 (there are $(O(1/\tau))^{d\Upsilon}$ many such polynomials).
Turning to step~(a), given an $r_1 \in \calP_1$ (note that such an $r_1$ must have $\|r\|_\coeff=1$) we can check whether $r_1 \in \calP(\eps)$ by simply checking whether the sum of squares of its $T$ largest-magnitude coefficients is at least $1-\eps^2$.  For step (b), it is easy to construct a rounded polynomial $r'_1$ with rational coefficients satisfying $\|r'_1 - r_1\|_\coeff \leq \zetacoeff/20$ as desired.  Let us write the polynomial $r'_1(x_1,\dots,x_{d\Upsilon})$ as
\[
r'_1(x) = \sum_{S \in {[d\Upsilon] \choose \leq d}} \widehat{r'_1}(S) x_S.
\]
We have that any polynomial $q_1(x_1,\dots,x_{d\Upsilon}) \in \calP(\eps)$ that satisfies $\|r'_1 - q_1\|_\coeff \leq \zetacoeff/4$ is defined by the following conditions on its coefficients:

\begin{itemize}

\item $\sum_{S \in {[d\Upsilon] \choose \leq d}} \widehat{q_1}(S)^2 = 1$;

\item $\sum_{S \in {\cal S}} \widehat{q_1}(S)^2 < 1-\eps$ for every ${\cal S} \subseteq {[d \Upsilon] \choose \leq d}$ with $|{\cal S}|=T$;

\item $\sum_{S \in {[d\Upsilon] \choose \leq d}} (\widehat{q_1}(S) - \widehat{r'_1}(S))^2 \leq (\zetacoeff/4)^2$.

\end{itemize}

\noindent Since all of the coefficients of $r'_1$ are rational numbers and we can assume without loss of generality that $(\zetacoeff/4)^2$ is rational, by \Cref{thm:first-order} (see \Cref{ap:setup}) we can algorithmically determine whether or not such a $q_1$ defined by $(\widehat{q_1}(S))_{S \in {[d \Upsilon] \choose \leq d}}$ exists. If such a $q_1$ exists, we can find a polynomial $q'_1 \in \calP(\eps)$ such that $\|r'_1-q'_1\|_\coeff \leq \zetacoeff/2$ by searching over polynomials of the form $u_2(x)/\|u_2(x)\|_\coeff$ where $u_2(x_1,\dots,x_{d\Upsilon})$ has all integer coefficients of magnitude first at most 1, then at most 2, then at most 3, and so on; this search will succeed because a suitable rounding of the coefficients of $q_1$ will give the desired $q'_1$.
So the set of polynomials $\calP(\eps)^{\zetacoeff}$ defined in Step~3 is indeed well-defined and can be constructed algorithmically.
We further remark that it is clear from construction that $\calP(\eps)^{\zetacoeff} \subset \calP(\eps)$ as claimed.

Fix any polynomial $q'\in\calP(\eps)$ with sparsity $s' \leq \Upsilon$, and let $\wh{q'}(S_1),\dots,\wh{q'}(S_{s'})$ be  the nonzero coefficients of $q'$, so $q'(x) = \sum_{i=1}^{s'} \wh{q'}(S_i) x_{S_i}.$
We proceed to argue the existence of an $r' \in \calP(\eps)^{(\zetacoeff)}$ as claimed in \Cref{claim:net-of-polys}.
We do this as follows:  First, below we will argue that some polynomial $z$ in the set $\calP_1$ has $\|z-q'\|_\coeff \leq \zetacoeff/5$.  We claim that this suffices,  for the following reason:  on one hand, if the polynomial $z$ belongs to $\calP(\eps)$ then it can serve as the required element $r' \in \calP(\eps)^{(\zetacoeff)}.$ On the other hand, if this polynomial $z$ does not belong to $\calP(\eps)$, then by step~3(b) there is a polynomial $z'$ with rational coefficients that is $\zetacoeff/20$-close to $z$, and since $z$ is $\zetacoeff/5$-close to $q'$, by the triangle inequality $\|z' - q'\|_\coeff \leq \zetacoeff/5 + \zetacoeff/20 = \zetacoeff/4.$  
So indeed there exists a polynomial $q_1 \in \calP(\eps)$ that is $\zetacoeff/4$-close to $z'$ (because $q'$ itself is such a polynomial), and hence the polynomial $q'_1$ that is $\zetacoeff/2$-close to $z'$ must, again by the triangle inequality, satisfy $\|q' - q'_1\|_\coeff \leq \|z' - q'_1\|_\coeff + \|q' - z'\|_\coeff \leq \zetacoeff/2 + \zetacoeff/4 < \zetacoeff$; so $q'_1$ can serve as the required element $r' \in \calP(\eps)^{(\zetacoeff)}.$

We argue that some polynomial $z$ in the set $\calP_1$ has $\|z-q'\|_\coeff \leq \zetacoeff/5$ as follows.
Let $u_1(x)$ be the polynomial obtained from $q'(x)$ by rounding each coefficient $\wh{q'}(S)$ to the nearest integer multiple of $\tau'$, and let $z(x)$ be $u_1(x)/\|u_1\|_\coeff$ as described in Step~3.
Since each coefficient of $u_1$ differs from the corresponding coefficient of $q'$ by at most $\tau'/2$, we have that
\[
\|q'-u_1\|_\coeff^2 \leq s'(\tau'/2)^2 \leq \tau'^2 \Upsilon/4,
\quad\text{so}\quad
\|q'-u_1\|_\coeff \leq \tau' \sqrt{\Upsilon}/2.
\]
An easy calculation shows that 
\[
1 - \tau' \sqrt{\Upsilon} \leq \|u_1\|_\coeff^2 \leq 1 + \tau' \sqrt{\Upsilon} + \Upsilon\tau'^2/4
\leq 1 + 2 \tau' \sqrt{\Upsilon},
\]
(where the last inequality uses $\tau' \leq 4/\sqrt{\Upsilon}$), so the distance incurred by rescaling is
\[
\|u_1-z\|_\coeff \leq 2 \tau' \sqrt{\Upsilon},
\]
and hence
\[
\|q' - z\|_\coeff \leq (5/2)\tau' \sqrt{\Upsilon}
\leq \zetacoeff \quad \quad \text{(by the choice of $\tau'$)},
\]
as required.
    \end{proof}

With \Cref{claim:net-of-polys} in hand it is a simple matter to prove \Cref{lem:poly-space-eps-net}:

\begin{proofof}{\Cref{lem:poly-space-eps-net}}
We give the argument for $\RV(\calP)(\eps)^{(\zetaMom)}$; the argument for $\RV^{(\zetaMom)}(\calP)$ is entirely similar.
Fix any polynomial $q' \in \calP(\eps)$; by \Cref{claim:net-of-polys}, there is a polynomial $r' \in \calP(\eps)^{(\zetacoeff)}$ such that $\|q'-r'\|_\coeff \leq \zetacoeff$.
By the setting of $\zetacoeff$ in Step~1 of the algorithm and \Cref{lem:close-coeff-moments}, 
we have that $\Mom_k(q'(\bX^{\otimes d\Upsilon}),r'(\bX^{\otimes d\Upsilon})) \leq \zetaMom$.
Since $r'(\bX^{\otimes d\Upsilon}) \in \RV(\calP(\eps)^{(\zetacoeff)})$ by the definition of $\RV(\calP(\eps)^{(\zetacoeff)})$ in Step~4, the lemma is proved.
\end{proofof}

\begin{algorithm}[H]

\addtolength\linewidth{-2em}
\vspace{0.5em}
\textbf{Description:} The algorithm outputs a finite $\zetacoeff$-net $\calP^{(\zetacoeff)}$ of $\calP$ (in terms of coefficient distance) and a finite $\zetacoeff$-net $\calP(\eps)^{(\zetacoeff)}$ of $\calP(\eps)$ (in terms of coefficient distance): for each polynomial $q\in\calP$ ($q' \in \calP(\eps)$, resp.), there exists $r \in \calP^{(\zetacoeff)}$ ($r' \in \calP(\eps)^{(\zetacoeff)}$, resp.) such that 
$\|q-r\|_\coeff \leq \zetacoeff$  ($\|q' - r'\|_\coeff \leq \zetacoeff$, resp.).
\\[0.25em]
\textbf{Input:} Description of random variable $\bX$; degree bound $d \geq 1$; sparsity parameters $s,\newLs \geq 1$; tolerance $\eps$; granularity parameter $\zetacoeff$ for the nets.\\ 
[0.25em]\textbf{Output:} Two finite sets of polynomials, $\calP^{(\zetacoeff)} \subset \calP$ and $\calP(\eps)^{(\zetacoeff)} \subseteq \calP(\eps)$.
\
\vspace{0.1in}
\\\CENP$(\bX,d,s,\newLs,\eps,\zetacoeff)$:\\
\vspace{0.1in}

\begin{enumerate}

\item 
Set 
$\tau = {\frac {\zetacoeff}{3\sqrt{s}}}$.  
Set 
$\tau' = {\frac {2\zetacoeff}{25\sqrt{\Upsilon}}}.$

\item 
Let $\calP^{(\zetacoeff)}$ be the set of those polynomials $r(x_1,\dots,x_{ds})$ such that $r(x)=u(x)/\|u(x)\|_\coeff$, where $u(x_1,\dots,x_{ds})$ ranges over all polynomials whose 
coefficients are of the form $i \tau$ for some integer $i$ of magnitude at most $2/\tau.$

\item 
Let $\calP_1$ be the set of those polynomials $r_1(x_1,\dots,x_{d\Upsilon})$ such that 
$r_1(x)=u_1(x)/\|u_1(x)\|_\coeff$, where $u_1(x_1,\dots,x_{d\Upsilon})$ ranges over all polynomials whose coefficients are of the form $i \tau'$ for some integer $i$ of magnitude at most $2/\tau'.$
Let $\calP(\eps)^{(\zetacoeff)}$ be the subset of $\calP(\eps)$ that is constructed as follows:  For each $r_1 \in \calP_1$,

\begin{itemize}
\item [(a)] If $r_1 \in \calP(\eps)$ then include $r_1$ in $\calP(\eps)^{(\zetacoeff)}$. Otherwise, if $r_1 \notin \calP(\eps),$
\item [(b)] Let $r'_1$ be a polynomial, obtained from $r_1$ by rounding each coefficient of $r_1$ to a rational number, which satisfies $\|r'_1 - r_1\|_\coeff \leq \zetacoeff/20.$  If there exists a polynomial $q_1 \in \calP(\eps)$ such that $\|r'_1-q_1\|_\coeff \leq \zetacoeff/4$, then include a polynomial $q'_1 \in \calP(\eps)$ such that $\|r'_1 - q'_1\|_\coeff \leq \zetacoeff/2$ in 
$\calP(\eps)^{(\zetacoeff)}.$
\end{itemize}
    
\item Return $\calP^{(\zetacoeff)}$ and $\calP(\eps)^{(\zetacoeff)}$.
\end{enumerate}
\caption{The \CENP\ algorithm.}
\label{alg:construct-eps-net-poly}

\end{algorithm}

\begin{algorithm}[H]
\addtolength\linewidth{-2em}
\vspace{0.5em}
\textbf{Description:} The algorithm outputs two finite $\zetaMom$-nets, denoted $\RV^{(\zetaMom)}(\calP)$ and  $\RV(\calP)(\eps)^{(\zetaMom)}$, for the spaces of random variables $\RV(\calP)$ and $\RV(\calP(\eps))$ respectively: for each polynomial $q\in\calP$ ($q' \in \calP(\eps)$, resp.), there exists $\bY \in \RV^{(\zetaMom)}(\calP)$ ($\bY' \in \RV(\calP)(\eps)^{(\zetaMom)}$, resp.) such that 
$\Mom_k(r(\bX^{\otimes ds}),\bY) \leq \zetaMom$ ($\Mom_k(r(\bX^{\otimes ds}),\bY) \leq \zetaMom,$ resp.).
\\[0.25em]
\textbf{Input:} Description of random variable $\bX$; degree bound $d \geq 1$; sparsity parameters $s,\newLs \geq 1$; tolerance $\eps$; granularity parameter $0 < \zetaMom < 1$ for the net; order parameter $k$ for the moment-distance.\\ 
[0.25em]\textbf{Output:} Two finite sets of random variables, $\RV^{(\zetaMom)}(\calP)$ and  $\RV(\calP)(\eps)^{(\zetaMom)}$.
\
\vspace{0.1in}
\\\CEND$(\bX,d,s,\newLs,\eps,\zetaMom,k)$:\\
\vspace{0.1in}

\begin{enumerate}
\item Parameter setting:  Set $\zetacoeff=\frac{\zetaMom}{\xi_{k,d,\bX}}.$


\item Run \CENP$(\bX,d,s,T,\eps,\zetacoeff)$ and let $\calP^{(\zetacoeff)},\calP(\eps)^{(\zetacoeff)}$ be the two sets of polynomials that it returns.

\item The net for $\RV(\calP)$:
Let $\RV^{(\zetaMom)}(\calP) := \RV(\calP^{(\zetacoeff)}).$

\item The net for $\RV(\calP(\eps))$:
Let $\RV^{(\zetaMom)}(\calP(\eps)) := \RV(\calP(\eps)^{(\zetacoeff)}).$

\item Return $\RV^{(\zetaMom)}(\calP)$ and $\RV^{(\zetaMom)}(\calP(\eps))$.
\end{enumerate}
\caption{The \CEND\ algorithm.}
\label{alg:construct-eps-net-dist}
\end{algorithm}

\subsection{Estimating Wasserstein distance between $\calP$ and $\calP(\eps)$}

\begin{lemma}
    The {\tt Estimate-Wasserstein} algorithm (\Cref{alg:estimate-wasserstein}) takes in as input $\bX$, degree bound $d,$ sparsity parameters $s,\newLs,$ and tolerance $\eps>0$ and outputs $c$ such that $c_\eps\leq c\leq 2c_\eps$, where as in \Cref{eq:infyinf} we have $c_\eps=\inf_{\bY\in \RV(\calP),\bY'\in \RV(\calP(\eps))} W_1(\bY,\bY')$. 
\end{lemma}
\begin{proof}
We first briefly comment on how Step~1(c) of the algorithm is carried out: in that step each of the random variables $\bY_1,\bY_1'$ is a real-valued random variable that takes on only finitely many distinct values.  Thus, given a pair $\bY_1,\bY_1'$ we can compute the Wasserstein distance $W_1(\bY_1,\bY_1')$ exactly using the well-known optimal coupling for real random variables, which is equivalent to the formula
\[
W_1(\bY_1,\bY_1') = \int_0^1 |F^{-1}(t) - G^{-1}(t)| dt
\]
where $F,G$ are the cdf's corresponding to $\bY_1$ and $\bY_1'$ respectively (see e.g.,~Equation~2.48 of \cite{villani2003optimal}).

Next, we observe that each time the algorithm computes a value of $c$ in Step~1(d), this value satisfies $c \geq c_\eps$. This is easy to see, since
    
    \begin{align*}
    c=&\min\{W_1(\bY_1,\bY'_1) :  \bY_1 \in \RV(\calP^{(\zetacoeff)}), \bY_1' \in \RV(\calP(\eps)^{(\zetacoeff)})\}\\
    \geq&\inf\{W_1(\bY,\bY') :  \bY \in \RV(\calP), \bY' \in \RV(\calP(\eps))\} = c_\eps,
    \end{align*}
recalling \Cref{eq:infyinf} and the fact that $\RV(\calP^{(\zetacoeff)})\subset \RV(\calP)$ and $\RV(\calP(\eps)^{(\zetacoeff)})\subset \RV(\calP(\eps))$ from \Cref{claim:net-of-polys}.

It remains to show that the output of  {\tt Estimate-Wasserstein} satisfies $c\leq 2c_\eps$, i.e.~to show that $c\leq 2\cdot W_1(\bY,\bY')$ for each pair of random variables $\bY \in \RV(\calP), \bY' \in \RV(\calP(\eps))$. So fix any pair of random variables $\bY =q(\bX^{\otimes n}) \in \RV(\calP), \bY'=q'(\bX^{\otimes n}) \in \RV(\calP(\eps))$, and consider the execution of the loop on which Step~1(d) outputs $c$.
By \Cref{claim:net-of-polys}, there exist random variables $\bY_1 =q_1(\bX^{\otimes n}) \in \RV(\calP^{\zetacoeff})$ and $\bY'_1 = q'_1(\bX^{\otimes n}) \in \RV(\calP(\eps)^{\zetacoeff})$ such that $\|q-q_1\|_\coeff \leq \zetacoeff$ and 
$\|q'-q'_1\|_\coeff \leq 
\zetacoeff$.
Recalling \Cref{eq:wasserstein-coeff}, we moreover have that $W_1(\bY_1,\bY) \leq \zetacoeff$ and $W_1(\bY'_1,\bY') \leq \zetacoeff$.
Hence by the triangle inequality for Wasserstein distance, we have
    \begin{align*}
        W_1(\bY,\bY') &\geq W_1(\bY_1,\bY_1')-W_1(\bY_1,\bY)-W_1(\bY'_1,\bY')\\
&\geq c -2\zetacoeff \geq c/2,
    \end{align*}
where we used that $c \leq W_1(\bY_1,\bY'_1)$ and $\zetacoeff \leq c/4$ in the loop when the algorithm terminates.


    Finally, \Cref{alg:estimate-wasserstein} terminates because at some point in the loop a value of $t$ will be reached which is large enough so that $4 \cdot 2^{-t} \leq c_\eps \leq c$; when this occurs the algorithm will halt, as claimed.  
\end{proof}

\begin{algorithm}
\addtolength\linewidth{-2em}
\vspace{0.5em}
\textbf{Description:} Given a parameter $\eps>0$, the algorithm gives a factor-2  estimate of  $c_\eps:=\inf_{\bY\in \RV(\calP),\bY'\in \RV(\calP(\eps))} W_1(\bY,\bY')$.
\\[0.25em]
\textbf{Input:} Description of random variable $\bX$; degree bound $d \geq 1$; sparsity parameters $s,\newLs \geq 1$; tolerance $\eps$.
\\ [0.25em]\textbf{Output:} An estimate $c$ such that $c_\eps\leq c\leq 2c_\eps$.
\
\vspace{0.2in}
\\\EW$(\bX,d,s,\newLs,\eps)$:

\begin{enumerate}
    \item For each integer $t=1,2,\dots$ do the following:
     
    \begin{enumerate}
    

    \item Set $\zetacoeff=2^{-t}$.
    
    \item Run \CENP$(\bX,d,s,\newLs,\eps,\zetacoeff)$ 
    to obtain two finite nets of polynomials $\calP^{(\zetacoeff)}$ and $\calP(\eps)^{(\zetacoeff)}$. 
    
    \label{algo-step:1}
    \item Let $S$ be the set of all values of $W_1(\bY_1,\bY_1')$ where $\bY_1$ ranges over all random variables in $\RV(\calP^{(\zetacoeff)})$ and $\bY_1'$ ranges over all random variables in $\RV(\calP(\eps)^{(\zetacoeff)})$.
    
    \item Let $c=\min(S)$. If 
    $4 \cdot \zetacoeff \leq c$ then output $c$; otherwise, 
    go to the next iteration of the loop.
    
    \end{enumerate}
    
\end{enumerate}
\caption{The \EW\ algorithm.}
\label{alg:estimate-wasserstein}
\end{algorithm}

\subsection{The \TS\  algorithm and the proof of \Cref{thm:sharp-upper}}
\Cref{alg:test-sparsity} gives the \TS\ algorithm.  With it, we are ready to prove \Cref{thm:sharp-upper}. 

\begin{proofof}{\Cref{thm:sharp-upper}}
    We will show that the \TS\ algorithm solves the $(\bX,d,\boldeta,s,\newLs,\eps)$ polynomial sparsity testing problem and uses $O_{\bX,d,\boldeta,s,\newLs,\eps}(1)$ samples.
     We first consider the easy case that $\newLs \geq \Upsilon.$  This case is handled by Phase~0:  by \Cref{thm:DFKO-informal}, if $p$ is $s$-sparse then \TS\ outputs ``\SSparse'' with probability at least $9/10$, and if $p$ is $\eps$-far from $\newLs$-sparse then it is also $\eps$-far from $\Upsilon$-sparse and hence \TS\ outputs ``\FFLS'' with probability at least $9/10.$ So in the rest of the argument we assume that $\newLs < \Upsilon.$  
     
     We first establish completeness and then soundness below.

\medskip
\noindent {\bf Case 1 (completeness): $\sparsity(p)\leq s$.} In this case we must show that with probability at least 2/3, the algorithm outputs \SSparse.  Since $\newLs < \Upsilon$,
by \Cref{thm:DFKO-informal}, the \TS\ algorithm continues past Phase~2 with probability at least 9/10. By a union bound over all $k$ calls to \EMWN$_{\bX,\ion}$ in Step~5, with probability at least $99/100$ we obtain an estimated moment vector $\tilde{m}(p)$ such that $\|\tilde{m}(p)-m(p)\|_2 \leq \zetaMom/10$, where $m(p)=(m_1(p(\bX^{\otimes n})),\dots,m_k(p(\bX^{\otimes n})))$ is the true moment vector of $p(\bX^{\otimes n})$.  
Moreover, by \Cref{lem:poly-space-eps-net}, there exists a random variable
    $\bY\in \RV(\calP)^{(\zetaMom/10)}$ such that $\Mom(p(\bX^{\otimes n}),\bY)\leq \zetaMom/10$. 
    Therefore, by the triangle inequality and a union bound, with overall probability at least $89/100$ there exists a $\bY\in \RV(\calP)^{(\zetaMom/10)}$ such that $\|\tilde{m}(p)-m(\bY)\|_2 <2 \zetaMom/10 < \zetaMom/2$, which means that \TS\ outputs ``\SSparse'' as desired.

\medskip
\noindent {\bf Case 2 (soundness):  $p$ is $\eps$-far from $\newLs$-sparse}. In this case we must show that with probability at least 2/3 the algorithm outputs ``\FFLS.'' As discussed in \Cref{sec:poly-sparsity-ub-overview}, if $p$ is $\eps'$-far from $\Upsilon$-sparse then with probability at least $9/10$ the \TS\ algorithm outputs ``\FFLS'' and halts in Phase~2 as desired.   Thus, we may suppose that $p$ is $\eps$-far from $\newLs$-sparse and $\eps'$-close to $\Upsilon$-sparse, i.e.~$p \in \calP(\eps,\eps')$. 
 As in Case~1, by a union bound over all $k$ calls to \EMWN$_{\bX,\ion}$ in Step~5, with probability at least $99/100$ we obtain an estimated moment vector $\tilde{m}(p)$ such that $\|\tilde{m}(p)-m(p)\|_2 \leq \zetaMom/10$. Moreover, for each random variable $\bY\in \RV(\calP)^{(\zetaMom/10)},$ there exists some $q\in\calP$ such that $\bY=q(\bX^{\otimes ds})$. 
Since $c_{\eps/2}/2 \leq c \leq c_{\eps/2}$, 
by our choices of $k$ and $\eps'$
we may apply \Cref{lem:moment-difference} to conclude that $\Mom_k(q(\bX^{\otimes ds}),p(\bX^{\otimes n})) \geq {\frac {c_{\eps/2}}{4C' 3^k}}\geq {\frac {c}{4C' 3^k}} = \zetaMom.$
So by the triangle inequality, with probability at least $99/100$ we have that  $\|\tilde{m}(p) - m(\bY)\|_2 \geq 9\zetaMom/10$  for each $\bY \in \RV({\cal P})^{(\zetaMom/10)}$, and hence the \TS\ algorithm outputs ``\FFLS.'' 

    It remains only to verify that the algorithm uses $O_{\bX,d,\boldeta,s,\newLs,\eps}(1)$ samples. The algorithm uses samples $(\bx_i,p(\bx_i)+\ion_i)$ only in Step~5, to estimate various moments of $p(\bX^{\otimes n})$ to additive error $\zetaMom/(10\sqrt{k})$ with failure probability $1/(100k)$,  using the sub-routine \EMWN.  Tracing through parameters it is straightforward to check that all the relevant parameters depend only on $\bX,d,\boldeta,s,\newLs,\eps$, so using the sample complexity bound provided by \Cref{lem:estimating-clean-moments}, the proof is complete.
\end{proofof}

\begin{algorithm}[H]
\addtolength\linewidth{-2em}
\vspace{0.5em}
\textbf{Input:} Description of random variable $\bX$; 
degree bound $d \geq 1$; sparsity parameters $s,\newLs \geq 1$; tolerance $\eps>0$.\\ [0.25em]
\textbf{Output:} ``\SSparse''  or ``\FFLS''

\

\TS$(\bX,
d,s,\newLs,\eps)$:\\

\medskip

\noindent {\bf Phase 0: Handle the case that $\newLs \geq \Upsilon.$}

\begin{itemize}

\item [0.] If $\newLs \geq \Upsilon$ then run the {\tt CoarseTest}$(s,\eps)$ algorithm from \Cref{sec:DFKO-algorithm} and halt, outputting whatever it outputs. Otherwise (if $\newLs < \Upsilon$) continue.

\end{itemize}

\noindent {\bf Phase 1:  Set parameters.}

\begin{itemize}

\item [1.]
Run \EW$(\bX,d,s,\newLs,\eps/2)$ and let $c$ be its output divided by 2.

\item [2.] Set parameters:  
Set $L := M^d \sqrt{d\Upsilon}$ as in \Cref{lem:moment-difference}, set $k := \lceil 2LC/c \rceil$ where $C$ is the constant from \Cref{lem:original-kong-valiant}, set $\eps' := {\frac {c}{16C' 3^k \xi_{k,d,\bX}}}$ where $C'$ is the universal constant from \Cref{lem:original-kong-valiant} and $\xi_{k,d,\bX}$ is from \Cref{lem:close-coeff-moments}, and set $\zetaMom := {\frac {c}{4C' 3^k}}.$ 

\end{itemize}

\noindent {\bf Phase 2: Preliminary check to distinguish $s$-sparse from far-from-$\Upsilon$-sparse.}

\begin{itemize}

	\item [3.] Run the {\tt CoarseTest}$(s,\eps')$ algorithm from \Cref{sec:DFKO-algorithm}. If it outputs ``$\eps$-far from $\Upsilon$-sparse'' then output ``\FFLS'' and halt, otherwise continue.
\end{itemize}

\noindent {\bf Phase 3:  Build a net of random variables.}

\begin{itemize}
    \item [4.] Run \CEND$(\bX,d,s,\newLs,\eps,\zetaMom/10,k)$ and let $\RV(\calP)^{(\zetaMom/10)}$ be the net for $\RV(\calP)$ that it returns.  
    
    \end{itemize}
    
\noindent {\bf Phase 4:  Search the net for a random variable with small moment-distance from $p(\bX^{\otimes n}).$}

\begin{itemize}

    \item [5.] For $\ell=1,\dots,k$, run the \EMWN$_{\bX,\ion}$ algorithm, with failure probability $1/(100 k)$ and additive error parameter $\zetaMom/(10\sqrt{k})$ at each execution, to obtain an estimate $\tilde{m}_\ell(p)$ of the $\ell^{\text{th}}$ moment $m_\ell(p(\bX^{\otimes n}))$ of $p(\bX^{\otimes n})$. Define the $k$-dimensional vector of estimated moments $\tilde{m}(p) $ to be $(\tilde{m}_1(p(\bX^{\otimes n})),\dots,\tilde{m}_k(p(\bX^{\otimes n}))).$
    \item [6.] For each $\bY \in \RV(\calP)^{(\zetaMom/10)}$, let $m(\bY)$ denote the $k$-dimensional vector of moments $m(\bY) = (m_1(\bY),\dots,m_k(\bY))$.  If there exists a $\bY\in \RV(\calP)^{(\zetaMom/10)}$ such that $\|\tilde{m}(p)-m(\bY)\|_2 <\zetaMom/2$, then output ``\SSparse,'' and otherwise output ``\FFLS.''
\end{itemize}

\caption{The \TS\ algorithm.}
\label{alg:test-sparsity}
\end{algorithm}

\subsection{Reduction to the case when $\|p\|_\coeff=1$}
\label{sec:reduction}

Recall that the conditions of 
\Cref{thm:sharp-upper}
state that $\frac{1}{K} \le \|p\|_\coeff \le K$, but throughout this section, we  assumed that $\|p\|_\coeff=1.$  In this subsection we explain why the assumption that $\|p\|_\coeff=1$ that we have made is without loss of generality.

The main idea is that we will get a high accuracy estimate of $\| p \|_\coeff$ (call it $\widetilde{p}_\coeff$) and then divide the label by $\widetilde{p}_\coeff$ -- this will be tantamount to getting labeled samples from a polynomial with $\| \cdot \|_\coeff=1$. We now elaborate on this idea. 

We start by observing that 
\[
\|p\|_\coeff^2 = \mathbf{E}[p^2(\bX^{\otimes n})].
\]
Thus, using \Cref{lem:estimating-clean-moments},  we can estimate $\|p\|_\coeff$ to accuracy $\pm \tau_0$, 
with a sample complexity scaling as $O(1) \cdot \poly(1/\tau_0)$.
Note that the implicit constant in the $O(1)$ depends on $\bX$, $K$, $d$ and  $\boldeta$, but not on the ambient dimension $n$. Thus, at this point, we have 
 an estimate $\widetilde{m}_2(p)$ such that 
\[
\abs{\widetilde{m}_2(p) - \|p\|_\coeff^2} \le \tau_0. 
\]
If we now define $\widetilde{p}_\coeff := \sqrt{\tilde{m}_2(p)}$, then it follows that 
\[
\abs{\widetilde{p}_\coeff - \|p \|_\coeff}  \le  \frac{\tau_0}{\| p \|_\coeff}.  
\]
We will assume $\tau_0$ is sufficiently small and thus, $\widetilde{p}_\coeff \le \| p \|_\coeff + \tau_0/ \|p\|_\coeff \le 2K$.

Now, we will run our algorithm {\tt Test-Sparsity} on the samples $(\bx, \widetilde{\by})$ where $ \widetilde{\by}= \by/\widetilde{p}_\coeff$,  the noise distribution is $\boldeta' = \boldeta/\widetilde{p}_\coeff$ and the error parameter $\tilde{\epsilon} := \epsilon/(2K)$.  Note that the  algorithm is equivalently getting samples from $(\bx, \tilde{q}(\bx) + \boldeta')$ where $\tilde{q} (x)  = p(x)/\widetilde{p}_\coeff$. Further note that $\sparsity(\tilde{q}) = \sparsity(p)$ and if $p$ is $\epsilon$-far from $T$-sparse, then $\tilde{q}$ is $\tilde{\epsilon} =\epsilon/\widetilde{p}_\coeff \ge \epsilon/(2K)$-far from $T$-sparse. We would almost be done by the guarantees of the algorithm {\tt Test-Sparsity}, except for one issue: the quantity $\| \tilde{q} \|_\coeff$ is not exactly $1$, rather, it is bounded by $\| \tilde{q} \|_\coeff \in [1 - \tau_0/\|p \|_\coeff, 1+\tau_0/\|p \|_\coeff]$. Given that $\| p \|_\coeff \ge 1/K$, it follows that $\| \tilde{q} \|_\coeff \in [1 - \tau_0 \cdot K, 1+\tau_0\cdot K]$.  

To see how this affects the guarantees of {\tt Test-Sparsity}, consider the idealized situation where the algorithm was getting samples of the form $(\bx, q(\bx) + \boldeta')$ where $q(x) = p(x)/\| p \|_\coeff$. 
Note that $\| q \|_\coeff=1$ and similar to $\tilde{q}$, if $p$ is $s$-sparse, so is $q$. On the other hand, if $p$ is $\epsilon$-far from $T$-sparse, $q$ is $\epsilon/(2K)$-far from $T$-sparse. We will now aim to show that the behavior of the algorithm {\tt Test-Sparsity} when the input samples are of the form $(\bx, q(\bx) + \boldeta')$ versus $(\bx, \tilde{q}(\bx) + \boldeta')$ is essentially the same (for a sufficiently small choice of $\tau_0$).



\begin{enumerate}
    \item The {\tt Test-Sparsity} algorithm  given in \Cref{alg:test-sparsity} interacts with the samples in a very limited mannner --- namely, among Phases 0 through 4 of the algorithm, only Phases $0$, $3$ and $4$ use the samples. Phases 0 and 3 run the {\tt CoarseTest} algorithm whose analysis does not assume that the polynomial's coefficient-norm is 1, so it suffices to consider Phase~4.
    \item In Phase~4, the only routine that is invoked is {\tt Estimate-Moments-With-Noise} (from \Cref{lem:estimating-clean-moments}). The algorithm {\tt Estimate-Moments-With-Noise} naturally comes with an error bound -- i.e., if we run this algorithm on the samples 
    $(\bx, \tilde{q}(\bx) + \boldeta')$, it will output the $\ell^\text{th}$ moment of $\tilde{q}(\bx)$ to some target error $\pm \tau$. The crucial thing to observe now is that for any $\ell$, the $\ell^\text{th}$ moment of $\tilde{q}(\bx)$ versus the $\ell^\text{th}$ moment of ${q}(\bx)$ differs by a multiplicative factor of at most $(1+K\tau_0 )^{\ell}$.  Let $L_0$ be the maximum moment computed in any invocation {\tt Estimate-Moments-With-Noise}. Thus, if we set $\tau_0 \ll \frac{1}{K \cdot L_0 \cdot \ln (1/\zeta)} $, then 
 any moment of $\tilde{q}(\bx)$ differs from the corresponding moment of ${q}(\bx)$ 
 only by a multiplicative factor of $(1+\zeta)$. 
 Thus, if we set $\tau_0$ to be sufficiently small (but still independent of $n$), the error incurred because the algorithm is getting samples from $(\bx,\tilde{q}(\bx) + \boldeta')$ (as  opposed to the ``ideal distribution"  $(\bx,{q}(\bx) + \boldeta')$) is no more than $\tau$. Thus, instead of the algorithm {\tt Estimate-Moments-With-Noise}  incurring error $\tau$, it will now incur error $2\tau$.   As we can set $\tau$ to be sufficiently small, any algorithm which uses {\tt Estimate-Moments-With-Noise} will continue to have the same performance guarantees when it gets samples from $(\bx, \tilde{q}(\bx) + \boldeta')$.
\end{enumerate}

Thus, in summary, we can assume that the algorithm {\tt Test-Sparsity} gets samples from $(\bx, p(\bx) + \boldeta)$ where $\| p \|_\coeff=1$ and the error resulting from this assumption can be absorbed in the error of the algorithm {\tt Estimate-Moments-With-Noise}.


%% file: sections/appendix-boundsonMSG.tex

\newcommand{\W}{\mathrm{W}}

\section{Some bounds on 
$\MSG_{\bX,d}$ for particular $\bX$s}\label{ap:examples}

We give upper and lower bounds on  $\MSG_{\Unif(\{-1,1\}),d}(1)$:
\begin{observation} \label{obs:upper-lower-uniform}
    For $\bX=\Unif(\{-1,1\})$ and any $d \geq 1$, we have $4^{d-1} \leq \MSG_{\bX,d}(1)\leq 2^{O(d^2)}$.
\end{observation} 
\begin{proof}
    Up to scaling by a non-zero constant and renaming of variables, a multilinear polynomial $p$ of sparsity $s=1$ and degree $r \leq d$ must be $p_r(x)=x_1 \cdots x_r$. Without loss of generality, assume that the coefficient of the single term in the 1-sparse polynomial is 1; then $p(\bX^{\otimes n})$ is distributed as $\Unif(\{-1,1\}).$ 

    To prove the lower bound, it suffices to give an example of a degree-$d$ polynomial $p$ such that $p(\bX^{\otimes n})$ is distributed identically to $\bX$ (i.e.~uniform over $\bits$).
    Such a polynomial is given, for example, by a complete decision tree of depth $d$, with $2^d$ leaf bits, $2^{d}-1$ internal nodes, and $d$ variables on each path, in which every variable occurs at only one node and leaf bits alternate  between 1 and $-1$. It is clear that $p(\bX^{\otimes n})$ is balanced, and an easy induction shows that the sparsity of $p$ is exactly $4^{d-1}.$

    For the upper bound, let $p$ be any degree-$d$ polynomial such that $p(\bX^{\otimes n})$ is distributed identically to $\bX$. Such a polynomial $p$ computes a $\bits$-valued function over $\bits^n$.  As reported in \cite{CHS20}, Wellens has shown that any  degree-$d$ polynomial computing a $\bits$-valued function over $\bits^n$ must depend on at most $4.416 \cdot 2^d$ variables, and hence it must be a ${4.416 \cdot 2^d \choose \leq d} = 2^{O(d^2)}$-sparse polynomial. 
\end{proof}

An easy corollary of \Cref{obs:upper-lower-uniform} is that $\MSG_{\Unif(\{-1,1\}),d}(s)$ is at least $4^{d-1}s$:  this is achieved by having $p(x_1,\dots,x_s) = x_1 + \cdots + x_s$ and having $q(x)$ be the $(s(2^d - 1))$-variable function which is the sum of $s$ variable-disjoint copies of the decision tree witnessing the $4^{d-1}$ lower bound of \Cref{obs:upper-lower-uniform}.

We can give an exact value for 
$\MSG_{\Unif(\{-1,1\}),2}(1)$:
\begin{observation} \label{obs:quarry}
    For $\bX=\Unif(\{-1,1\})$ and $d=2$, we have $\MSG_{\bX,d}(1)=4$.
\end{observation} 
\begin{proof}
    The lower bound is given by \Cref{obs:upper-lower-uniform}, so we show that $\MSG_{\Unif(\{-1,1\}),2} \leq 4.$  Let $q(x_1,\dots,x_n)$ be a degree-2 multilinear polynomial such that $q(\bX^{\otimes n})$ is distributed as $\Unif(\{-1,1\})$. The general form of $q$ is $$q(x)=\widehat{q}(\emptyset) + \sum_{i} \widehat{q}(i) x_i+\sum_{i< j} \widehat{q}(i,j)x_ix_j;$$
    since $\E[q(\bX^{\otimes n})]= \E[\Unif(\{-1,1\}^n)]=0,$ we have $\widehat{q}(\emptyset)=0.$
    We further note that  $\Ex[q(\bX^{\otimes n})^2]^{1/2}=\|q\|_\coeff=1$, so by Parseval's identity we have that $\sum_i \widehat{q}(i)^2 + \sum_{i < j} \widehat{q}(i,j)^2 = 1.$


    If $\widehat{q}(i,j)=0$ for all $i<j$ then it is easy to see that $q=p_1$ (up to renaming of variables), so let $i<j$ be such that the coefficient $\widehat{q}(i,j)$ is non-zero. Notice that for any $x \in \bits^n$ we have that $|q(x)-q(x^{\oplus i})-q(x^{\oplus j})+q(x^{\oplus\{i,j\}})|=4|\widehat{q}(i,j)|\in\{0,2,4\}$, as $q(\bX^{\otimes n})$ takes values in $\{-1,1\}$. If $\widehat{q}(i,j)=1$ then $q=p_2$ and we are done, so we may suppose that all the quadratic coefficients $\widehat{q}(i,j)$ have magnitude $1/2.$
    
    We further observe that for any $i$, we have $|q(x)-q(x^{\oplus i})|=2|\widehat{q}(i)\pm \sum_{\ell < i}\widehat{q}(\ell,i) \pm \sum_{r > i} \widehat{q}(i,r) |\in\{0,2\}$. 
    Since from above we know that all nonzero $A_{\ell,i}$ and $A_{i,r}$ values must have magnitude 1/2, this implies that  $|\widehat{q}(i)|\in\{0,\frac{1}{2}\}$. 
    
    Thus, we have argued that all the non-zero coefficients of $q$ must have magnitude $\frac{1}{2}$. Together with the fact that $\|q\|_\coeff=1$, $q$ needs to have exactly 4 non-zero coefficients. Thus, we have that $\MSG_{\Unif(\{-1,1\}),2}(1)=4$ and the proof is complete.
\end{proof}

%% file: sections/appendix-estimating-moments.tex
\section{Proof of \Cref{lem:estimating-clean-moments}: Estimating moments in the presence of noise}
\label{ap:estimating-moments}

\subsection{Preliminary tools}



We will use the following result from \cite{CDS20stoc}; for completeness we provide the simple proof.
\begin{fact} [Fact~A.1 of \cite{CDS20stoc}] \label{fact:properties_moments}
For any real random variable $\bX$, for any $n\geq k \geq 1$, we have
\begin{equation}\label{eq:monotone_root_moments}
\E[|\bX|^n]^{1/n} \ge \E[|\bX|^k]^{1/k}
\end{equation}
and
\begin{equation}\label{eq:monotone_moments}
\E[|\bX|^n] \ge \E[|\bX|^k] \E[|\bX|^{n-k}].
\end{equation}
\end{fact}

\begin{proof}
\Cref{eq:monotone_root_moments} follows by considering two random variables $\bA=\bX^{k}$, $\bB=1$, $p=n/k$ and $q=n/(n-k)$ and applying H\"{o}lder's inequality:
\[
\E[|\bA \bB|] \le \E[|\bA|^p]^{1/p} \E[|\bB|^q]^{1/q}.
\]
To prove \Cref{eq:monotone_moments} we apply \Cref{eq:monotone_root_moments}  twice with parameters $(n,k)$ and $(n,n-k)$, obtaining
\[
\E[|\bX|^n]^{k/n} \ge \E[|\bX|^k] \text{~and~} \E[|\bX|^n]^{(n-k)/n} \ge \E[|\bX|^{n-k}],
\]
which together imply $\E[|\bX|^n] \ge \E[|\bX|^k] \E[|\bX|^{n-k}]$.
\end{proof}

\subsection{Intermediate results: Proofs of \Cref{claim:bounds-noiseless-moments}, \Cref{lem:estimating-noiseless-moments}, \Cref{claim:bounds-noisy-moments}, \Cref{lem:estimating-noisy-moments-noisy-data}, and \Cref{lem:estimating-noisy-cumulants}}

Throughout \Cref{claim:bounds-noiseless-moments}, \Cref{lem:estimating-noiseless-moments}, \Cref{claim:bounds-noisy-moments}, \Cref{lem:estimating-noisy-moments-noisy-data}, and \Cref{lem:estimating-noisy-cumulants} the random variable $\bX$, the random variable $\ion$, and the polynomial $p$ are as described in \Cref{lem:estimating-clean-moments} (i.e.~$\E[\bX]=0,\Var[\bX]=1,$ and $\supp(\bX) \subseteq [-M,M]$; $\ion$ is a real-valued random variable with finite moments of all orders; and $p$ is a degree-at-most-$d$ multilinear polynomial  that is promised to have $\|p\|_\coeff \in [1/K,K]$).

\begin{claim}
[Bounds on noiseless moments of $p(\bX^{\otimes n})$]
\label{claim:bounds-noiseless-moments}
For any positive integer $\ell$,
we have that
\[
m_\ell(p(\bX^{\otimes n})) 
\leq m_\ell(|p(\bX^{\otimes n})|) 
\leq 
(\ell-1)^{d\ell/2}\cdot\lambda^{d-d\ell/2}\cdot K^\ell.
\]
\end{claim}
\begin{proof}
Recalling the hypercontractivity theorem, \Cref{thm:OD10.21}, we have the inequalities  
\begin{align*}
m_\ell(p(\bX^{\otimes n}))\leq m_\ell(|p(\bX^{\otimes n})|)\leq& (\sqrt{\ell-1}\cdot \lambda^{1/\ell-1/2})^{d\ell}m_2(p(\bX^{\otimes n}))^{\ell/2} \tag{\Cref{eq:10.21}}\\
=&(\sqrt{\ell-1}\cdot \lambda^{1/\ell-1/2})^{d\ell}\|p\|_\coeff^{\ell}\\
\leq&(\ell-1)^{d\ell/2}\cdot\lambda^{d-d\ell/2}\cdot K^\ell. \qedhere
\end{align*}
\end{proof}

\begin{lemma} 
[Estimating noiseless moments of $p(\bX^{\otimes n})$ given noiseless data]
\label{lem:estimating-noiseless-moments}
For any $\eps,\delta>0$ and any positive integer $\ell$, 
let $\bz^{(i)},\dots,\bz^{(m)}$ be obtained as $\bz^{(i)} = p(\bx^{(i)})$ where the $\bx^{(i)}$'s are i.i.d.~according to $\bX^{\otimes n}$ and  $m=O((2\ell-1)^{d\ell}\cdot\lambda^{d-d\ell}\cdot K^{2\ell}/(\delta\eps^2)).$ 
Then with probability at least $1-\delta$, we have that
\begin{equation} \label{eq:acc1}
\abs{{\frac {\sum_{i=1}^m (\bz^{(i)})^\ell}{m}} - m_\ell(p(\bX^{\otimes n}))} \leq \eps.
\end{equation}
\end{lemma}
\begin{proof}
First notice that $\Ex_{\bx^{(i)}}[(\bz^{(i)})^\ell]=\Ex[p(\bX^{\otimes n})^\ell]=m_\ell(p(\bX^{\otimes n}))$. To show the claim, we bound the variance
\[
\Var\left[(\bz^{(i)})^\ell\right]\leq\Ex\left[(\bz^{(i)})^{2\ell}\right]=m_{2\ell}(p(\bX^{\otimes n}))\leq 
(2\ell-1)^{d\ell}\cdot\lambda^{d-d\ell}\cdot K^{2\ell}
\]
where the last inequality follows from \Cref{claim:bounds-noiseless-moments}. Therefore,
\[
\Var\left[\frac{\sum_{i=1}^m(\bz^{(i)})^\ell}{m}\right]\leq
(2\ell-1)^{d\ell}\cdot\lambda^{d-d\ell}\cdot K^{2\ell}.
\]
As $m=O((2\ell-1)^{d\ell}\cdot\lambda^{d-d\ell}\cdot K^{2\ell}/(\delta\eps^2))$, the claim then follows by Chebyshev's inequality.
\end{proof}

\begin{claim}
[Bounds on noisy moments of $p(\bX^{\otimes n})+ \ion$]
\label{claim:bounds-noisy-moments}
For any positive integer $\ell$,
we have that
\[
|m_\ell(p(\bX^{\otimes n})+\ion)| \leq 2^\ell\cdot\left(\ell^{d\ell/2}\cdot \lambda^{d(1-\ell)/2}\cdot K^\ell+m_{\maxl{\ell}{\ell+1}}(\ion)+1\right)
.
\]
\end{claim}
\begin{proof}
Expanding the $\ell^{th}$ moment, we have
\begin{align*}
    |m_\ell(p(\bX^{\otimes n})+\ion)|
    =&\left|\Ex\left[\left(p(\bX^{\otimes n})+\ion\right)^\ell\right]\right|\\
    =&\left|\sum_{i=0}^\ell {\ell\choose i}\Ex\left[\left(p(\bX^{\otimes n})\right)^i\right]\Ex\left[\ion^{\ell-i}\right]\right|\\
    \leq&\sum_{i=0}^\ell {\ell\choose i}\Ex\left[\left|p(\bX^{\otimes n})\right|^i\right]\Ex\left[|\ion|^{\ell-i}\right].
\end{align*}
In the case when $\ell$ is even, by H\"older's inequality, the above quantity is upper bounded by
\begin{align*}
    \sum_{i=0}^\ell {\ell\choose i}\Ex\left[\left|p(\bX^{\otimes n})\right|^i\right]\Ex\left[|\ion|^{\ell-i}\right]
    \leq&\sum_{i=0}^\ell {\ell\choose i}\Ex\left[\left(p(\bX^{\otimes n})\right)^\ell\right]^{i/\ell}\Ex\left[\ion^{\ell}\right]^{1-i/\ell}\\
    \leq&\sum_{i=0}^\ell {\ell\choose i}\cdot\max\{m_\ell(p(\bX^{\otimes n})),m_\ell(\ion)\}\\
    \leq&2^\ell\cdot\left(m_\ell(p(\bX^{\otimes n}))+m_\ell(\ion)\right)\\
    \leq&2^\ell\cdot\left((\ell-1)^{d\ell/2}\cdot\lambda^{d-d\ell/2}\cdot K^\ell+m_\ell(\ion)\right)\\
    \leq &2^\ell\cdot\left(\ell^{d\ell/2}\cdot\lambda^{d-d\ell/2}\cdot K^\ell+m_\ell(\ion)\right)
\end{align*}
where the second to the last inequality follows from \Cref{claim:bounds-noiseless-moments}.

In the case when $\ell$ is odd, again by H\"older's inequality, we obtain the upper bound  
\begin{align*}
    &\sum_{i=0}^\ell {\ell\choose i}\Ex\left[\left|p(\bX^{\otimes n})\right|^i\right]\Ex\left[|\ion|^{\ell-i}\right]\\
    \leq&\sum_{i=0}^\ell {\ell\choose i}\Ex\left[\left(p(\bX^{\otimes n})\right)^{\ell+1}\right]^{i/(\ell+1)}\Ex\left[\ion^{\ell+1}\right]^{1-(i+1)/(\ell+1)}\\
    \leq&\sum_{i=0}^\ell {\ell\choose i}\cdot\max\{m_{\ell+1}(p(\bX^{\otimes n})),m_{\ell+1}(\ion)\}^{\ell/(\ell+1)}\\
    \leq&2^\ell\cdot\left(m_{\ell+1}(p(\bX^{\otimes n}))^{\ell/(\ell+1)}+m_{\ell+1}(\ion)^{\ell/(\ell+1)}\right)\\
    \leq&2^\ell\cdot\left((\ell^{d\ell/2+d/2}\cdot \lambda^{d/2-d\ell/2}\cdot K^{\ell+1})^{\ell/(\ell+1)}+m_{\ell+1}(\ion)^{\ell/(\ell+1)}\right)\\
    =&2^\ell\cdot\left(\ell^{d\ell/2}\cdot \lambda^{d(\ell-\ell^2)/(2\ell+2)}\cdot K^\ell+m_{\ell+1}(\ion)^{\ell/(\ell+1)}\right)\\
\leq&2^\ell\cdot\left(\ell^{d\ell/2}\cdot \lambda^{d(1-\ell)/2}\cdot K^\ell+m_{\ell+1}(\ion)+1\right).
\end{align*}
Combining the two cases, we must have
\[
|m_\ell(p(\bX^{\otimes n})+\ion)|\leq 
2^\ell\cdot\left(\ell^{d\ell/2}\cdot \lambda^{d(1-\ell)/2}\cdot K^\ell+m_{\maxl{\ell}{\ell+1}}(\ion)+1\right).
\qedhere
\]
\end{proof}

\begin{lemma} 
[Estimating noisy moments of $p(\bX^{\otimes n}) + \ion$ given noisy data]
\label{lem:estimating-noisy-moments-noisy-data}
For any $\eps,\delta>0$ and any positive integer $\ell$, 
let $\by^{(i)},\dots,\by^{(m)}$ be obtained as $\by^{(i)} = p(\bx^{(i)})+\ion$ where the $\bx^{(i)}$'s are i.i.d.~according to $\bX^{\otimes n}$, each $\ion$ is independent of everything else, and 
\[
m = O(\frac{1}{\delta \eps^2}\cdot2^{2\ell}\cdot((2\ell)^{d\ell}\cdot \lambda^{d(1-2\ell)/2}\cdot K^{2\ell}+m_{\maxl{2\ell}{2\ell+2}}(\ion)+1)).
\]
Then with probability at least $1-\delta$, we have that
\begin{equation} \label{eq:acc2}
\abs{{\frac {\sum_{i=1}^m (\by^{(i)})^\ell}{m}} - m_\ell(p(\bX^{\otimes n})+\ion)} \leq \eps.
\end{equation}
\end{lemma}
\begin{proof}
The proof is analogous to the proof of \Cref{lem:estimating-noiseless-moments}. Notice that the expectation
\[
\Ex\left[\frac{\sum_{i=1}^m(\by^{(i)})^\ell}{m}\right]=m_{\ell}(p(\bX^{\otimes n})+\ion)
\] 
and we can bound the variance with
\begin{align*}
\Var\left[\frac{\sum_{i=1}^m(\by^{(i)})^\ell}{m}\right]=&\frac{\Var\left[(\by^{(1)})^{\ell}\right]}{m}\leq\frac{\Ex\left[(\by^{(1)})^{2\ell}\right]}{m}\\
\leq&\frac{1}{m}\cdot2^{2\ell}\cdot((2\ell)^{d\ell}\cdot \lambda^{d(1-2\ell)/2}\cdot K^{2\ell}+m_{\maxl{2\ell}{2\ell+2}}(\ion)+1)\\
\leq&O(\delta \eps^2)
\end{align*}
for $m=O(\frac{1}{\delta \eps^2}\cdot2^{2\ell}\cdot((2\ell)^{d\ell}\cdot \lambda^{d(1-2\ell)/2}\cdot K^{2\ell}+m_{\maxl{2\ell}{2\ell+2}}(\ion)+1))$. The claim then follows from Chebyshev's inequality.
\end{proof}

\begin{lemma}
[Estimating noisy cumulants of $p(\bX^{\otimes n}) + \ion$ given noisy data]
\label{lem:estimating-noisy-cumulants}
There is an algorithm with the following property:  
For any $\delta>0$, $1>\eps>0$ and any positive integer $\ell$, 
the algorithm takes as input $
m=\poly(\ell^{d\ell^2},K^{\ell^2},1/(\delta\eps),1/\lambda^{d\ell^2},m_{\maxl{1}{\ell+1}}(\ion)^\ell)$ many independent random samples of $\by^{(i)} = p(\bx^{(i)})+\ion$ where the $\bx^{(i)}$'s are i.i.d.~according to $\bX^{\otimes n}$ and each $\ion$ is independent of everything else.  The algorithm outputs an estimate $\tilde{\kappa}_\ell(\by)$ of $\kappa_{\ell}(\by)$ that with probability at least $1-\delta$ satisfies
\begin{equation} \label{eq:acc3}
\abs{\tilde{\kappa}_\ell(\by) - \kappa_\ell(\by)} \leq \eps.
\end{equation}
\end{lemma}
\begin{proof}
First we apply \Cref{lem:estimating-noisy-moments-noisy-data} to obtain estimates $\tilde{m}_k(\by):=\sum_{i=1}^m(\by^{(i)})^k/m$ for each $1\leq k\leq \ell$ such that with probability at least $1-\delta/\ell$, we have
\[
|\tilde{m}_k(\by)-m_k(\by)|\leq \eps',
\]
for some $1>\eps'>0$ we will define later. Our estimates of $\kappa_\ell(\by)$ is as follows:
\[
\tilde{\kappa}_\ell(\by):=\sum_{k=1}^\ell(-1)^{k-1}(k-1)!B_{\ell,k}\left(\tilde{m}_1(\by),\tilde{m}_2(\by),...,\tilde{m}_{\ell-k+1}(\by)\right)
\]
where $B_{\ell,k}$ are incomplete Bell polynomials
\[
B_{\ell,k}\left(\tilde{m}_1(\by),\tilde{m}_2(\by),...,\tilde{m}_{\ell-k+1}(\by)\right)=\sum\frac{\ell!}{j_1!\cdots j_{\ell-k+1}!}\left(\frac{\tilde{m}_1(\by)}{1}\right)^{j_1}\cdots\left(\frac{\tilde{m}_{\ell-k+1}(\by)}{(\ell-k+1)!}\right)^{j_{\ell-k+1}}
\]
and the sum is taken over all non-negative tuples $(j_1,...,j_{{\ell-k+1}})$ satisfying
\[
j_1+\cdots+j_{{\ell-k+1}}=k\text{ and }j_1+2j_2+\cdots+(\ell-k+1)j_{{\ell-k+1}}=\ell.
\]
Fix such a tuple $(j_1,...,j_{{\ell-k+1}})$; we give a bound for $\left(\frac{\tilde{m}_1(\by)}{1}\right)^{j_1}\cdots\left(\frac{\tilde{m}_{\ell-k+1}(\by)}{(\ell-k+1)!}\right)^{j_{\ell-k+1}}$. Notice that since $|\tilde{m}_k(\by)-m_k(\by)|\leq \eps'$ for all $k\in[\ell],$ there exists a set of sign assignments $\sigma_k=\pm1$ for all $k\in[\ell]$ such that 
\[
\left(\tilde{m}_1(\by)\right)^{j_1}\cdots\left(\tilde{m}_{\ell-k+1}(\by)\right)^{j_{\ell-k+1}}
\leq \left(m_1(\by)+\sigma_1\eps'\right)^{j_1}\cdots\left(m_{\ell-k+1}(\by)+\sigma_{\ell-k+1}\eps'\right)^{j_{\ell-k+1}}.
\]
Let $M_L=\max_{\alpha\in[\ell]}\{|m_\alpha(\by)|,1\}$. We expand and further bound the above expression by 
\begin{align*}
    &\left(m_1(\by)^{j_1}+\sum_{i=0}^{j_1-1}{j_1\choose i}m_1(\by)^i(\sigma_1\eps')^{j_1-i}\right)\\
    &\cdots\left(m_{\ell-k+1}(\by)^{j_{\ell-k+1}}+\sum_{i=0}^{j_{\ell-k+1}-1}{j_{\ell-k+1}\choose i}m_{\ell-k+1}(\by)^i(\sigma_{\ell-k+1}\eps')^{j_{\ell-k+1}-i}\right)\\
    =&m_1(\by)^{j_1}\cdots m_{\ell-k+1}(\by)^{j_{\ell-k+1}}+\sum_{\substack{b\in\{0,1\}^{\ell-k+1}\\b\neq\vec{1}}}\prod_{\alpha=1}^{\ell-k+1}\left(m_\alpha(\by)^{j_{\alpha}}\right)^{b_{\alpha}}\left(\sum_{i=0}^{j_\alpha-1}{j_\alpha\choose i}m_\alpha(\by)^i(\sigma_\alpha\eps')^{j_\alpha-i}\right)^{1-b_\alpha}\tag{expanding the product}\\
    \leq&m_1(\by)^{j_1}\cdots m_{\ell-k+1}(\by)^{j_{\ell-k+1}}+\sum_{\substack{b\in\{0,1\}^{\ell-k+1}\\b\neq\vec{1}}}\prod_{\alpha=1}^{\ell-k+1}\left|m_\alpha(\by)\right|^{j_{\alpha}\cdot b_{\alpha}}\left(\sum_{i=0}^{j_\alpha-1}{j_\alpha\choose i}|m_\alpha(\by)|^i(\eps')^{j_\alpha-i}\right)^{1-b_\alpha}\\
    \leq&m_1(\by)^{j_1}\cdots m_{\ell-k+1}(\by)^{j_{\ell-k+1}}+\sum_{\substack{b\in\{0,1\}^{\ell-k+1}\\b\neq\vec{1}}}\prod_{\alpha=1}^{\ell-k+1}(M_L)^{j_{\alpha}\cdot b_{\alpha}}\left(\sum_{i=0}^{j_\alpha-1}{j_\alpha\choose i}(M_L)^i(\eps')^{j_\alpha-i}\right)^{1-b_\alpha}\tag{follows from the definition of $M_L$}\\
    \leq&m_1(\by)^{j_1}\cdots m_{\ell-k+1}(\by)^{j_{\ell-k+1}}+\eps'\sum_{\substack{b\in\{0,1\}^{\ell-k+1}\\b\neq\vec{1}}}\prod_{\alpha=1}^{\ell-k+1}(M_L)^{j_{\alpha}\cdot b_{\alpha}}\left(\sum_{i=0}^{j_\alpha-1}{j_\alpha\choose i}(M_L)^i\right)^{1-b_\alpha}\tag{since each term in the summation contains a positive power of $\eps'$ and $\eps'<1$}\\
    \leq&m_1(\by)^{j_1}\cdots m_{\ell-k+1}(\by)^{j_{\ell-k+1}}+\eps'\sum_{\substack{b\in\{0,1\}^{\ell-k+1}\\b\neq\vec{1}}}\prod_{\alpha=1}^{\ell-k+1}(M_L)^{j_{\alpha}\cdot b_{\alpha}}\left(1+M_L\right)^{j_\alpha(1-b_\alpha)}\tag{binomial expansion}\\
    \leq&m_1(\by)^{j_1}\cdots m_{\ell-k+1}(\by)^{j_{\ell-k+1}}+\eps'\sum_{\substack{b\in\{0,1\}^{\ell-k+1}\\b\neq\vec{1}}}\prod_{\alpha=1}^{\ell-k+1}(M_L)^{j_{\alpha}}(2M_L)^{j_\alpha}\tag{since $b_\alpha\in\{0,1\}$ and $M_L\geq 1$}\\
    \leq&m_1(\by)^{j_1}\cdots m_{\ell-k+1}(\by)^{j_{\ell-k+1}}+\eps'\cdot 2^{\ell-k+1}\prod_{\alpha=1}^{\ell-k+1}(M_L)^{2j_{\alpha}}2^{j_\alpha}\\
    =&m_1(\by)^{j_1}\cdots m_{\ell-k+1}(\by)^{j_{\ell+1}}+\eps'\cdot 2^{\ell-k+1}(M_L)^{2k}.\\
\end{align*}
By a similar argument, we derive a lower bound for the estimate 
\[
\left(\tilde{m}_1(\by)\right)^{j_1}\cdots\left(\tilde{m}_{\ell-k+1}(\by)\right)^{j_{\ell-k+1}}\geq m_1(\by)^{j_1}\cdots m_{\ell-k+1}(\by)^{j_{\ell+1}}-\eps'\cdot 2^{\ell-k+1}(M_L)^{2k}.
\]
Therefore, 
\begin{align*}
    |\tilde{\kappa}_\ell(\by)-\kappa_\ell(\by)|=&\left|\sum_{k=1}^\ell(-1)^{k-1}(k-1)!\left(B_{\ell,k}(\tilde{m}_1(\by),...,\tilde{m}_{\ell-k+1}(\by)-B_{\ell,k}({m}_1(\by),...,{m}_{\ell-k+1}(\by))\right)\right|\\
    =&\left|\sum_{k=1}^\ell(-1)^{k-1}(k-1)!\sum\frac{\ell!}{j_1!\cdots j_{\ell-k+1}!}\right.\\
    &\left.\left(\left(\frac{\tilde{m}_1(\by)}{1}\right)^{j_1}\cdots\left(\frac{\tilde{m}_{\ell-k+1}(\by)}{(\ell-k+1)!}\right)^{j_{\ell-k+1}}-\left(\frac{{m}_1(\by)}{1}\right)^{j_1}\cdots\left(\frac{{m}_{\ell-k+1}(\by)}{(\ell-k+1)!}\right)^{j_{\ell-k+1}}\right)\right|\tag{by definition of Bell polynomials}\\
    \leq&\sum_{k=1}^\ell(k-1)!\sum\frac{\ell!}{j_1!\cdots j_{\ell-k+1}!}\\
    &\left|\left(\frac{\tilde{m}_1(\by)}{1}\right)^{j_1}\cdots\left(\frac{\tilde{m}_{\ell-k+1}(\by)}{(\ell-k+1)!}\right)^{j_{\ell-k+1}}-\left(\frac{{m}_1(\by)}{1}\right)^{j_1}\cdots\left(\frac{{m}_{\ell-k+1}(\by)}{(\ell-k+1)!}\right)^{j_{\ell-k+1}}\right|\tag{triangle inequality}\\
    \leq&\sum_{k=1}^\ell(k-1)!\sum\frac{\ell!}{j_1!\cdots j_{\ell-k+1}!}\\
    &\left|\left(\tilde{m}_1(\by)\right)^{j_1}\cdots\left(\tilde{m}_{\ell-k+1}(\by)\right)^{j_{\ell-k+1}}-\left({m}_1(\by)\right)^{j_1}\cdots\left({m}_{\ell-k+1}(\by)\right)^{j_{\ell-k+1}}\right|\\
    \leq&\sum_{k=1}^\ell(k-1)!\sum\frac{\ell!}{j_1!\cdots j_{\ell-k+1}!}\cdot\eps'\cdot 2^{\ell-k+1}(M_L)^{2k}\tag{by the previous argument}\\
    \leq&\sum_{k=1}^\ell(k-1)!\cdot2^\ell\cdot\ell!\cdot\eps'\cdot 2^{\ell-k+1}(M_L)^{2k}\tag{bound the number of possible sequences $j_1,...,j_{\ell-k+1}$ by $2^\ell$}\\
    \leq&\ell!\cdot2^\ell\cdot\ell!\cdot\eps'\cdot 2^{\ell}(M_L)^{2\ell}\\
    =&\eps'\cdot(\ell!)^2\cdot 2^{2\ell}\cdot (M_L)^{2\ell}.
\end{align*}
By \Cref{claim:bounds-noisy-moments}, we have that 
\[
M_L=\max_{\alpha\in[\ell]}\{|m_\alpha(\by)|,1\}\leq 
2^\ell\cdot\left(\ell^{d\ell/2}\cdot \lambda^{d(1-\ell)/2}\cdot K^\ell+m_{\maxl{1}{\ell+1}}(\ion)+1\right).
\]
Therefore, $|\tilde{\kappa}_\ell(\by)-\kappa_\ell(\by)|$ is bounded above by $\eps$ when taking 
\[
\eps'=
\frac{\eps}{(\ell!)^2\cdot  2^{2\ell^2+2\ell}\cdot\left(2^\ell\cdot\left(\ell^{d\ell/2}\cdot \lambda^{d(1-\ell)/2}\cdot K^\ell+m_{\maxl{1}{\ell+1}}(\ion)+1\right)\right)^{2\ell}}.
\]
To achieve the $\eps'$ error, 
$m=\poly(\ell^{d\ell^2},K^{\ell^2},1/(\delta\eps),1/\lambda^{d\ell^2},m_{\maxl{1}{\ell+1}}(\ion)^\ell)$ samples suffices.
\end{proof}

\subsection{Proof of \Cref{lem:estimating-clean-moments}}
\label{sec:proof-estimating-noisy-moments}
\begin{proof}[Proof of \Cref{lem:estimating-clean-moments}]
Following \Cref{lem:estimating-noisy-cumulants}, let $m'$ be the number of samples such that the estimates $\tilde{\kappa}_i(\by)$ of $\kappa_i(\by)$ for every $1\leq i\leq \ell$ satisfy
\[
|\tilde{\kappa}_i(\by)-\kappa_i(\by)|\leq \eps'
\]
for some $\eps'$ we define later with probability at least $1-\delta/\ell$.
Since for independent random variables $\bX,\bY$, $\kappa_i(\bX+\bY)=\kappa_i(\bX)+\kappa_i(\bY)$, letting $\tilde{\kappa}_i(p(\bX^{\otimes n}))=\tilde{\kappa}_i(\by)-\kappa_i(\ion),$ we have
\[
|\tilde{\kappa}_i(p(\bX^{\otimes n}))-\kappa_i(p(\bX^{\otimes n}))|\leq \eps'.
\]
\Cref{eq:moments_cumulants} expresses $m_\ell(p(\bX^{\otimes n}))$ in terms of $\cum_i(p(\bX^{\otimes n}))$ for $i\in[\ell]$.
Naturally, let our estimate of $m_\ell(p(\bX^{\otimes n}))$ be
\[
\tilde{m}_\ell(p(\bX^{\otimes n}))=\sum_{k=1}^\ell B_{\ell,k}\bigg( \tilde{\cum}_1(p(\bX^{\otimes n})),\ldots,\tilde{\cum}_{\ell-k+1}(p(\bX^{\otimes n})) \bigg).
\]
Then we have the estimation error bound
\begin{align*}
    &|\tilde{m}_\ell(p(\bX^{\otimes n}))-m_\ell(p(\bX^{\otimes n}))|\\
    =&\left|\sum_{k=1}^\ell B_{\ell,k}\left(\bigg( \tilde{\cum}_1(p(\bX^{\otimes n})),\ldots,\tilde{\cum}_{\ell-k+1}(p(\bX^{\otimes n})) \bigg)-
    \bigg( {\cum}_1(p(\bX^{\otimes n})),\ldots,{\cum}_{\ell-k+1}(p(\bX^{\otimes n})) \bigg)\right)\right|\\
    =&\left|\sum_{k=1}^\ell\sum\frac{\ell!}{j_1!\cdots j_{\ell-k+1}!}\left(\frac{1}{1}\right)^{j_1}\cdots\left(\frac{1}{(\ell-k+1)!}\right)^{j_{\ell-k+1}}\right.\\
    &\left.
    \tilde{\cum}_1(p(\bX^{\otimes n}))^{j_1}\cdots \tilde{\cum}_{\ell-k+1}(p(\bX^{\otimes n}))^{j_{\ell-k+1}}-{\cum}_1(p(\bX^{\otimes n}))^{j_1}\cdots {\cum}_{\ell-k+1}(p(\bX^{\otimes n}))^{j_{\ell-k+1}}
    \right|\\
    \leq&\sum_{k=1}^\ell\sum\frac{\ell!}{j_1!\cdots j_{\ell-k+1}!}\left(\frac{1}{1}\right)^{j_1}\cdots\left(\frac{1}{(\ell-k+1)!}\right)^{j_{\ell-k+1}}\\
    &\left|
    \tilde{\cum}_1(p(\bX^{\otimes n}))^{j_1}\cdots \tilde{\cum}_{\ell-k+1}(p(\bX^{\otimes n}))^{j_{\ell-k+1}}-{\cum}_1(p(\bX^{\otimes n}))^{j_1}\cdots {\cum}_{\ell-k+1}(p(\bX^{\otimes n}))^{j_{\ell-k+1}}
    \right|
\end{align*}
where the second summation is taken over all tuples $(j_1,...,j_{\ell-k+1})$ of non-negative satisfying
\[
j_1+\cdots+j_{{\ell-k+1}}=k\text{ and }j_1+2j_2+\cdots+(\ell-k+1)j_{{\ell-k+1}}=\ell.
\]
Notice that since for all $i\in[\ell],$ the estimate $\tilde{\cum}_{i}(p(\bX^{\otimes n}))$ is within an additive $\eps'$ to ${\cum}_{i}(p(\bX^{\otimes n}))$, by an argument analogous to that when bounding $\tilde{m}_1(\by)^{j_1}\cdots\tilde{m}_{\ell-k+1}(\by)^{j_{\ell-k+1}}$ in terms of $m_1(\by)^{j_1}\cdots m_{\ell-k+1}(\by)^{j_{\ell+1}}$ in \Cref{lem:estimating-noisy-cumulants}, we obtain the bounds
\begin{align*}
&\tilde{\cum}_1(p(\bX^{\otimes n}))^{j_1}\cdots \tilde{\cum}_{\ell-k+1}(p(\bX^{\otimes n}))^{j_{\ell-k+1}}\\
&\in\left[{\cum}_1(p(\bX^{\otimes n}))^{j_1}\cdots {\cum}_{\ell-k+1}(p(\bX^{\otimes n}))^{j_{\ell-k+1}}\pm \eps'\cdot 2^{\ell-k+1}\calK^{2k}\right]
\end{align*}
where $\calK:=\max_{\alpha\in[\ell]}\{|\kappa_\alpha(p(\bX^{\otimes n}))|,1\}$. By \Cref{clm:up_bound_cumulant}, $|\kappa_\alpha(p(\bX^{\otimes n}))|\leq m_\ell(|p(\bX^{\otimes n})|)\cdot e^\ell\cdot \ell!.$ 
Together with \Cref{claim:bounds-noiseless-moments}, 
\[
\calK\leq 
\ell^{d\ell/2}\cdot \lambda^{d-d\ell/2}\cdot K^{\ell}\cdot e^\ell\cdot \ell!.
\]
Combining the above, we obtain
\begin{align*}
    &|\tilde{m}_\ell(p(\bX^{\otimes n}))-m_\ell(p(\bX^{\otimes n}))|\\
    \leq&\sum_{k=1}^\ell\sum\frac{\ell!}{j_1!\cdots j_{\ell-k+1}!}\left(\frac{1}{1}\right)^{j_1}\cdots\left(\frac{1}{(\ell-k+1)!}\right)^{j_{\ell-k+1}}\\
    &\left|
    \tilde{\cum}_1(p(\bX^{\otimes n}))^{j_1}\cdots \tilde{\cum}_{\ell-k+1}(p(\bX^{\otimes n}))^{j_{\ell-k+1}}-{\cum}_1(p(\bX^{\otimes n}))^{j_1}\cdots {\cum}_{\ell-k+1}(p(\bX^{\otimes n}))^{j_{\ell-k+1}}
    \right|\\
    \leq&\sum_{k=1}^\ell\sum\frac{\ell!}{j_1!\cdots j_{\ell-k+1}!}\left(\frac{1}{1}\right)^{j_1}\cdots\left(\frac{1}{(\ell-k+1)!}\right)^{j_{\ell-k+1}}\cdot\eps'\cdot 2^{\ell-k+1}\calK^{2k}\\
    \leq&\ell \cdot 2^\ell\cdot \ell!\cdot\eps'\cdot 2^{\ell}\cdot\calK^{2\ell}\\
    \leq&\eps'\cdot \ell^{d\ell^2}\cdot \lambda^{-d\ell^2}\cdot K^{2\ell^2}\cdot e^{2\ell^2}\cdot \ell^{3\ell^2}.
\end{align*}
Letting $\eps'=\frac{\tau}{\ell^{d\ell^2}\cdot \lambda^{-d\ell^2}\cdot K^{2\ell^2}\cdot e^{2\ell^2}\cdot \ell^{3\ell^2}}$ satisfies the error guarantee. Hence taking 
\[
m=\poly(\ell^{d\ell^2},K^{\ell^2},1/(\delta\tau),1/\lambda^{d\ell^2},m_{\maxl{1}{\ell+1}}(\ion)^\ell)
\]
suffices to produce the estimation guarantee.
\end{proof}




